\documentclass[11pt,a4paper]{article}
\usepackage{amssymb}
\usepackage{amsmath}
\usepackage{amsfonts}
\usepackage{bbm}
\usepackage{amsthm}
\usepackage{mathrsfs}
\usepackage{hyperref}
\usepackage{color}
\usepackage[margin=2.37cm]{geometry}
\usepackage[all,cmtip]{xy}
\usepackage[utf8]{inputenc}
\usepackage{graphicx}
\usepackage{varwidth}
\usepackage{comment}
\usepackage{enumitem}

\usepackage{upgreek}
\usepackage{rotating}

\usepackage{tikz}
\usetikzlibrary{shapes.geometric}

\definecolor{darkred}{rgb}{0.8,0.1,0.1}
\hypersetup{
     colorlinks=false,         
     linkcolor=darkred,
     citecolor=blue,
}

\theoremstyle{plain}
\newtheorem{theo}{Theorem}[section]

\newtheorem{propo}[theo]{Proposition}

\theoremstyle{definition}
\newtheorem{defi}[theo]{Definition}

\newtheorem{exx}[theo]{Example}
\newenvironment{ex}
  {\pushQED{\qed}\exx}
  {\popQED\endexx}

\newtheorem{remm}[theo]{Remark}
\newenvironment{rem}
  {\pushQED{\qed}\remm}
  {\popQED\endremm}

\numberwithin{equation}{section}

\def\nn{\nonumber}

\def\bbK{\mathbb{K}}
\def\bbR{\mathbb{R}}
\def\bbC{\mathbb{C}}

\def\bbZ{\mathbb{Z}}

\def\Hom{\mathrm{Hom}}
\def\hom{\underline{\mathrm{hom}}}

\def\map{\underline{\mathrm{map}}}

\def\id{\mathrm{id}}
\def\supp{\mathrm{supp}}

\def\dd{\mathrm{d}}
\def\vol{\mathrm{vol}}
\def\cl{\mathrm{cl}}
\def\sc{\mathrm{sc}}
\def\cc{\mathrm{c}}

\def\pc{\mathrm{pc}}
\def\fc{\mathrm{fc}}
\def\spc{\mathrm{spc}}
\def\sfc{\mathrm{sfc}}
\def\hh{\mathrm{h}}
\def\vv{\mathrm{v}}
\def\dR{\mathrm{dR}}
\def\CS{\mathrm{CS}}
\def\MW{\mathrm{MW}}
\def\dim{\mathrm{dim}}
\def\1{I}

\def\op{\mathrm{op}}
\def\ev{\mathrm{ev}}

\def\pr{\mathrm{pr}}
\def\cone{\operatorname{cone}}

\def\Loc{\mathbf{Loc}}
\def\LocC{\mathbf{LocC}}

\def\Set{\mathbf{Set}}

\def\Ch{\mathbf{Ch}}

\def\CC{\mathbf{C}}

\def\sSet{\mathbf{sSet}}
\def\sVec{\mathbf{sVec}}

\def\FFF{\mathfrak{F}}

\def\G{\mathcal{G}}
\def\H{\mathcal{H}}

\def\V{\mathcal{V}}
\def\W{\mathcal{W}}
\def\Z{\mathcal{Z}}

\def\scrD{\mathscr{D}}

\def\holim{\mathrm{holim}}
\def\colim{\mathrm{colim}}
\def\hocolim{\mathrm{hocolim}}

\def\asym{\mathsf{asym}}

\newcommand\und[1]{\underline{#1}}

\def\sk{\vspace{1mm}}

\makeatletter
\let\@fnsymbol\@alph
\makeatother

\title{%
Green hyperbolic complexes on Lorentzian manifolds
}

\author{%
Marco Benini$^{1,2,a}$, 
Giorgio Musante$^{1,b}$\ and\ 
Alexander Schenkel$^{3,c}$\vspace{4mm}\\
{\small ${}^1$ Dipartimento di Matematica, Universit\`a di Genova,}\\
{\small Via Dodecaneso 35, 16146 Genova, Italy.}\vspace{2mm}\\
{\small ${}^2$ INFN, Sezione di Genova,}\\
{\small Via Dodecaneso 33, 16146 Genova, Italy.}\vspace{2mm}\\
{\small ${}^3$ School of Mathematical Sciences, University of Nottingham,}\\
{\small University Park, Nottingham NG7 2RD, United Kingdom.}\vspace{4mm}\\
{\small \begin{tabular}{ll}
Email: & ${}^a$~\texttt{benini@dima.unige.it}\\
& ${}^b$~\texttt{musante@dima.unige.it}\\
& ${}^c$~\texttt{alexander.schenkel@nottingham.ac.uk}\vspace{2mm}
\end{tabular}
}
}

\date{July 2023}

\begin{document}

\maketitle

\begin{abstract}
\noindent We develop a homological generalization of Green hyperbolic operators, called Green hyperbolic complexes, which cover many examples of derived critical loci for gauge-theoretic quadratic action functionals in Lorentzian signature. We define Green hyperbolic complexes through a generalization of retarded and advanced Green's operators, called retarded and advanced Green's homotopies, which are shown to be unique up to a contractible space of choices. We prove homological generalizations of the most relevant features of Green hyperbolic operators, namely that (1)~the retarded-minus-advanced cochain map is a quasi-isomorphism, (2)~a differential pairing (generalizing the usual fiber-wise metric) on a Green hyperbolic complex leads to covariant and fixed-time Poisson structures and (3)~the retarded-minus-advanced cochain map is compatible with these Poisson structures up to homotopy. 
\end{abstract}



\renewcommand{\baselinestretch}{0.8}\normalsize
\tableofcontents
\renewcommand{\baselinestretch}{1.0}\normalsize



\section{\label{sec:intro}Introduction and summary}
Green hyperbolic operators \cite{Bar}, 
generalizing normally hyperbolic ones \cite{BGP}, 
are one of the cornerstones of mathematical field theory 
on (globally hyperbolic) Lorentzian manifolds $M$. 
By definition, these are linear differential operators 
$P: \Gamma(E) \to \Gamma(E)$ acting on sections of a vector bundle 
$E \to M$ that admit retarded and advanced Green's operators 
$G_\pm: \Gamma_\cc(E) \to \Gamma(E)$, i.e.\ linear maps 
satisfying $G_\pm\, P \varphi = \varphi = P\, G_\pm \varphi$ 
that propagate any compactly supported section 
$\varphi \in \Gamma_\cc(E)$ to the causal future/past 
$\supp(G_\pm \varphi) \subseteq J_M^\pm(\supp (\varphi))$ of its support. 
In particular, Green hyperbolic operators cannot vanish 
on non-zero compactly supported sections. This basic observation 
shows that the differential operator governing the dynamics 
of any gauge theory can {\it not} be Green hyperbolic as 
it is necessarily degenerate (it must vanish on gauge transformations). 
To resolve this incompatibility, the present paper develops 
the concept of Green hyperbolic complexes, a homological 
generalization of Green hyperbolic operators 
that encompasses the typical degeneracies of gauge theories.
\sk

To start with, we consider complexes of linear differential operators 
$(F,Q)$ on $M$, which consist of a $\bbZ$-graded vector bundle $F \to M$ 
and of a collection $Q$ of degree-increasing linear 
differential operators squaring to zero. From the perspective of 
the derived critical locus of a gauge-theoretic quadratic action 
functional, i.e.\ the Batalin-Vilkovisky formalism, 
the degree $0$ part of $F$ accommodates the (gauge) fields, 
the negative degrees accommodate the (higher) gauge transformations, 
also known as ghosts, and the positive degrees 
accommodate the anti-fields. The collection of degree-increasing 
differential operators $Q$ simultaneously encodes 
both the action of (higher) gauge transformations 
and the equation of motion of the (gauge) fields. 
\sk

In analogy with the definition of a Green hyperbolic operator, 
a {\it Green hyperbolic complex} is by Definition \ref{defi:Green-hyp-cpx} 
a complex of linear differential operators $(F,Q)$ 
that admits so-called {\it retarded and advanced 
Green's homotopies} $\Lambda_\pm$. The latter are homological 
generalizations of retarded and advanced Green's operators, 
which are formalized by $(-1)$-cochains 
$\Lambda_\pm \in \map(\FFF_{J_M^\pm},\FFF_{J_M^\pm})^{-1}$ 
such that $\delta \Lambda_\pm = \id$, 
where $\map$ denotes a suitable mapping complex 
(with differential $\delta$) defining a dg-enrichment 
on the category of cochain complex valued functors 
(see Section \ref{subsec:FunCh}), 
$\FFF_{J_M^\pm}$ is the functor that assigns the cochain complex 
$\FFF_{J_M^\pm(K)}$ of sections of $(F,Q)$ supported in the causal 
future/past of any compact subset $K \subseteq M$ and the $0$-cocycle 
$\id \in \map(\FFF_{J_M^\pm},\FFF_{J_M^\pm})^{0}$ 
is the identity natural transformation of the functor $\FFF_{J^\pm_M}$. 
Informally, one may think of the defining condition 
$\delta \Lambda_\pm = \id$ as the analog of the usual identities 
$G_\pm\, P \varphi = \varphi = P\, G_\pm \varphi$ involving the 
retarded and advanced Green's operators $G_\pm$ associated with a 
Green hyperbolic operator $P$ and of the functors $\FFF_{J^\pm_M}$ 
as a formalization of the future/past propagation of supports 
typical of $G_\pm$. 
Encoding the support conditions via the functors 
$\FFF_{J^\pm_M}$ and using the homologically well-behaved mapping complex $\map$ 
are the key ingredients that distinguish our approach 
from previous work on formal PDE theory and elliptic complexes, 
see e.g.\ \cite[Chs.~1 and 2]{Tarkhanov} for a textbook reference. 
\sk

Let us emphasize that Green hyperbolic complexes are 
a genuine generalization of Green hyperbolic operators 
in the following sense. Any linear differential operator 
$P$ acting on sections of a vector bundle $E$ defines 
a two-term complex of linear differential operators 
$(F_{(E,P)},Q_{(E,P)})$, which we take to be concentrated 
in degrees $0$ and $1$, see Example \ref{ex:single-diff-op}. 
It turns out that $P$ is a Green hyperbolic operator 
if and only if $(F_{(E,P)},Q_{(E,P)})$ is a Green hyperbolic complex. 
This follows from the observation that a retarded/advanced 
Green's homotopy for $(F_{(E,P)},Q_{(E,P)})$ is the same
as a retarded/advanced Green's operator for $P$, 
see Example \ref{ex:Lambdapm=Gpm}.
\sk

A key result in the theory of Green hyperbolic operators 
is the uniqueness of the associated retarded and advanced 
Green's operators. One may wonder if Green hyperbolic complexes 
behave similarly, namely the associated retarded and advanced 
Green's homotopies are unique (in the appropriate sense, see below). 
Proposition \ref{propo:Greenhyp-unique} provides a positive answer
to this question, which also clarifies the role of the mapping complex $\map$ 
entering the definition of retarded and advanced Green's homotopies. 
Crucially, $\map$ is a derived functor and hence compatible with weak equivalences, 
i.e.\ it sends natural quasi-isomorphisms to quasi-isomorphisms. 
A non-trivial, yet direct consequence of this feature 
is that the {\it spaces} of retarded and advanced Green's homotopies 
$\und{\G\H}_\pm$ from \eqref{eqn:GHpm} are either empty or contractible. 
Since at least one retarded and one advanced Green's homotopy exist 
by definition of a Green hyperbolic complex, it follows that 
$\und{\G\H}_\pm$ is contractible, which provides the correct 
homotopical formalization of uniqueness for retarded and advanced 
Green's homotopies. In contrast, note that uniqueness in the ordinary 
sense is not a relevant question in this context. For instance, adding 
any $(-1)$-coboundary $\delta \lambda_\pm$ from the mapping complex 
$\map(\FFF_{J_M^\pm},\FFF_{J_M^\pm})$ to a retarded/advanced 
Green's homotopy $\Lambda_\pm$ yields a new one 
$\Lambda_\pm + \delta \lambda_\pm$, however one would like to regard 
the latter as being equivalent to the former since they represent 
the same cohomology class. What would be a genuinely new 
retarded/advanced Green's homotopy is one that differs 
from $\Lambda_\pm$ by a $(-1)$-cocycle that represents 
a non-trivial cohomology class. Our contractibility result 
ensures in particular that such a $(-1)$-cocycle does not exist. 
\sk

The relevance of Green hyperbolic operators $P$ 
for mathematical field theory on globally hyperbolic Lorentzian manifolds \cite{BFV,FV} 
relies on the following key results \cite{BGP, Bar}. 
\begin{enumerate}[label=(\arabic*)]
\item The so-called retarded-minus-advanced propagator 
$G := G_+ - G_-$ descends to the isomorphism 
\begin{flalign}
G: \operatorname{coker}_\cc(P) \overset{\cong}{\longrightarrow} \ker_\sc(P)
\end{flalign}
that characterizes the vector space 
$\ker_\sc(P) := \ker(P: \Gamma_\sc(E) \to \Gamma_\sc(E))$ 
of solutions with spacelike compact support $_\sc$ 
in terms of the vector space 
$\operatorname{coker}_\cc(P) := \operatorname{coker}(P: \Gamma_\cc(E) \to \Gamma_\cc(E)) = \Gamma_\cc(E) \big/ P\Gamma_\cc(E)$ 
of sections with compact support $_\cc$ modulo equations of motion. 

\item When $P$ is formally self-adjoint with respect to a fiber metric 
$\langle-,-\rangle$ on $E$, we use the retarded-minus-advanced propagator 
$G$ or the retarded/advanced Green's operator $G_\pm$ to construct the (covariant) Poisson structure 
\begin{subequations}
\begin{flalign}
\tau_M: \operatorname{coker}_\cc(P)^{\otimes 2} \longrightarrow \bbR 
\end{flalign}
by any of the following four equivalent definitions: 
for all $[\varphi_1], [\varphi_2] \in \operatorname{coker}_\cc(P)$, 
\begin{flalign}
\tau_M([\varphi_1] \otimes [\varphi_2]) &:= \int_M \langle \varphi_1 \wedge \ast G \varphi_2 \rangle = - \int_M \langle \varphi_2 \wedge \ast G \varphi_1 \rangle \nn \\ 
&\phantom{:}= \pm \int_M \langle \varphi_1 \wedge \ast G_\pm \varphi_2 \rangle \mp \int_M \langle \varphi_2 \wedge \ast G_\pm \varphi_1 \rangle \quad,
\end{flalign}
\end{subequations}
where $\ast$ denotes the Hodge star operator fixed 
by the metric and orientation of $M$. 

\item When $P = \Box^\nabla + B$ is the sum of 
the d'Alembert operator of a metric connection $\nabla$ 
on $(E,\langle-,-\rangle)$ and of a symmetric endomorphism $B$, 
one obtains the (fixed-time) Poisson structure\footnote{In the literature this is sometimes called the (pre)symplectic structure on the spacelike compact solution space.} associated with a spacelike Cauchy surface $\Sigma \subseteq M$
\begin{subequations}
\begin{flalign}
\sigma_{\Sigma}: \ker_\sc(P)^{\otimes 2} \longrightarrow \bbR 
\end{flalign}
by defining, for all $\psi_1, \psi_2 \in \ker_\sc(P)$, 
\begin{flalign}
\sigma_{\Sigma}(\psi_1 \otimes \psi_2) := - \int_\Sigma \langle \psi_1 \wedge \ast \nabla \psi_2 - \psi_2 \wedge \ast \nabla \psi_1 \rangle \quad.
\end{flalign}
\end{subequations}
The isomorphism $G: \operatorname{coker}_\cc(P) \stackrel{\cong}{\longrightarrow} \ker_\sc(P) $ 
from item~(1) is compatible with the Poisson structures in the sense that
\begin{flalign}
\sigma_\Sigma \circ G^{\otimes 2} = \tau_M \quad.
\end{flalign}
Hence, it defines an isomorphism of Poisson vector spaces 
\begin{flalign}
G: \big(\operatorname{coker}_\cc(P),\tau_M\big) \overset{\cong}{\longrightarrow} 
 \big(\ker_\sc(P),\sigma_\Sigma\big)\quad.
\end{flalign}
\end{enumerate}

It turns out that our broader concept of Green hyperbolic complexes 
leads to very similar results. 
More explicitly, in this paper we prove the following. 
\begin{enumerate}[label=(\arabic*$^\prime$)]
\item To any choice of retarded and advanced Green's homotopies 
$\Lambda_\pm$ for $(F,Q)$ we associate the so-called 
{\it retarded-minus-advanced cochain map} $\Lambda_\hh$ 
and show in Theorem \ref{th:ret-minus-adv-qiso} 
that it provides a quasi-isomorphism 
\begin{flalign}
\Lambda_{\hh} := \Lambda_{+\, \hh} - \Lambda_{-\, \hh}: \FFF_{\hh\cc}[1] \overset{\sim}{\longrightarrow} \FFF_{\hh\sc} \quad, 
\end{flalign}
where $\Lambda_{\pm\, \hh} := \hocolim(\Lambda_{\pm})$, 
between the $1$-shift of the cochain complex $\FFF_{\hh\cc}$ of 
sections of $(F,Q)$ with compact support and the cochain complex 
$\FFF_{\hh\sc}$ of sections with spacelike compact support. 
Here the subscript $_{\hh}$ denotes that the support conditions 
are implemented via suitable homotopy colimits 
replacing the usual ordinary colimits. (As we explain below, 
this technical complication can be removed in the presence 
of a so-called Green's witness.) 

\item Based on the concept from Definition \ref{defi:diff-pairing} of a differential pairing $(-,-)$ on $(F,Q)$, 
which is a homological replacement of the fiber metric $\langle-,-\rangle$, 
and mimicking the multiple equivalent ways of 
writing the classical Poisson structure from item~(2), 
in Proposition \ref{propo:covariant-Poisson} we construct 
three (covariant) Poisson structures 
\begin{subequations}
\begin{flalign}
\tau_M^\pm,\; \tau_M := \asym(\widetilde{\tau}_M): \FFF_{\hh\cc}[1]^{\otimes 2} \longrightarrow \bbR
\end{flalign}
using either $\Lambda_{\pm\, \hh}$ 
or the retarded-minus-advanced cochain map $\Lambda_\hh$. 
Explicitly, $\tau_M^\pm$ is given, for all homogeneous 
$\varphi_1, \varphi_2 \in \FFF_{\hh\cc}[1]$, by
\begin{flalign}
\tau_M^\pm (\varphi_1 \otimes \varphi_2) := \pm \int_M (\varphi_1,\Lambda_{\pm\, \hh} \varphi_2) \mp (-1)^{|\varphi_1| |\varphi_2|} \int_M (\varphi_2,\Lambda_{\pm\, \hh} \varphi_1) \quad, 
\end{flalign}
and is manifestly graded anti-symmetric. 
Using also graded anti-symmetrization, the Poisson structure
$\tau_M$ is defined from the cochain map 
$\widetilde{\tau}_M: \FFF_{\hh\cc}[1]^{\otimes 2} \to \bbR$ given, 
for all homogeneous 
$\varphi_1, \varphi_2 \in \FFF_{\hh\cc}[1]$, by  
\begin{flalign}
\widetilde{\tau}_M (\varphi_1 \otimes \varphi_2) & := \int_M (\varphi_1,\Lambda_\hh \varphi_2) \quad. 
\end{flalign}
\end{subequations}
Graded anti-symmetrization is needed to define $\tau_M$ because 
the chosen retarded and advanced Green's homotopies $\Lambda_\pm$, 
which enter the definition of $\Lambda_\hh$, 
may not be compatible with the differential pairing $(-,-)$. 
For the same reason, the three Poisson structures 
$\tau_M^+$, $\tau_M^-$ and $ \tau_M$ in general do not coincide on the nose, 
however they are related by a homotopy 
$\lambda_M \in [\FFF_{\hh\cc}[1]^{\otimes 2},\bbR]^{-1}$ according to 
\begin{flalign}
\tau_M^\pm = \tau_M \pm \partial \lambda_M \quad. 
\end{flalign}
Graded anti-symmetrization and the homotopy $\lambda_M$ 
are complications that can be removed in the presence of a so-called 
formally self-adjoint Green's witness, as we explain below. 
Under this additional assumption, 
it turns out that the three Poisson structures coincide 
$\tau_M^+ = \tau_M^- = \tau_M$.

\item Upon the choice of a spacelike Cauchy surface 
$\Sigma \subseteq M$, the concept of a differential pairing $(-,-)$ 
leads also to the (fixed-time) Poisson structure 
\begin{subequations}
\begin{flalign}
\sigma_\Sigma: \FFF_{\hh\sc}^{\otimes 2} \longrightarrow \bbR
\end{flalign}
from Proposition \ref{propo:time-zero-Poisson}, which is defined, 
for all homogeneous $\psi_1, \psi_2 \in \FFF_{\hh\sc}$, by the formula 
\begin{flalign}
\sigma_\Sigma(\psi_1 \otimes \psi_2) := (-1)^{m-1} \int_{\Sigma} (\psi_1,\psi_2)\quad.
\end{flalign}
\end{subequations}
We prove in Theorem \ref{th:Linfty-qiso} 
that the retarded-minus-advanced cochain map 
$\Lambda_\hh$ is compatible with the Poisson structures 
\begin{flalign}
\sigma_\Sigma \circ \Lambda_\hh^{\otimes 2} = \tau_M + \partial \lambda_\hh
\end{flalign}
up to an explicitly constructed graded anti-symmetric homotopy 
$\lambda_\hh \in [\FFF_{\hh\cc}[1]^{\otimes 2},\bbR]^{-1}$. 
This defines an equivalence
\begin{flalign}
(\Lambda_\hh,\lambda_\hh) : \big(\FFF_{\hh\cc}[1],\tau_M\big)  \overset{\sim}{\longrightarrow}  \big(\FFF_{\hh\sc},\sigma_\Sigma\big) 
\end{flalign}
in the simplicial category of Poisson complexes set up in \cite[Sec.~3.1]{GwilliamHaugseng}.
\end{enumerate}
Let us emphasize that the new results from items~(1$^\prime$--3$^\prime$) 
in the theory of Green hyperbolic complexes recover precisely 
the corresponding results from items~(1--3) in the theory of 
Green hyperbolic operators, see e.g.\ Remark \ref{rem:exact=qiso}, 
when $(F,Q) = (F_{(E,P)},Q_{(E,P)})$ 
is the Green hyperbolic complex from Example \ref{ex:single-diff-op} 
associated to a Green hyperbolic operator $P$ 
acting on sections of a vector bundle $E$. 
\sk

Coming back to the connection with the Batalin-Vilkovisky formalism, 
the results listed above show that the proposed concept 
of a Green hyperbolic complex captures the essential information 
encoded by Lorentzian linear gauge theories. 
Auxiliary structures that are often considered in concrete applications, 
such as gauge-fixings and auxiliary fields, do not play 
any distinguished role and have no conceptual significance in the proposed approach, 
reflecting what is expected from a physical point of view.  
\sk

Of course, auxiliary structures may be very useful in practice. 
For instance, we identify a convenient auxiliary structure 
on a complex of linear differential operators $(F,Q)$ 
by introducing the concept of a {\it Green's witness} $W$. 
The latter consists of a collection of degree-decreasing 
linear differential operators such that $P := Q\, W + W\, Q$ 
is a degree-wise Green hyperbolic operator. 
The name Green's witness for $W$ is motivated 
by Theorem \ref{th:witness-Lambdapm}, which states that $(F,Q)$ 
is a Green hyperbolic complex as a consequence of the presence of $W$
and furthermore provides specific choices of retarded and advanced 
Green's homotopies $\Lambda_\pm := W\, G_\pm$ constructed out of 
the retarded and advanced Green's operators $G_\pm$ for $P$. 
This specific choice of $\Lambda_\pm$ coming from a Green's witness $W$ 
removes the technical complications (associated with homotopy colimits) 
that arise in our approach, demonstrating the usefulness of Green's witnesses. 
More precisely, in item~(1$^\prime$) 
one can replace $\Lambda_\hh$ with the simpler 
$\Lambda := \colim(\Lambda_+ - \Lambda_-): \FFF_\cc[1] \to \FFF_\sc$ 
involving sections with compact and respectively spacelike compact 
supports in the ordinary sense, 
see Remarks \ref{rem:Lambda} and \ref{rem:Lambda-qiso}. 
Furthermore, if the Green's witness $W$ 
is {\it formally self-adjoint} with respect to a differential pairing 
$(-,-)$ on $(F,Q)$ in the sense of 
Definition \ref{defi:self-adj-witness}, then graded anti-symmetrization 
and the homotopy $\lambda_M$ in item~(2$^\prime$) become superfluous, 
see Proposition \ref{propo:covariant-Poisson-self-adj-witness} 
and Remark \ref{rem:covariant-Poisson-self-adj-witness}. 
In this case also item~(3$^\prime$) becomes simpler, 
as one sees from Theorem \ref{th:Linfty-qiso-self-adj-witness}. 
\sk

Green's witnesses are frequently available in gauge-theoretic examples of interest. 
As a matter of fact, Green hyperbolic complexes 
endowed with a differential pairing and a formally self-adjoint Green's witness 
encompass linear Chern-Simons theory, see 
Examples \ref{ex:deRham}, \ref{ex:CS-diff-pairing}, 
\ref{ex:deRham-witness}, \ref{ex:Lambdapm-deRham}
and \ref{ex:CS-self-adj-witness}, 
formally self-adjoint normally hyperbolic operators $P = \Box^\nabla + B$, 
see Examples \ref{ex:single-diff-op}, 
\ref{ex:single-diff-op-diff-pairing}, \ref{ex:single-diff-op-witness}, 
\ref{ex:Lambdapm-single-diff-op} and 
\ref{ex:single-diff-op-self-adj-witness}, 
and Maxwell $p$-forms, see Examples \ref{ex:Maxwell-p-forms}, 
\ref{ex:MW-diff-pairing}, \ref{ex:Maxwell-p-forms-witness}, 
\ref{ex:Lambdapm-Maxwell-p-forms} and \ref{ex:MW-self-adj-witness}. 
In particular, the simplified versions of 
items~(1$^\prime$--3$^\prime$) apply to these examples. 
\sk

The previous discussion may induce the reader to think that it might be 
worth to regard a Green's witness as an essential, as opposed to auxiliary, 
structure. We stress that this is not the case. It is the general theory 
of Green hyperbolic complexes that allows one to correctly formalize 
uniqueness of retarded and advanced Green's homotopies in terms 
of contractibility of suitably defined spaces. 
The presence of a Green's witness only ensures that 
such spaces contain specific points that are particularly well-behaved, 
which is very useful in applications in the context of quantum field theory, 
see \cite{BeniniMusanteSchenkel}. Furthermore, in \cite{CS-WZW} 
a gauge-theoretic model is illustrated that possesses Green's homotopies 
(actually, a variant of this concept that is relevant in that context), 
but does not seem to admit a Green's witness. 
\sk

Let us now briefly outline the structure of the remainder of the paper. 
Section \ref{sec:prelim} collects the necessary preliminary material. 
More in detail, Section \ref{subsec:Green} reviews the basics 
of the theory of Green hyperbolic operators, 
Section \ref{subsec:Ch} recalls some useful tools from homological 
algebra and Section \ref{subsec:FunCh} is devoted to the category of 
cochain complex valued functors. In particular, we explain that 
the latter is a dg-category, whose cochain complex of morphisms 
from $\V$ to $\W$ is given by a mapping complex $\map(\V,\W)$ 
formalizing a concept of homotopy coherent natural transformations 
and (higher) homotopies between them, we illustrate 
a concrete model for the homotopy colimit functor 
and we observe that the latter is dg-left adjoint 
to the diagonal dg-functor. The core of the paper is 
Section \ref{sec:Green-hyp-cpx}. Retarded and advanced Green's 
homotopies, as well as Green hyperbolic complexes, 
are introduced in Section \ref{subsec:ret-adv-Green-homotopies}. 
In particular, Proposition \ref{propo:Greenhyp-acyclic} 
provides a recognition principle for Green hyperbolic complexes 
and Proposition \ref{propo:Greenhyp-unique} shows that the spaces 
of retarded/advanced Green's homotopies 
are either empty or contractible. 
The retarded-minus-advanced cochain map $\Lambda_\hh$ is constructed 
in Section \ref{subsec:ret-adv-map}, which is devoted to the proof 
of Theorem \ref{th:ret-minus-adv-qiso} stating that $\Lambda_\hh$ 
is a quasi-isomorphism. Section \ref{subsec:Poisson} defines 
the concept of a differential pairing. This is used in Propositions 
\ref{propo:covariant-Poisson} and \ref{propo:time-zero-Poisson} 
to construct two types of Poisson structures $\tau_M^\pm,\tau_M$ 
(which involve $\Lambda_{\pm\, \hh}, \Lambda_\hh$ and coincide 
up to a specified homotopy) and $\sigma_\Sigma$ (depending 
on the choice of a spacelike Cauchy surface $\Sigma \subseteq M$) 
defined on the domain and respectively on the codomain of $\Lambda_\hh$. 
Theorem \ref{th:Linfty-qiso} shows that 
$\Lambda_\hh$ is compatible with $\tau_M$ and $\sigma_\Sigma$ 
up to a homotopy that is constructed explicitly. 
The paper is completed by Section \ref{sec:witness}, 
which is devoted to Green's witnesses. 
Those are defined in Section \ref{subsec:witness}, 
which shows with Theorem \ref{th:witness-Lambdapm} 
that a Green's witness ensures Green hyperbolicity and provides 
specific choices of retarded and advanced Green's homotopies. 
This result yields all examples of Green hyperbolic complexes 
presented in this paper. Furthermore, we explain 
in Remarks \ref{rem:Lambda} and \ref{rem:Lambda-qiso} 
that a Green's witness simplifies considerably 
the construction of the retarded-minus-advanced cochain map 
and the proof that the latter is a quasi-isomorphism. 
Remark \ref{rem:exact=qiso} emphasizes 
that this quasi-isomorphism recovers the well-known exact sequence 
\eqref{eqn:PGP-exact-seq} associated with a Green hyperbolic operator. 
Section \ref{subsec:self-adj} concludes the paper by 
defining formally self-adjoint Green's witnesses, 
which lead to simplified versions of the Poisson structures 
$\tau_M$ and $\sigma_\Sigma$, 
see Propositions \ref{propo:covariant-Poisson-self-adj-witness} 
and \ref{propo:time-zero-Poisson-ordinary-colim}, and 
of their compatibility with the (simplified) retarded-minus-advanced 
cochain map $\Lambda$, see Theorem \ref{th:Linfty-qiso-self-adj-witness}.

\paragraph*{Relation to previous approaches:} The problem of constructing
Poisson structures for (gauge) field theories on Lorentzian manifolds
has a rich history, see e.g.\ \cite{Marolf,Marolf2,ForgerRomero} for
important earlier contributions and the introduction of \cite{Khavkine} for a detailed historical overview.
In the context of gauge theories, the traditional aim was to endow the algebra of \textit{gauge invariant on-shell observables}
of a gauge field theory with a suitable Poisson bracket structure, generalizing Peierls' original
construction \cite{Peierls} for theories without gauge symmetry.
Explicit proposals for such generalizations are given in \cite{HackSchenkel,Khavkine,Sharapov}.
It is important to emphasize that endowing the algebra of \textit{gauge invariant on-shell observables}
with a Poisson bracket is only a truncation (obtained by passing to $0^{\mathrm{th}}$ cohomology)
of the homological problem that we address and solve in our paper. Our Poisson structures
are defined at the level of cochain complexes (without ever passing to cohomology),
which is crucial to make contact with the recent developments in homotopical algebraic (quantum) field
theory \cite{BSW}. The use of cohomological methods in these earlier approaches, most notably 
in \cite{Khavkine,Sharapov,Benini,Khavkine2}, is mostly of practical nature in order to
facilitate the construction and prove properties (e.g.\ non-degeneracy) 
of the Poisson bracket on gauge invariant on-shell observables.
\sk

An alternative approach that has been taken in the literature is to construct
Poisson brackets in the context of the BRST or BV formalism, see e.g.\ \cite{FredenhagenRejzner,Sharapov2,WrochnaZahn}.
While these approaches are intrinsically homological, such as ours, the main difference
lies in the fact that all previous approaches (that we are aware of) manifestly make use of auxiliary
structures in their constructions, such as redundant fields and suitable gauge fixings. While this is
completely acceptable from a practical point of view, it leaves open important questions, most notably:
What is a quasi-isomorphism invariant definition of the concept of retarded/advanced Green's operators
and their associated Poisson structures? Are these in a suitable sense unique, and hence independent of any 
auxiliary choices such as redundant fields and gauge fixings? Our paper answers these questions.


\section{\label{sec:prelim}Preliminaries}

\subsection{\label{subsec:Green}Retarded and advanced Green's operators}
In this section we recall some basic concepts 
pertaining to the theory of Green hyperbolic linear differential operators 
on oriented and time-oriented globally hyperbolic Lorentzian manifolds
that will be used in the rest of this work. 
More details on this topic, including proofs of the statements 
recalled below, can be found in \cite{BGP,Bar}.
\sk

Consider an oriented and time-oriented globally hyperbolic 
Lorentzian manifold $M$ of dimension $m \geq 2$. Given a finite rank 
real or complex vector bundle $E \to M$, 
denote the vector space of its smooth sections with support 
contained in a closed subset $C \subseteq M$ by $\Gamma_C(E)$. 
(Following the usual convention we omit the subscript $_M$ to 
denote the vector space $\Gamma(E)$ of all smooth sections, 
without any support restriction.) 
We will often consider sections with support prescribed 
according to a {\it directed system} $\scrD$, that is a (non-empty) 
directed subset of the directed set $\cl$ of closed subsets of $M$. 
(More explicitly, we require that, for all $D_1,D_2 \in \scrD$, 
there exists $D \in \scrD$ such that $D_1 \subseteq D \supseteq D_2$). 
We can therefore define the vector space 
\begin{flalign}\label{eqn:colim-support}
\Gamma_{\scrD}(E) := \colim_{D \in \scrD}^{}\, \Gamma_D(E)
\end{flalign}
of smooth sections with $\scrD$-support as a colimit over a directed set. 
We shall consider the directed systems $\scrD = \cc, \pc, \fc, \spc, \sfc, \sc$ 
of compact, past compact, future compact, strictly past compact, 
strictly future compact and respectively spacelike compact 
closed subsets of $M$. For example, \eqref{eqn:colim-support} 
for $\scrD = \cc$ returns the usual vector space $\Gamma_{\cc}(E)$ 
of sections with compact support. 
\sk

A linear differential operator $P: \Gamma(E) \to \Gamma(E)$ is called 
{\it Green hyperbolic} if there exist 
{\it retarded and advanced Green's operators} $G_\pm$, 
which are by definition linear maps 
$G_\pm: \Gamma_{\cc}(E) \to \Gamma(E)$ such that, 
for all $\varphi \in \Gamma_{\cc}(E)$, 
\begin{enumerate}[label=(\roman*)]
\item $G_\pm\, P \varphi = \varphi$, 
\item $P\, G_\pm \varphi = \varphi$, 
\item $\supp(G_\pm \varphi)$ is contained 
in the causal future/past $J_M^\pm(\supp(\varphi))$ of the support 
of $\varphi$. 
\end{enumerate} 
In \cite{Bar} linear extensions 
$G_\pm: \Gamma_{\pc/\fc}(E) \to \Gamma_{\pc/\fc}(E)$ are defined 
on sections with past/future compact support 
in such a way that the properties~(i-iii) above hold for all $\varphi \in \Gamma_{\pc/\fc}(E)$. 
In particular, this entails that the restricted differential operators  
$P: \Gamma_{\pc/\fc}(E) \to \Gamma_{\pc/\fc}(E)$ are linear isomorphisms 
and the extended retarded/advanced Green's operators 
$G_\pm: \Gamma_{\pc/\fc}(E) \to \Gamma_{\pc/\fc}(E)$ 
are their (unique) inverses. Since the latter restrict 
to the retarded/advanced Green's operators 
$G_\pm: \Gamma_{\cc}(E) \to \Gamma(E)$, 
those are necessarily unique too.
\sk

We shall often consider the {\it retarded-minus-advanced propagator} 
\begin{flalign}\label{eqn:ret-minus-adv-prop}
G := G_+ - G_- : \Gamma_{\cc}(E) \longrightarrow \Gamma_{\sc}(E)\quad . 
\end{flalign}
The most relevant feature of the retarded-minus-advanced propagator $G$ 
is the exact sequence 
\begin{flalign}\label{eqn:PGP-exact-seq}
\xymatrix{
0 \ar[r] & \Gamma_{\cc}(E) \ar[r]^-{P} & \Gamma_{\cc}(E) \ar[r]^-{G} & \Gamma_{\sc}(E) \ar[r]^-{P} & \Gamma_{\sc}(E) \ar[r] & 0
}
\end{flalign}
see e.g.\ \cite{BeniniDappiaggi, Bar}, which follows directly 
from the properties of the (extended) Green's operators $G_\pm$. 
\sk

Consider a real vector bundle $E \to M$ that is endowed 
with a {\it fiber metric} $\langle-,-\rangle$, i.e.\ a fiber-wise non-degenerate, 
symmetric, bilinear form. One defines the integration pairing 
\begin{flalign}\label{eqn:fiberpairingintegration}
\langle\!\langle \varphi,\widetilde\varphi\rangle\!\rangle \,:=\, \int_M \langle\varphi,\widetilde\varphi\rangle\,\vol_M\quad, 
\end{flalign}
for all pairs of sections $(\varphi,\widetilde\varphi) \in \Gamma(E)^2$ 
with compact overlapping support, where $\vol_M$ denotes the volume form on $M$.
Given two real vector bundles endowed with fiber metrics $(E_1,\langle-,-\rangle_1)$ 
and $(E_2,\langle-,-\rangle_2)$ and a linear differential operator 
$Q: \Gamma(E_1) \to \Gamma(E_2)$, one defines its {\it formal adjoint} 
$Q^\ast: \Gamma(E_2) \to \Gamma(E_1)$ 
as the unique linear differential operator such that, 
for all pairs of sections $(\varphi_1,\varphi_2) \in \Gamma(E_1) \times \Gamma(E_2)$ 
with compact overlapping support, one has 
\begin{flalign}
\langle\!\langle \varphi_2, Q \varphi_1\rangle\!\rangle_2 = \langle\!\langle Q^\ast \varphi_2, \varphi_1 \rangle\!\rangle_{1} \quad.
\end{flalign}
Given a real vector bundle with fiber metric $(E,\langle-,-\rangle)$ and a linear differential 
operator $P: \Gamma(E) \to \Gamma(E)$, one says that $P$ is 
{\it formally self-adjoint} if $P^\ast = P$. 
When $P: \Gamma(E) \to \Gamma(E)$ 
is a formally self-adjoint Green hyperbolic linear differential operator, 
it follows that the associated retarded and advanced Green's operators 
$G_\pm$ are ``formal adjoints'' of each other, 
namely, for all $\varphi,\widetilde\varphi \in \Gamma_{\cc}(E)$, one has 
\begin{flalign}\label{eqn:Gpmadjoint}
\langle\!\langle \varphi, G_+ \widetilde\varphi \rangle\!\rangle \,=\, \langle\!\langle G_- \varphi,\widetilde\varphi \rangle\!\rangle \,=\, \langle\!\langle \widetilde\varphi,G_- \varphi \rangle\!\rangle\quad. 
\end{flalign}
As a consequence, the retarded-minus-advanced propagator $G = G_+ - G_-$ 
is ``formally skew-adjoint'', namely, 
for all $\varphi,\widetilde\varphi \in \Gamma_{\cc}(E)$, one has 
\begin{flalign}\label{eqn:Gskewadjoint}
\langle\!\langle \varphi, G \widetilde\varphi \rangle\!\rangle = - \langle\!\langle G \varphi, \widetilde\varphi \rangle\!\rangle = - \langle\!\langle \widetilde\varphi, G \varphi \rangle\!\rangle \quad.
\end{flalign}
This implies that the linear map
\begin{flalign}\label{eqn:standardtau}
\tau_M :=\langle\!\langle -, G(-)\rangle\!\rangle : \Gamma_{\cc}(E)^{\otimes 2} \longrightarrow \bbR
\end{flalign}
is anti-symmetric and hence it descends to a Poisson structure
$\tau_M:\Gamma_{\cc}(E)^{\wedge 2}\to \bbR $ on $\Gamma_{\cc}(E)$.
Using also \eqref{eqn:PGP-exact-seq}, one finds that this Poisson structure
descends to the cokernel $\operatorname{coker}(P:\Gamma_{\cc}(E)\to \Gamma_{\cc}(E)) = \Gamma_{\cc}(E)\big/ P\,\Gamma_{\cc}(E)$,
on which it becomes a non-degenerate Poisson structure.
As a consequence of the property \eqref{eqn:Gpmadjoint} of retarded/advanced Green's operators,
the Poisson structure \eqref{eqn:standardtau} can be presented in multiple equivalent ways
\begin{flalign}\label{eqn:standardtaudifferent}
\tau_M(\varphi\otimes\tilde{\varphi})= 
\langle\!\langle \varphi, G \widetilde{\varphi}\rangle\!\rangle = 
- \langle\!\langle \widetilde{\varphi}, G \varphi\rangle\!\rangle =
\pm \langle\!\langle \varphi, G_\pm\widetilde{\varphi}\rangle\!\rangle \mp
\langle\!\langle \widetilde{\varphi}, G_\pm\varphi\rangle\!\rangle \quad,
\end{flalign}
for all $\varphi,\tilde{\varphi}\in  \Gamma_{\cc}(E)$.

\subsection{\label{subsec:Ch}Cochain complexes}
This section reviews some elementary aspects of the theory 
of cochain complexes and sets our conventions. 
This topic is widely covered by the literature, 
see e.g.\ \cite{Weibel} and \cite{Hovey}. 
\sk

Let $\bbK$ be a field of characteristic zero. 
(In the main part of this work $\bbK$ will be either the real numbers $\bbR$ 
or the complex numbers $\bbC$.) A {\it cochain complex} 
$V=((V^n)_{n\in\bbZ}^{}, (Q^n)_{n\in\bbZ}^{})$ 
consists of a $\bbZ$-graded $\bbK$-vector space, i.e.\ a collection 
of $\bbK$-vector spaces $V^n$ labeled by their degree $n \in \bbZ$, 
and a differential $Q=(Q^n)_{n\in\bbZ}^{}$, 
i.e.\ a collection of $\bbK$-linear maps $Q^n : V^n\to V^{n+1}$ 
that increase the degree by $1$ 
and satisfy $Q^{n+1}\, Q^n = 0$ for all $n\in\bbZ$. 
A {\it cochain map} $f=(f^n)_{n\in\bbZ}^{}: V \to W$ 
consists of a collection of $\bbK$-linear maps $f^n : V^n\to W^n$, $n \in \bbZ$, 
that is compatible with the differentials of $V$ and $W$, 
i.e.\ $Q_W^n\,f^n = f^{n+1}\,Q_V^n$, for all $n\in\bbZ$. 
The category of cochain complexes over $\bbK$ 
with cochain maps as morphisms is denoted by $\Ch_\bbK$. 
\sk

The category $\Ch_\bbK$ is endowed with a closed symmetric monoidal structure, 
whose tensor product, monoidal unit, symmetric braiding and internal hom 
are described below. 
Given two cochain complexes $V,W\in\Ch_\bbK$, 
their tensor product $V\otimes W\in\Ch_\bbK$ 
consists of the graded vector space defined degree-wise for all $n\in\bbZ$ by 
\begin{subequations}\label{eqn:tensorproduct}
\begin{flalign}
(V\otimes W)^n\,:=\,\bigoplus_{q\in\bbZ} \big(V^q\otimes W^{n-q}\big)\quad,
\end{flalign}
and of the differential 
$Q_\otimes = (Q_\otimes^n)_{n\in\bbZ}^{}$ defined by the graded Leibniz rule
\begin{flalign}
Q_\otimes^n(v\otimes w) := Q_V^q v \otimes w + (-1)^{q}\,  v\otimes Q_W^{n-q} w\quad,
\end{flalign}
\end{subequations} 
for all $v\in V^q$ and $w\in W^{n-q}$.
The monoidal unit $\bbK\in\Ch_\bbK$ is obtained regarding the ground field 
as a cochain complex concentrated in degree $0$ 
(whose differential necessarily vanishes). 
The symmetric braiding is given by the cochain isomorphisms
$V\otimes W \overset{\cong}{\to} W\otimes V$, 
$v\otimes w \mapsto (-1)^{\vert v\vert \,\vert w\vert} \,w\otimes v$ 
determined by the Koszul sign rule, for all homogeneous $v\in V$ and $w\in W$
whose degree is denoted by $\vert v\vert,\vert w\vert \in\bbZ$. 
Furthermore, given two cochain complexes $V,W\in\Ch_\bbK$, 
their internal hom $[V,W]\in\Ch_\bbK$ consists of 
the graded vector space defined degree-wise for all $n\in\bbZ$ by 
\begin{subequations}\label{eqn:internalhom}
\begin{flalign}
[V,W]^n \,:=\, \prod_{q\in\bbZ} \Hom_\bbK(V^q, W^{n+q})\quad,
\end{flalign}
where $\Hom_\bbK$ denotes the vector space of linear maps, 
and of the differential $\partial=(\partial^n)_{n\in\bbZ}^{}$ defined by 
\begin{flalign}
\partial^n f \,:=\, Q_W \circ f - (-1)^n f \circ Q_V \quad,
\end{flalign}
\end{subequations}
for all $f \in [V,W]^n$. 
The category of cochain complexes $\Ch_\bbK$ becomes a dg-category 
when endowed with the cochain complexes 
of morphisms from $V$ to $W$ given by the internal hom $[V,W]$ 
and the obvious identities and compositions. 
\sk

Associated with every cochain complex $V\in\Ch_\bbK$ is its cohomology 
$H^\bullet(V)$, a graded vector space defined degree-wise by 
$H^n(V) := \ker(Q^n)/\operatorname{im}(Q^{n-1})$, for all $n\in\bbZ$. 
Cohomology extends in an obvious way to a functor $H^\bullet$ on $\Ch_\bbK$ 
taking values in the category of graded vector spaces. 
A cochain map $f:V\to W$ is a {\it quasi-isomorphism} 
if passing to cohomology gives an isomorphism 
$H^\bullet(f): H^\bullet(V)\to H^\bullet(W)$ of graded vector spaces. 
When this is the case, one says that $V$ and $W$ are {\it quasi-isomorphic}. 
Informally, quasi-isomorphic cochain complexes should be regarded 
as ``being the same''. One approach to the formalization of this idea 
is offered by model category theory \cite{Hovey}. 
Concretely, a model category is a bicomplete category that comes endowed 
with a model structure, consisting of three distinguished classes 
of morphisms, called weak equivalences, fibrations and cofibrations, 
subject to suitable axioms, see \cite[Sec.~1.1]{Hovey}. 
Conceptually, the primary role is played by the weak equivalences, 
which formalize a relaxed notion of ``being the same'' compared to isomorphisms. 
Fibrations and cofibrations, instead, are crucial from a practical 
viewpoint as they allow to construct homotopical functors, 
i.e.\ functors that preserve weak equivalences, see \cite[Ch.~2]{Riehl}. 
For instance, ordinary (co)limits frequently fail to preserve weak 
equivalences; model category theory fixes this issue replacing 
them with the respective derived functors, called homotopy (co)limits 
\cite[Ch.~19]{Hirschhorn}, which are homotopical functors by construction. 
(Indeed, homotopy (co)limits play a key role in the present paper 
precisely because they preserve weak equivalences.) 
\sk

It is proven in \cite[Secs.~2.3 and 4.2]{Hovey} that the category of cochain complexes 
$\Ch_\bbK$ carries a closed symmetric monoidal model category structure, 
which is determined by defining the weak equivalences as the 
quasi-isomorphisms and the fibrations as the degree-wise surjective 
cochain maps. (Cofibrations are detected by the so-called 
left-lifting property against acyclic fibrations, 
i.e.\ morphisms that are simultaneously weak equivalences and fibrations.) 
It is not difficult to check that every cochain complex $V$ 
in the model category $\Ch_\bbK$ is both fibrant and cofibrant, 
i.e.\ the unique morphism $V \to 0$ to the terminal object is a fibration
and the unique morphism $0 \to V$ from the initial object is a cofibration. 
In particular, this guarantees that both
the tensor product functor $\otimes : \Ch_\bbK\times\Ch_\bbK\to  \Ch_\bbK$
and the internal hom functor 
$[-,-]: \Ch_\bbK^\op\times \Ch_\bbK \to \Ch_\bbK$ are homotopical, i.e.\ 
they preserve quasi-isomorphisms.
\sk

The internal hom $[V,W]\in\Ch_\bbK$ between two cochain complexes 
$V,W \in \Ch_\bbK$ admits an elegant interpretation in terms of (higher) 
cochain homotopies. In fact, given two $n$-cocycles $f,g \in [V,W]^n$, 
i.e.\ $\partial f = 0 = \partial g$, 
one defines a {\it cochain homotopy} $\lambda$ from $f$ to $g$ 
as an $(n-1)$-cochain $\lambda \in [V,W]^{n-1}$ 
such that $\partial \lambda = g - f$. More explicitly, a cochain homotopy 
$\lambda$ consists of a collection of linear maps 
$(\lambda^q : V^q \to W^{q+n-1})_{q\in\bbZ}^{}$ such that 
$Q_W^{q+n-1} \circ \lambda^q - (-1)^{n-1} \lambda^{q+1} \circ Q_V^q = g^q - f^q$, 
for all $q\in\bbZ$. Notice that $\partial \lambda$ is an $n$-coboundary 
in $[V,W]$, therefore a necessary and sufficient condition 
for the existence of a cochain homotopy is that the cohomology classes 
$[f]=[g] \in H^n([V,W])$ coincide. Note further that this concept of
cochain homotopies specializes for $n=0$ to the ordinary concept of cochain homotopies 
between two cochain maps $f,g: V \to W$. 
In fact, by definition of the internal hom complex, 
see \eqref{eqn:internalhom}, both $f$ and $g$ are $0$-cocycles in $[V,W]$. 
\sk

Finally, let us set our convention for shifts of cochain complexes. 
Given a cochain complex $V\in \Ch_\bbK$ and an integer $p\in\bbZ$, 
we define its {\it $p$-shift} $V[p]\in \Ch_\bbK$ as the cochain complex 
that consists of the graded vector space defined degree-wise 
for all $n\in\bbZ$ by 
\begin{subequations}
\begin{flalign}
V[p]^n := V^{n+p}\quad,
\end{flalign} 
and of the differential $Q_{V[p]}=(Q_{V[p]}^n)_{n\in \bbZ}^{}$ defined by 
\begin{flalign}
Q_{V[p]}^n := (-1)^p\,Q_V^{n+p}\quad,
\end{flalign}
\end{subequations}
for all $n\in\bbZ$.
One immediately observes that $V[p][q] = V[p+q]$, for all $p,q\in\bbZ$, 
and $V[0] = V$. Recalling also the definition of the tensor product 
\eqref{eqn:tensorproduct}, one obtains natural cochain isomorphisms 
$\bbK[p]\otimes V \cong V[p]$ for all $p\in\bbZ$.

\subsection{\label{subsec:FunCh}\texorpdfstring{$\Ch_\bbK$-valued}{Cochain complex-valued} functors}
Along with cochain complexes, we shall also consider functors 
$\V: \CC \to \Ch_\bbK$ on a (small) category $\CC$ (often just a directed set) 
taking values in $\Ch_\bbK$. Taking also the natural transformations 
$\eta: \V \to \W$ between such functors as morphisms, 
one obtains the functor category $\Ch_\bbK^\CC$. 
Using the closed symmetric monoidal structure on $\Ch_\bbK$, 
we can equip the functor category $\Ch_\bbK^\CC$ 
with tensoring, powering and enriched hom over $\Ch_\bbK$ as follows. 
\sk

Given a cochain complex $V \in \Ch_\bbK$ and a functor 
$\V \in \Ch_\bbK^\CC$, their {\it tensoring} 
\begin{flalign}
V \otimes \V \in \Ch_\bbK^\CC
\end{flalign}
is defined as the functor that assigns to each $c \in \CC$ 
the tensor product $V \otimes \V(c) \in \Ch_\bbK$ 
and to each morphism $\gamma: c_0 \to c_1$ in $\CC$ the cochain map 
$\id \otimes \V(\gamma): V \otimes \V(c_0) \to V \otimes \V(c_1)$ 
in $\Ch_\bbK$. 
\sk

Given a cochain complex $V \in \Ch_\bbK$ and a functor 
$\V \in \Ch_\bbK^\CC$, their {\it powering} 
\begin{flalign}
\V^V \in \Ch_\bbK^\CC
\end{flalign}
is defined as the functor that assigns to each $c \in \CC$ the internal hom 
$[V,\V(c)] \in \Ch_\bbK$ and to each morphism $\gamma: c_0 \to c_1$ 
in $\CC$ the cochain map 
$[\id,\V(\gamma)]: [V,\V(c_0)] \to [V,\V(c_1)]$ in $\Ch_\bbK$. 
\sk

Given two functors $\V,\W \in \Ch_\bbK^\CC$, their {\it enriched hom} 
\begin{subequations}\label{eqn:enriched-hom}
\begin{flalign}
\hom(\V,\W) \in \Ch_\bbK
\end{flalign}
is defined as the equalizer 
\begin{flalign}\label{eqn:enriched-hom-pair}
\hom(\V,\W) := \lim \Bigg(
\xymatrix@C=3em{
\displaystyle\prod_{c_0 \in \CC} \big[ \V(c_0),\W(c_0) \big]~~ \ar@<2pt>[r]^-{\V^\ast} \ar@<-2pt>[r]_-{\W_\ast} & \displaystyle\prod_{\gamma: c_0 \to c_1} \big[ \V(c_0),\W(c_1) \big]
}
\Bigg)\quad,
\end{flalign}
where $\V^\ast$ is defined on the $\gamma$-component of the codomain 
by ``pull-back'' along $\gamma : c_0\to c_1$ of the $c_1$-component of the domain, 
i.e.\ $\pr_{\gamma}\, \V^\ast := [\V(\gamma),\id]\, \pr_{c_1}$, 
while $\W_\ast$ is defined on the $\gamma$-component of the codomain 
by ``push-forward'' along $\gamma:c_0\to c_1$ of the $c_0$-component of the domain, 
i.e.\ $\pr_{\gamma}\, \W_\ast := [\id,\W(\gamma)]\, \pr_{c_0}$. 
More explicitly, $\hom(\V,\W) \in \Ch_\bbK$ consists of the graded 
vector space given degree-wise for all $n \in \bbZ$ by 
\begin{flalign}\label{eqn:enriched-hom-graded}
\hom(\V,\W)^n = \Bigg\{ \eta \in \prod_{c_0 \in \CC} \big[ \V(c_0),\W(c_0) \big]^n 
\,:\; \pr_{c_1}\eta \circ \V(\gamma) = \W(\gamma) \circ \pr_{c_0}\eta \,,\; \forall \gamma: c_0 \to c_1 \Bigg\} \quad,
\end{flalign}
i.e.\ $\eta = (\eta_{c_0})_{c_0 \in \CC}$ is a degree $n$ natural transformation 
from $\V$ to $\W$ (regarded as functors valued in graded vector spaces), 
and of the differential $\partial$ defined component-wise for all $c_0 \in \CC$ by 
\begin{flalign}
\pr_{c_0}(\partial \eta) := \partial(\pr_{c_0}\eta)\quad.
\end{flalign}
\end{subequations}

Tensoring, powering and enriched hom are related via the adjunctions 
exhibited by the isomorphisms 
\begin{flalign}
\Ch_\bbK^\CC \big( \V, \W^V \big) \cong \Ch_\bbK^\CC \big( V \otimes \V, \W \big) \cong \Ch_\bbK \big( V,\hom(\V,\W) \big)\quad, 
\end{flalign}
which are natural with respect to $V \in \Ch_\bbK$ 
and $\V,\W \in \Ch_\bbK^\CC$. 
\sk

Also in the functor category $\Ch_\bbK^\CC$ 
one has a notion of weak equivalences 
given by the {\em natural quasi-isomorphisms} $f: \V \to \W$ 
in $\Ch_\bbK^\CC$, i.e.\ the natural transformations whose components 
$f_c: \V(c) \to \W(c)$ in $\Ch_\bbK$ are quasi-isomorphisms, 
for all $c \in \CC$. 
This class of weak equivalences 
is part of the (projective) model category structure on $\Ch_\bbK^\CC$, 
which is entirely determined by defining also the fibrations 
as the natural transformations $f: \V \to \W$ in $\Ch_\bbK^\CC$ 
whose components $f_c: \V(c) \to \W(c)$ in $\Ch_\bbK$ are fibrations, 
i.e.\ degree-wise surjective cochain maps, for all $c \in \CC$. 
(Once again, the cofibrations are detected by the left-lifting property 
against acyclic fibrations.)
With this model category structure and the tensoring, powering 
and enriched hom from above, $\Ch_\bbK^\CC$ turns out to be 
a $\Ch_\bbK$-model category, see \cite[Sec.~4.2]{Hovey}. 
In contrast to the model category $\Ch_\bbK$, whose objects are all fibrant and cofibrant, 
objects of the functor category $\Ch_\bbK^\CC$ are in general fibrant, but they may 
fail to be cofibrant. 
As a consequence, the enriched hom functor 
$\hom(-,-): (\Ch_\bbK^\CC)^\op\times \Ch_\bbK^\CC \to \Ch_\bbK$ may fail 
to preserve weak equivalences. This shortcoming is solved 
by constructing the associated derived functor 
$\map(-,-): (\Ch_\bbK^\CC)^\op\times \Ch_\bbK^\CC \to \Ch_\bbK$, 
which instead preserves weak equivalences 
(see Remark \ref{rem:map-preserves-we}). 
Concretely, we replace the enriched hom $\hom(\V,\W) \in \Ch_\bbK$ 
from $\V \in \Ch_\bbK^\CC$ to $\W \in \Ch_\bbK^\CC$ 
with the mapping complex $\map(\V,\W) \in \Ch_\bbK$ defined below 
as the homotopy limit of a suitable cosimplicial cochain complex 
$C(\V,\W) \in \Ch_\bbK^\Delta$. See also \cite[Sec.~3]{Tamarkin} 
for an alternative presentation of this construction 
and Remark \ref{rem:map-preserves-we} for a description in terms of resolutions and derived functors. 

\paragraph*{Mapping complex:} The parallel pair in \eqref{eqn:enriched-hom-pair}
is only a truncation of the cosimplicial cochain complex 
\begin{subequations}\label{eqn:C}
\begin{flalign}
C(\V,\W) = \Bigg( \xymatrix@C=2em{
C(\V,\W)^0 \ar@<4pt>[r] \ar@<-4pt>[r] & C(\V,\W)^1 
\ar@<8pt>[r] \ar@<0pt>[r] \ar@<-8pt>[r] \ar@<0pt>[l] & \ldots \ar@<4pt>[l] \ar@<-4pt>[l]
} \Bigg) \in \Ch_\bbK^\Delta 
\end{flalign} 
that consists of the cochain complexes 
\begin{flalign}\label{eqn:C-cpx}
C(\V,\W)^n := \prod_{\und{c}: [n] \to \CC} \big[ \V(c_0),\W(c_n) \big] \in \Ch_\bbK\quad,
\end{flalign}
for all integers $n \geq 0$, of the coface maps 
\begin{flalign}\label{eqn:C-coface}
d^k: C(\V,\W)^n \longrightarrow C(\V,\W)^{n+1}
\end{flalign}
in $\Ch_\bbK$, defined in \eqref{eqn:C-coface-def} below 
for all integers $n \geq 0$ and $k=0,\ldots,n+1$, and of the codegeneracy maps 
\begin{flalign}\label{eqn:C-codeg}
s^k: C(\V,\W)^{n+1} \longrightarrow C(\V,\W)^{n}
\end{flalign}
\end{subequations}
in $\Ch_\bbK$, defined in \eqref{eqn:C-codeg-def} below 
for all integers $n \geq 0$ and $k=0,\ldots,n$. 
The product in \eqref{eqn:C-cpx} runs over all $\CC$-valued functors 
$\und{c}: [n] \to \CC$ on the totally ordered set 
$[n] := \{0 < 1 < \cdots < n\}$ (regarded as a category). 
(Equivalently, $\und{c} = 
(c_0 \overset{\gamma_0}{\to} c_1 \overset{\gamma_1}{\to} \cdots \overset{\gamma_{n-1}}{\to} c_n)$ 
is an $n$-tuple of composable morphisms in $\CC$.) 
Denoting by $\hat{k}: [n] \to [n+1]$ the injective order preserving map 
that skips the element $k \in [n+1]$, one defines the coface map $d^0$ by
\begin{subequations}\label{eqn:C-coface-def}
\begin{flalign}
\xymatrix@C=5em{
C(\V,\W)^n \ar[d]_-{\pr_{\und{c} \circ \hat{0}}} \ar@{-->}[r]^-{d^0} & C(\V,\W)^{n+1} \ar[d]^-{\pr_{\und{c}}} \\ 
\big[\V(c_1),\W(c_{n+1})\big] \ar[r]_-{[\V(\gamma_0),\id]} & \big[\V(c_0),\W(c_{n+1})\big]
}
\end{flalign}
the coface map $d^k$, $k=1,\ldots,n$, by 
\begin{flalign}
\xymatrix@C=5em{
C(\V,\W)^n \ar[d]_-{\pr_{\und{c} \circ \hat{k}}} \ar@{-->}[r]^-{d^k} & C(\V,\W)^{n+1} \ar[d]^-{\pr_{\und{c}}} \\ 
\big[\V(c_0),\W(c_{n+1})\big] \ar[r]_-{\id} & \big[\V(c_0),\W(c_{n+1})\big]
}
\end{flalign}
and the coface map $d^{n+1}$ by 
\begin{flalign}
\xymatrix@C=5em{
C(\V,\W)^n \ar[d]_-{\pr_{\und{c} \circ \widehat{n+1}}} \ar@{-->}[r]^-{d^{n+1}} & C(\V,\W)^{n+1} \ar[d]^-{\pr_{\und{c}}} \\ 
\big[\V(c_0),\W(c_{n})\big] \ar[r]_-{[\id,\W(\gamma_{n})]} & \big[\V(c_0),\W(c_{n+1})\big]
}
\end{flalign}
\end{subequations}
Similarly, denoting by $\check{k}: [n+1] \to [n]$ the surjective order 
preserving map that hits the element $k \in [n]$ twice, one defines 
the codegeneracy map $s^k$, $k=0,\ldots,n$, by 
\begin{flalign}\label{eqn:C-codeg-def}
\xymatrix@C=5em{
C(\V,\W)^{n+1} \ar[d]_-{\pr_{\und{c} \circ \check{k}}} \ar@{-->}[r]^-{s^k} & C(\V,\W)^{n} \ar[d]^-{\pr_{\und{c}}} \\ 
\big[\V(c_0),\W(c_{n})\big] \ar[r]_-{\id} & \big[\V(c_0),\W(c_{n})\big]
}
\end{flalign}
The {\em mapping complex} from $\V \in \Ch_\bbK^\CC$ to $\W \in \Ch_\bbK^\CC$ is defined as the homotopy limit 
\begin{flalign}\label{eqn:map}
\map(\V,\W) := \holim(C(\V,\W)) \in \Ch_\bbK
\end{flalign}
of the cosimplicial diagram 
$C(\V,\W) \in \Ch_\bbK^\Delta$ in \eqref{eqn:C}. 
This may be computed by the $\prod$-total complex 
associated with the cosimplicial cochain complex $C(\V,\W)$. Explicitly, 
$\map(\V,\W) \in \Ch_\bbK$ consists of the graded $\bbK$-vector space 
defined degree-wise for all $n \in \bbZ$ by 
\begin{flalign}\label{eqn:map-graded}
\map(\V,\W)^n := \prod_{q \geq 0} C(\V,\W)^{q,n-q}\quad,
\end{flalign}
where $q\geq 0$ denotes the cosimplicial degree and $n-q\in\bbZ$ 
the cochain degree. The (total) differential 
\begin{subequations}\label{eqn:map-diff-def}
\begin{flalign}\label{eqn:map-diff}
\delta := \delta_{\hh} + \delta_{\vv}
\end{flalign}
is the sum of the horizontal differential $\delta_{\hh}$, 
defined component-wise by 
\begin{flalign} 
\pr_0 \circ \delta_{\hh} := 0\quad, \qquad \pr_q \circ \delta_{\hh} := \sum_{k=0}^q (-1)^k\, d^{q-k} \circ \pr_{q-1}\quad, 
\end{flalign}
for all $q \geq 1$, and of the vertical differential 
$\delta_{\vv}$, defined component-wise for all $q \geq 0$ 
and all $\und{c}: [q] \to \CC$ by 
\begin{flalign} 
\pr_{q,\und{c}} \circ \delta_{\vv} := (-1)^q\, \partial \circ \pr_{q,\und{c}} \quad, 
\end{flalign}
\end{subequations}
where $\pr_{q,\und{c}} := \pr_{\und{c}} \circ \pr_{q}$ 
denotes the projection onto the $(q,\und{c})$-component 
$[\V(c_0),\W(c_q)]^{n-q}$ of $\map(\V,\W)^n$ and $\partial$ is the 
differential of the internal hom $[\V(c_0),\W(c_q)] \in \Ch_\bbK$. 
The signs in $\delta_{\hh}$ xcan be understood projecting 
onto $\und{c}$-components. The $k$-th summand of 
$\delta_{\hh}$ acts via $d^{q-k}$, hence 
$\delta_{\hh}$ is pulled through the last $k$ morphisms in $\CC$ 
of the tuple $\und{c}$ before acting. Since each of those morphisms 
contributes $1$ to the total degree and $\delta_{\mathrm{h}}$ has degree 
$1$, this gives rise to the sign $(-1)^k$. Similarly, the sign in 
$\delta_{\vv}$ arises from the fact that $\delta_{\vv}$ acts by $\partial$ 
after being pulled through $\und{c}$. Since the latter contributes $q$ 
to the total degree and $\delta_{\vv}$ has degree $1$, 
this gives rise to the sign $(-1)^q$. 

\begin{rem}\label{rem:map-preserves-we}
The explicit model of the mapping complex $\map(-,-)$ 
presented above stems from a cotriple resolution 
in the sense of \cite[Sec.~13.3]{Fresse}. 
Let us explain this relation in more detail. 
Consider the forgetful functor $U : \Ch_\bbK^{\CC}\to \prod_{c\in\CC} \Ch_\bbK$
that assigns to a functor $\V\in \Ch_\bbK^{\CC}$ its family of values $\big(\V(c)\big)_{c\in \CC}$
on all objects. Because the category $\Ch_\bbK$ is cocomplete, 
this functor admits a left adjoint, i.e.\ there exists an adjunction
\begin{flalign}
\xymatrix{
F: \prod\limits_{c\in\CC} \Ch_\bbK 
\ar@<0.5ex>[r]~ &~ \ar@<0.5ex>[l] \Ch_\bbK^{\CC}: U
}\quad,
\end{flalign}
which defines a concept of free objects in $\Ch_\bbK^{\CC}$.
Associated with this adjunction, one defines a comonad 
(sometimes also called a cotriple) 
$T := F U: \Ch_\bbK^{\CC} \to \Ch_\bbK^{\CC}$ 
with coproduct $T = F U \to F (U F) U = T^2$ 
constructed out of the adjunction unit and counit 
$T = F U \to \id$ given by the adjunction counit. 
For $\V \in \Ch_\bbK^{\CC}$, the comonad $T$ allows us to define 
a simplicial resolution 
\begin{flalign}
\mathrm{Res}(\V) := 
\bigg( \xymatrix@C=1.5em{
T(\V) \ar@<0pt>[r] & T^2(\V) \ar@<4pt>[l] \ar@<-4pt>[l] \ar@<4pt>[r] \ar@<-4pt>[r] & \cdots \ar@<8pt>[l] \ar@<0pt>[l] \ar@<-8pt>[l]
} \bigg) \in (\Ch_\bbK^\CC)^{\Delta^\op} \;. 
\end{flalign}
For $\W \in \Ch_\bbK^{\CC}$, composing the simplicial resolution above 
with the enriched hom $\hom(-,\W)$ computes the cosimplicial cochain complex   
\begin{flalign}
C(\V,\W) = \hom \big( \mathrm{Res}(\V),\W \big) = \bigg( \xymatrix@C=1.5em{
\hom \big( T(\V), \W \big) \ar@<0pt>[r] & \hom \big( T^2(\V), \W \big) \ar@<4pt>[l] \ar@<-4pt>[l] \ar@<4pt>[r] \ar@<-4pt>[r] & \cdots \ar@<8pt>[l] \ar@<0pt>[l] \ar@<-8pt>[l]
} \bigg) \in \Ch_\bbK^\Delta
\end{flalign}
from \eqref{eqn:C}. Finally, forming the $\prod$-total complex 
computes the homotopy limit and determines the mapping complex 
$\map(\V,\W) = \holim(C(\V,\W)) \in \Ch_\bbK$ as in \eqref{eqn:map}. 
Similar techniques as in \cite[Lem.~13.3.3 and Ch.~17]{Fresse} 
show that this construction presents the mapping complex $\map(-,-)$ 
as the derived functor of the enriched hom $\hom(-,-)$. 
\sk

Being a derived functor, it follows that $\map(-,-)$ 
preserves weak equivalences. This can also be seen more directly 
via the following argument. 
Let $f:\V^\prime \to \V$ and $g:\W \to \W^\prime$ in $\Ch_\bbK^\CC$ 
be natural quasi-isomorphisms. For each integer $n \geq 0$ 
and each functor $\und{c}: [n] \to \CC$, one has that 
$[f_{c_0},g_{c_n}]: [\V(c_0),\W(c_n)] \to [\V^\prime(c_0),\W^\prime(c_n)]$ 
in $\Ch_\bbK$ is a quasi-isomorphism because the internal hom 
$[-,-]: \Ch_\bbK^\op \times \Ch_\bbK \to \Ch_\bbK$ 
is a homotopical functor. This entails that 
$C(f,g): C(\V,\W) \to C(\V^\prime,\W^\prime)$ in $\Ch_\bbK^\Delta$ 
is a natural quasi-isomorphism of cosimplicial cochain complexes, 
i.e.\ a natural transformation that is a quasi-isomorphism in each cosimplicial degree. 
Since the homotopy limit $\holim: \Ch_\bbK^\Delta \to \Ch_\bbK$ 
is a homotopical functor, it follows that 
$\map(f,g) := \holim(C(f,g)): \map(\V,\W) \to \map(\V^\prime,\W^\prime)$ 
in $\Ch_\bbK$ is a quasi-isomorphism. 
\end{rem}

\begin{rem}\label{rem:hom-vs-map}
The mapping complex $\map(\V,\W)\in \Ch_\bbK$ in \eqref{eqn:map} 
formalizes a concept of {\em homotopy coherent natural transformations} 
from $\V$ to $\W$, as well as notions of (higher) homotopies between such
mappings. Indeed, an $n$-cochain $\eta \in \map(\V,\W)^n$ consists of 
an $(n-q)$-cochain $\pr_q \eta \in C(\V,\W)^{q,n-q}$
for each $q \geq 0$, whose components 
$\pr_{q,\und{c}} \eta \in [\V(c_0),\W(c_q)]^{n-q}$ 
are labeled also by functors $\und{c}: [q] \to \CC$. 
The cocycle condition $\delta \eta = 0$ can be interpreted as follows. 
For all $c_0 \in \CC$, $\pr_{0,c_0} \eta$ are cocycles 
since $\partial (\pr_{0,c_0} \eta) = \pr_{0,c_0} (\delta \eta) = 0$. 
Even though these may fail to be the components of a natural 
transformation, $\eta$ contains the data $\pr_{1,\gamma_0} \eta$, 
for all $\gamma_0: c_0 \to c_1$ in $\CC$, of a homotopy 
witnessing this failure. Indeed, one has   
\begin{flalign}
\W(\gamma_0) \circ (\pr_{0,c_0} \eta) - (\pr_{0,c_1} \eta) \circ \V(\gamma_0) = \pr_{1,\gamma_0} (\delta_{\mathrm{h}} \eta) = - \pr_{1,\gamma_0} (\delta_{\mathrm{v}} \eta) = \partial (\pr_{1,\gamma_0} \eta) \quad. 
\end{flalign}
The data $\pr_{1,\gamma_0} \eta$, 
for all $\gamma_0: c_0 \to c_1$ in $\CC$, may again fail to be 
natural (in the appropriate sense), but again this failure is controlled 
in a similar fashion by the higher homotopy consisting of the data 
$\pr_{2,\und{c}} \eta$, for all 
$\und{c} = (c_0 \overset{\gamma_0}{\to} c_1 \overset{\gamma_1}{\to} c_2)$ in $\CC$. 
This pattern goes on with increasingly higher homotopies. 
\sk

Note that any $n$-cochain $\eta \in \hom(\V,\W)^n$ in the enriched hom
determines a corresponding $n$-cochain 
$\widetilde{\eta} \in \map(\V,\W)^n$ in the mapping complex. 
The latter is defined by setting 
$\pr_0 \widetilde{\eta} := \eta \in C(\V,\W)^{0,n}$ 
and $\pr_q \widetilde{\eta} := 0 \in C(\V,\W)^{q,n-q}$ 
for $q \geq 1$. It is straightforward to confirm 
that the assignment $\eta \mapsto \widetilde{\eta}$ is compatible 
with the respective differentials, hence we obtain an inclusion 
\begin{flalign}\label{eqn:hom-vs-map}
\hom(\V,\W) \overset{\subseteq}{\longrightarrow} \map(\V,\W)
\end{flalign}
in $\Ch_\bbK$. This inclusion may be interpreted by saying that 
(strict) naturality is a special case of homotopy coherent naturality. 
In the rest of the paper we shall identify the $n$-cochains 
$\eta \in \hom(\V,\W)^n$ in the enriched hom 
with the corresponding $n$-cochains $\widetilde{\eta} \in \map(\V,\W)^n$ 
in the mapping complex, thus dropping the decoration 
$\widetilde{\phantom{\eta}}$ from our notation. 
\end{rem}

\paragraph*{dg-category structure:}
The functor category $\Ch_\bbK^\CC$ can be endowed with 
a dg-category structure whose cochain complex of morphisms 
from $\V$ to $\W$ is given by the mapping complex $\map(\V,\W)$. 
The identities are the obvious ones and the compositions are explicitly 
given, for all $\V, \W, \Z \in \Ch_\bbK^\CC$, by the cochain maps 
\begin{subequations}
\begin{flalign}
\circ: \map(\W,\Z) \otimes \map(\V,\W) \longrightarrow \map(\V,\Z)
\end{flalign}
in $\Ch_\bbK$ that send $g \in \map(\W,\Z)^n$ and 
$f \in \map(\V,\W)^m$ to $g \circ f \in \map(\V,\Z)^{m+n}$, 
which is defined component-wise, for all $q \geq 0$ 
and all functors $\und{c}: [q] \to \CC$, by 
\begin{flalign}
\pr_{q,\und{c}} (g \circ f) := \sum_{k=0}^q (-1)^{k(q-k+n)}\, (\pr_{q-k,\und{c}^{\geq k}} g) \circ (\pr_{k,\und{c}^{\leq k}} f) \quad. 
\end{flalign}
\end{subequations}
In the equation displayed above 
$\und{c}^{\leq k}:= (c_0 \overset{\gamma_0}{\to} \cdots \overset{\gamma_{k-1}}{\to} c_k)$
and $\und{c}^{\geq k}:=(c_k \overset{\gamma_{k+1}}{\to} \cdots \overset{\gamma_{q-1}}{\to} c_q)$ 
denote the tuples of composable morphisms in $\CC$ obtained 
by splitting $\und{c}$ at $c_k$. 
Recalling that each morphism in $\CC$ contributes $1$ to the total 
degree in the mapping complex, the sign of the $k$-th summand 
can be understood as the one that arises when $\und{c}^{\leq k}$ 
is pulled through $\und{c}^{\geq k}$ and $g$. 

\paragraph*{Homotopy colimits:}
Along with mapping complexes, we shall make extensive use of homotopy 
colimits $\hocolim : \Ch_\bbK^\CC\to \Ch_\bbK$. Specifically, we shall use the explicit model presented below. 
Given a functor $\V \in \Ch_\bbK^\CC$, 
consider the simplicial cochain complex (called the simplicial replacement of $\V$)
\begin{subequations}
\begin{flalign}
S(\V) = \Bigg( \xymatrix@C=2em{
S(\V)^0 \ar@<0pt>[r] & S(\V)^1 \ar@<4pt>[l] \ar@<-4pt>[l] \ar@<4pt>[r] \ar@<-4pt>[r] & \ldots \ar@<8pt>[l] \ar@<0pt>[l] \ar@<-8pt>[l]
} \Bigg) \in \Ch_\bbK^{\Delta^\op} 
\end{flalign}
that consists of the cochain complexes 
\begin{flalign}
S(\V)^n := \bigoplus_{\und{c}: [n] \to \CC} \V(c_0) \in \Ch_\bbK \quad,
\end{flalign}
for all integers $n \geq 0$, of the face maps 
\begin{flalign}
d_k: S(\V)^{n+1} \longrightarrow S(\V)^n
\end{flalign}
in $\Ch_\bbK$, defined in \eqref{eqn:S-face-def} below 
for all integers $n \geq 0$ and $k = 0, \ldots, n+1$, 
and of the degeneracy maps 
\begin{flalign}
s_k: S(\V)^n \longrightarrow S(\V)^{n+1}
\end{flalign}
in $\Ch_\bbK$, defined in \eqref{eqn:S-deg-def} below for all 
integers $n \geq 0$ and $k = 0, \ldots, n$. 
\end{subequations}
One defines the face map $d_0$ by 
\begin{subequations}\label{eqn:S-face-def}
\begin{flalign}
\xymatrix@C=5em{
S(\V)^{n+1} \ar@{-->}[r]^-{d_0} & S(\V)^n \\ 
\V(c_0) \ar[u]^-{\iota_{\und{c}}} \ar[r]_-{\V(\gamma_0)} & \V(c_1) \ar[u]_-{\iota_{\und{c} \circ \hat{0}}}
}
\end{flalign}
and the face map $d_k$, $k=1,\ldots,n+1$, by 
\begin{flalign}
\xymatrix@C=5em{
S(\V)^{n+1} \ar@{-->}[r]^-{d_k} & S(\V)^n \\ 
\V(c_0) \ar[u]^-{\iota_{\und{c}}} \ar[r]_-{\id} & \V(c_0) \ar[u]_-{\iota_{\und{c} \circ \hat{k}}}
}
\end{flalign}
\end{subequations}
Furthermore, one defines the degeneracy map $s_k$, $k=0,\ldots,n$, by 
\begin{flalign}\label{eqn:S-deg-def}
\xymatrix@C=5em{
S(\V)^{n} \ar@{-->}[r]^-{s_k} & S(\V)^{n+1} \\ 
\V(c_0) \ar[u]^-{\iota_{\und{c}}} \ar[r]_-{\id} & \V(c_0) \ar[u]_-{\iota_{\und{c} \circ \check{k}}}
}
\end{flalign}
The {\it homotopy colimit} 
\begin{flalign}\label{eqn:hocolim}
\hocolim(\V) \in \Ch_\bbK
\end{flalign}
of $\V\in\Ch_\bbK^{\CC}$
then may be computed by the $\bigoplus$-total complex associated 
with the simplicial cochain complex $S(\V) \in \Ch_\bbK^{\Delta^\op}$. Explicitly, 
$\hocolim(\V) \in \Ch_\bbK$ consists of the graded $\bbK$-vector space 
defined degree-wise for all $n \in \bbZ$ by 
\begin{flalign}
\hocolim(\V)^n := \bigoplus_{q \geq 0} S(\V)^{q,n+q} \quad,
\end{flalign}
where $-q \leq 0$ denotes the (cohomological) simplicial degree 
and $n+q \in \bbZ$ the cochain degree. The (total) differential 
\begin{subequations}
\begin{flalign}
\dd := - \dd_{\mathrm{h}} + \dd_{\mathrm{v}}
\end{flalign}
is the sum of (the opposite of) the horizontal differential 
$\dd_{\mathrm{h}}$, defined component-wise by 
\begin{flalign}
\dd_{\mathrm{h}} \circ \iota_0 := 0 \quad, \qquad \dd_{\mathrm{h}} \circ \iota_q := \sum_{k=0}^q (-1)^{-k}\, \iota_{q-1} \circ d_k \quad,
\end{flalign}
for all $q \geq 1$, and of the vertical differential $\dd_{\mathrm{v}}$, 
defined component-wise for all $q \geq 0$ and all $\und{c}: [q] \to \CC$ 
by 
\begin{flalign}
\dd_{\mathrm{v}} \circ \iota_{q,\und{c}} := (-1)^{-q}\, \iota_{q,\und{c}} \circ Q \quad, 
\end{flalign}
\end{subequations}
where $\iota_{q,\und{c}} := \iota_{q} \circ \iota_{\und{c}}$ 
denotes the inclusion of the $(q,\und{c})$-component $\V(c_0)^{n+q}$ 
of $\hocolim(\V)^n$ and $Q$ is the differential of 
$\V(c_0) \in \Ch_\bbK$. The signs in $\dd_{\mathrm{h}}$ 
can be understood including $\und{c}$-components. 
Since the $k$-th summand of $\dd_{\mathrm{h}}$ acts via $d_k$, 
$\dd_{\mathrm{h}}$ is pulled through the first $k$ morphisms in $\CC$
of the tuple $\und{c}$. Since each of those morphisms contributes $-1$ 
to the total degree and $\dd_{\mathrm{h}}$ has degree $1$, 
this gives rise to the sign $(-1)^{-k}$. 
Similarly, the sign in $\dd_{\mathrm{v}}$ arises from the fact that 
it acts by $Q$ after being pulled through $\und{c}$. Since the latter 
contributes $-q$ to the total degree and $\dd_{\mathrm{v}}$ 
has degree $1$, this gives rise to the sign $(-1)^{-q}$. The additional 
relative sign between the horizontal and vertical parts 
of the total differential $\dd$ is purely conventional and
chosen to ensure that the dg-adjunction from Proposition \ref{propo:dg-adjunction} below holds true. 
\begin{rem}\label{rem:hocolim-colim}
The homotopy colimit \eqref{eqn:hocolim} comes (as usual)
with a canonical natural transformation 
\begin{flalign}\label{eqn:hocolim-colim}
\hocolim \longrightarrow \colim
\end{flalign}
to the ordinary colimit, whose component 
at $\V \in \Ch_\bbK^\CC$ sends $\iota_{0,c_0} v \in \hocolim(\V)^n$, 
with $c_0 \in \CC$ and $v \in \V(c_0)^n$, 
to $\iota_{c_0} v \in \colim(\V)^n$ 
and $\iota_{q,\und{c}} v \in \hocolim(\V)^n$, 
with $q \geq 1$, $\und{c}: [q] \to \CC$ 
and $v \in \V(c_0)^{n+q}$, to $0 \in \colim(\V)^n$. 
(Here $\iota_{c_0}: \V(c_0) \to \colim(\V)$ in $\Ch_\bbK$, 
for all $c_0 \in \CC$, 
denote the canonical cochain maps from a diagram to its colimit.) 
Let us emphasize that \eqref{eqn:hocolim-colim} is furthermore 
a natural quasi-isomorphism when $\CC$ is a filtered category 
as in this case the ordinary colimit is a model for the 
homotopy colimit (by the AB5 axiom of Grothendieck Abelian 
categories). This situation will occur frequently 
in Sections \ref{sec:Green-hyp-cpx} and \ref{sec:witness}. 
\end{rem}

\paragraph*{The homotopy colimit as a dg-left adjoint:}
Let us extend the homotopy colimit to a dg-functor by defining, 
for all $\V, \W \in \Ch_\bbK^\CC$, 
its action on cochain complexes of morphisms
\begin{subequations}
\begin{flalign}\label{eqn:hocolimdgfunctorTMP}
\hocolim: \map(\V, \W) \longrightarrow [\hocolim(\V), \hocolim(\W)]
\end{flalign}
in $\Ch_\bbK$ through 
the adjunct cochain map (denoted with abuse of notation by the same symbol)
\begin{flalign}
\hocolim: \map(\V, \W) \otimes \hocolim(\V) \longrightarrow \hocolim(\W)
\end{flalign}
in $\Ch_\bbK$ that sends $\eta \in \map(\V,\W)^m$ and 
$\iota_{q,\und{c}} v \in \hocolim(\V)^n$, 
with $q \geq 0$, $\und{c}: [q] \to \CC$ and $v \in \V(c_0)^{n+q}$, to  
\begin{flalign}
\hocolim(\eta) (\iota_{q,\und{c}} v) := \sum_{k=0}^q (-1)^{-qm+k(q-k)}\, \iota_{q-k,\und{c}^{\geq k}} ((\pr_{k,\und{c}^{\leq k}} \eta)^{n+q}\, v) \in \hocolim(\W)^{m+n} \quad.
\end{flalign}
\end{subequations} 
Recalling that each morphism in $\CC$ contributes $-1$ 
to the total degree in the homotopy colimit, 
the sign of each summand can be understood by observing that 
$\eta$ is pulled through $\und{c}$ 
and furthermore $\und{c}^{\leq k}: [k] \to \CC$ is pulled through $\und{c}^{\geq k}: [q-k] \to \CC$. 
A straightforward check shows that this defines a dg-functor 
\begin{flalign}\label{eqn:hocolim-dg-fun}
\hocolim: \Ch_\bbK^\CC \longrightarrow \Ch_\bbK \quad.
\end{flalign}
Consider now also the diagonal dg-functor 
\begin{flalign}\label{eqn:diagonal-dg-fun}
\Delta: \Ch_\bbK \longrightarrow \Ch_\bbK^\CC
\end{flalign}
that sends a cochain complex $V \in \Ch_\bbK$ to the constant functor 
$\Delta V \in \Ch_\bbK^\CC$ and, for all $V,W \in \Ch_\bbK$, 
an $n$-cochain $f \in [V,W]^n$ to the $n$-cochain 
$\Delta f \in \map(\Delta V,\Delta W)^n$ 
defined by $\pr_{0,c_0} (\Delta f) := f$, for all $c_0 \in \CC$, 
and $\pr_{q,\und{c}} (\Delta f) := 0$, for all $q \geq 1$ and 
$\und{c}: [q] \to \CC$. Direct inspection shows that, for all $V \in \Ch_\bbK$ 
and $\V \in \Ch_\bbK^\CC$, the cochain map 
\begin{subequations}
\begin{flalign}\label{eqn:dg-adjunction}
\map(\V,\Delta V) \longrightarrow [\hocolim(\V), \hocolim(\Delta V)]\longrightarrow  [\hocolim(\V),V]
\end{flalign}
in $\Ch_\bbK$, which is obtained by composing \eqref{eqn:hocolimdgfunctorTMP} 
with the map induced by $\hocolim(\Delta V) \to \colim(\Delta V) = V$ (see \eqref{eqn:hocolim-colim}),
is an isomorphism that is natural with respect to both $\V$ and $V$ (in the dg-enriched sense). Explicitly, 
the previous cochain map is given by the adjunct of the cochain map 
\begin{flalign}
\map(\V,\Delta V) \otimes \hocolim(\V) \longrightarrow V
\end{flalign}
\end{subequations}
in $\Ch_\bbK$ that sends $\eta \in \map(\V,\Delta V)^m$ and $\iota_{q,\und{c}} v \in \hocolim(\V)^n$, 
with $q \geq 0$, $\und{c}: [q] \to \CC$ and $v \in \V(c_0)^{n+q}$, to 
$(-1)^{-qm}\, (\pr_{q,\und{c}} \eta)^{n+q}\, v \in V^{m+n}$. 
The isomorphism \eqref{eqn:dg-adjunction} will be used frequently 
throughout the rest of the paper to identify $\map(\V,\Delta V)$ 
and $[\hocolim(\V),V]$. This result is summarized below.  
\begin{propo}\label{propo:dg-adjunction}
The homotopy colimit dg-functor $\hocolim: \Ch_\bbK^\CC \to \Ch_\bbK$ 
from \eqref{eqn:hocolim-dg-fun} is dg-left adjoint 
to the diagonal dg-functor $\Delta: \Ch_\bbK \to \Ch_\bbK^\CC$ 
from \eqref{eqn:diagonal-dg-fun}. 
\end{propo}


\section{\label{sec:Green-hyp-cpx}Green hyperbolic complexes}
The purpose of this section is to introduce a generalization 
of Green hyperbolic linear differential operators, see \cite{Bar}, 
which we call {\it Green hyperbolic complexes}. 
Those play a central role in the study of derived critical loci 
of gauge-theoretic quadratic action functionals on globally hyperbolic Lorentzian manifolds, 
see \cite{linYM}. 
The core idea is based on the new concept of 
{\it retarded and advanced Green's homotopies}, 
which extend the well-known retarded and advanced Green's operators to a higher homological context. 
Examples \ref{ex:single-diff-op} and \ref{ex:Lambdapm=Gpm} clarify how one derives 
from Green hyperbolic linear differential operators a class 
of examples of Green hyperbolic complexes. 
In analogy with existence and uniqueness 
of retarded and advanced Green's operators for Green hyperbolic 
linear differential operators, Proposition 
\ref{propo:Greenhyp-unique} asserts that Green hyperbolic 
complexes admit unique (in the appropriate sense) 
retarded and advanced Green's homotopies. Furthermore, 
introducing the analog of the retarded-minus-advanced propagator, 
which we name {\it retarded-minus-advanced cochain map}, 
we shall prove in Theorem \ref{th:ret-minus-adv-qiso} 
a non-trivial generalization of the exact sequence \eqref{eqn:PGP-exact-seq} 
associated with any Green hyperbolic linear differential operator. 
Our theorem offers the following new interpretation of the exact sequence 
\eqref{eqn:PGP-exact-seq}: the latter witnesses the fact that 
the retarded-minus-advanced propagator 
establishes a quasi-isomorphism between 
(suitably shifted versions of) the complexes of sections 
with compact and respectively spacelike compact support, 
whose differential is just the original Green hyperbolic 
linear differential operator. 
\sk

For the rest of this section $M$ will denote a fixed oriented 
and time-oriented globally hyperbolic Lorentzian manifold 
of dimension $m \geq 2$. 
We work over the field $\bbK = \bbR$ of real or $\bbK = \bbC$ 
of complex numbers. 
We adopt the conventions of Section \ref{subsec:Green} 
for sections of vector bundles with restricted support. 
Furthermore, given a ($\bbZ$-)graded ($\bbK$-)vector bundle $F \to M$ 
(degree-wise of finite rank), 
\begin{flalign}\label{eqn:FFF}
\FFF^n := \Gamma(F^n)
\end{flalign}
will denote the vector space 
of degree $n$ smooth sections, i.e.\ smooth sections 
of the degree $n$ vector bundle $F^n \to M$. 
Similarly, $\FFF_{C}^n := \Gamma_{C}(F^n)$ 
and $\FFF_{\scrD}^n := \Gamma_{\scrD}(F^n)$ will denote 
the vector spaces of degree $n$ smooth sections 
with support contained in a closed subset $C \subseteq M$ 
and, respectively, with $\scrD$-support.

\subsection{\label{subsec:ret-adv-Green-homotopies}Retarded and advanced Green's homotopies}
\begin{defi}
A {\it complex of linear differential operators} on $M$ is a pair $(F,Q)$ 
consisting of a graded vector bundle $F \to M$ and of a collection 
$Q = (Q^n: \FFF^n \to \FFF^{n+1})_{n \in \bbZ}$ 
of degree increasing linear differential operators 
such that, for all $n \in \bbZ$, $Q^{n+1}\, Q^n = 0$. 
\end{defi}

\begin{ex}\label{ex:deRham}
The prime example of a complex of linear differential operators on $M$ 
is the de~Rham complex $(\Lambda^\bullet M,\dd_{\dR})$, 
which consists of the graded vector bundle $\Lambda^\bullet M$ 
of differential forms on $M$ 
and of the usual de~Rham differential $\dd_{\dR}$. 
Note that, for dimension $m=3$, shifting by $1$ the de~Rham complex provides 
the complex of linear differential operators 
$(F_{\CS},Q_{\CS}) = (\Lambda^\bullet M[1],\dd_{\dR\, [1]})$ 
associated with linear Chern-Simons theory. 
\end{ex}

\begin{ex}\label{ex:single-diff-op}
Another class of examples of complexes 
of linear differential operators arise from pairs $(E,P)$ consisting 
of a vector bundle $E \to M$ and of a linear differential operator $P$ 
acting on its sections. 
The associated complex of linear differential operators $(F_{(E,P)},Q_{(E,P)})$ 
consists of the graded vector bundle $F_{(E,P)} \to M$ 
concentrated in degrees $0$ and $1$ and 
defined by $F_{(E,P)}^n:=E$, for $n=0,1$, 
and of the collection of linear differential operators $Q_{(E,P)}$, 
whose only non-vanishing component is 
$Q_{(E,P)}^0:=P: \FFF_{(E,P)}^0 \to \FFF_{(E,P)}^1$. 
The linear differential operator 
$P = \Box + m^2: C^\infty(M) \to C^\infty(M)$, 
which governs the dynamics of a Klein-Gordon field of mass $m \geq 0$,
falls within this class. 
\end{ex}

\begin{ex}\label{ex:Maxwell-p-forms}
A richer example is provided by Maxwell $p$-forms on $M$, for $p \leq m-1$. 
Here the complex of linear differential operators $(F_{\MW},Q_{\MW})$ consists 
of the graded vector bundle $F_{\MW} \to M$ concentrated between degrees $-p$ 
and $p+1$ and defined by 
\begin{flalign}
F_{\MW}^n := 
\begin{cases}
\Lambda^{p+n} M\quad, & n=-p,\ldots,0\quad, \\
\Lambda^{p+1-n} M\quad, & n=1,\ldots,p+1\quad,
\end{cases}
\end{flalign}
where $\Lambda^k M \to M$ denotes the vector bundle 
of differential $k$-forms on $M$, 
and of the collection $Q_{\MW}$ of degree increasing linear differential operators, 
whose only non-vanishing components are 
\begin{flalign}
Q_{\MW}^n :=
\begin{cases}
\dd_{\dR} \quad, & n=-p,\ldots,-1 \quad, \\
\delta_{\dR}\, \dd_{\dR} \quad, & n=0 \quad, \\ 
\delta_{\dR} \quad, & n=1,\ldots,p \quad,
\end{cases}
\end{flalign}
where $\dd_{\dR}$ and $\delta_{\dR} := (-1)^k \ast^{-1} \dd_{\dR}\, \ast$ 
denote the de Rham differential and respectively codifferential 
on $k$-forms. 
(The latter is obtained using the Hodge star operator $\ast$, 
which is fixed by the metric and the orientation of $M$.) 
We observe that the underlying cochain complex $(\FFF_{\MW},Q_\MW)\in\Ch_\bbR$
reproduces the derived critical locus of linear Yang-Mills theory 
when $p=1$, see \cite{linYM}, and its higher generalizations 
for $p=2,\ldots,m-1$, see \cite{hrep}. 
\end{ex}

Recalling our conventions from Section \ref{subsec:Green}, 
we shall denote the complex of sections 
with support contained in a closed subset $C \subseteq M$ 
by $\FFF_C \in \Ch_\bbK$. 
When also a directed system $\scrD$ is considered, 
the cochain complexes $\FFF_D \in \Ch_\bbK$, for $D \in \scrD$, 
and their inclusions $\FFF_D \subseteq \FFF_{D^\prime}$, 
for $D \subseteq D^\prime \in \scrD$, define the functor 
\begin{flalign}\label{eqn:FFF(-)}
\FFF_{(-)} \in \Ch_\bbK^\scrD 
\end{flalign}
in an obvious way. Passing to the (homotopy) colimit provides 
the cochain complexes 
\begin{flalign}\label{eqn:FFF(-)colimit}
\FFF_{(\hh)\scrD} := \mathrm{(ho)}\colim \big( \FFF_{(-)}: \scrD \to \Ch_\bbK \big) \in \Ch_\bbK 
\end{flalign}
of $\scrD$-supported sections and the canonical quasi-isomorphism 
\begin{flalign}
\FFF_{\hh\scrD} \overset{\sim}{\longrightarrow} \FFF_{\scrD}
\end{flalign}
in $\Ch_\bbK$, see Remark \ref{rem:hocolim-colim}. 
For instance, when $\scrD = \cc$ is the directed system 
of compact subsets of $M$, 
one obtains the cochain complex $\FFF_{(\hh)\cc} \in \Ch_\bbK$ 
of compactly supported sections as the (homotopy) colimit of 
$\FFF_{(-)} \in \Ch_\bbK^\cc$. 
Since the causal future/past $J_M^\pm(K) \subseteq M$ 
of a compact subset $K \subseteq M$ is closed 
and forming the causal future/past preserves inclusions, 
one obtains order preserving maps $J_M^\pm: \cc \to \cl$ 
between directed sets. This allows us to define the functor 
\begin{flalign}\label{eqn:FFFJpm(-)}
\FFF_{J_M^\pm(-)} := \FFF_{(-)} \circ J_M^\pm \in \Ch_\bbK^\cc
\end{flalign} 
by composing the functor $\FFF_{(-)} \in \Ch_\bbK^\cl$ 
with the order preserving map $J_M^\pm$ regarded as a functor. 

\begin{defi}\label{defi:Green-hyp-cpx}
Let $(F,Q)$ be a complex of linear differential operators on $M$. 
\begin{itemize}
\item[(i)] A {\it retarded/advanced Green's homotopy} 
$\Lambda_\pm \in \map(\FFF_{J_M^\pm(-)},\FFF_{J_M^\pm(-)})^{-1}$ 
is a $(-1)$-cochain in the mapping complex \eqref{eqn:map} from 
the functor $\FFF_{J_M^\pm(-)} \in \Ch_\bbK^\cc$ to itself, 
whose differential $\delta \Lambda_\pm = \id$ is the identity 
natural transformation of the functor 
$\FFF_{J_M^\pm(-)} \in \Ch_\bbK^\cc$.

\item[(ii)] We say that $(F,Q)$ is a {\it Green hyperbolic complex}
if it admits a retarded and an advanced Green's homotopy $\Lambda_\pm$.
\end{itemize}
\end{defi}

\begin{ex}\label{ex:Lambdapm=Gpm}
Let us consider retarded/advanced Green's homotopies $\Lambda_\pm$ 
for the complex of linear differential operators $(F_{(E,P)},Q_{(E,P)})$ 
arising from a single linear differential operator $P$ acting 
on sections of a vector bundle $E \to M$, 
see Example \ref{ex:single-diff-op}. 
In this case, as we explain below, retarded/advanced Green's homotopies 
are in one-to-one correspondence with retarded/advanced Green's 
operators for $P$. Hence, $(F_{(E,P)},Q_{(E,P)})$ is a Green hyperbolic complex
if and only if $P$ is a Green hyperbolic linear differential operator.
\sk

Let us provide the relevant arguments. Since $F_{(E,P)}$ is concentrated in degrees 
$0$ and $1$, the only possibly non-vanishing components of a 
retarded/advanced Green's homotopy $\Lambda_\pm$ for 
$(F_{(E,P)},Q_{(E,P)})$ are the linear maps 
\begin{flalign}
(\pr_{0,K_0} \Lambda_\pm)^1: \FFF_{J_M^\pm(K_0)}^1 \longrightarrow \FFF_{J_M^\pm(K_0)}^{0} \quad,
\end{flalign} 
for all compact subsets $K_0 \subseteq M$. Recalling 
the definition \eqref{eqn:map-diff-def} of the mapping complex 
differential $\delta$, one finds that the condition 
$\delta \Lambda_\pm = \id$ has components of two types. 
The first type, involving only the horizontal differential 
$\delta_{\hh}$, is given by 
\begin{subequations}
\begin{flalign}
\FFF^0_{J_M^\pm(K_0 \subseteq K_1)} \circ (\pr_{0,K_0} \Lambda_{\pm})^1 - (\pr_{0,K_1} \Lambda_{\pm})^1 \circ \FFF^1_{J_M^\pm(K_0 \subseteq K_1)} = (\pr_{1,K_0 \subseteq K_1}(\delta_{\hh} \Lambda_\pm))^1 = 0 \quad,
\end{flalign}
for all inclusions $K_0 \subseteq K_1$ between compact subsets of $M$. 
The second type, involving only the vertical differential $\delta_{\vv}$, 
is given by 
\begin{flalign}
(\pr_{0,K_0} \Lambda_\pm)^{1} \circ P = (\pr_{0,K_0} \Lambda_\pm)^{1} \circ Q_{(E,P)}^0 = (\pr_{0,K_0}(\delta_{\vv} \Lambda_\pm))^0 = \id \quad, \\
P \circ (\pr_{0,K_0} \Lambda_\pm)^{1} = Q_{(E,P)}^0 \circ (\pr_{0,K_0} \Lambda_\pm)^{1} = (\pr_{0,K_0}(\delta_{\vv} \Lambda_\pm))^1 = \id \quad,
\end{flalign}
\end{subequations}
for all compact subsets $K_0 \subseteq M$. 
We summarize below all conditions in fully explicit form: 
\begin{enumerate}
\item[(a)] The linear maps $(\pr_{0,K_0} \Lambda_\pm)^1: \FFF_{J_M^\pm(K_0)}^1 \to \FFF_{J_M^\pm(K_0)}^{0}$ 
for all compact subsets $K_0 \subseteq M$ are compatible 
with inclusions $K_0 \subseteq K_1$ between compact subsets of $M$. 
In other words, these are the components of a natural transformation 
$(\pr_{0,(-)} \Lambda_\pm)^1: \FFF_{J_M^\pm(-)}^1 \to \FFF_{J_M^\pm(-)}^0$. 
\item[(b-c)] The linear maps 
$P \circ (\pr_{0,K_0} \Lambda_\pm)^{1} = \id: \FFF_{J_M^\pm(K_0)}^1 \to \FFF_{J_M^\pm(K_0)}^1$ 
and $(\pr_{0,K_0} \Lambda_\pm)^{1} \circ P = \id: \FFF_{J_M^\pm(K_0)}^0 \to \FFF_{J_M^\pm(K_0)}^0$ 
coincide for all compact subsets $K_0 \subseteq M$. 
\end{enumerate}
Recalling also the unique extensions of retarded/advanced Green's 
operators from \cite{Bar}, see also Section \ref{subsec:Green}, 
the datum of the natural transformation 
$(\pr_{0,(-)} \Lambda_\pm)^1: \FFF_{J_M^\pm(-)}^1 \to \FFF_{J_M^\pm(-)}^0$ 
from (a) subject to (b-c) is 
equivalent to the datum of a linear map $G_\pm: \Gamma_c(E) \to \Gamma(E)$ 
such that, for all $\varphi \in \Gamma_c(E)$, $\supp(G_\pm \varphi)$ is contained 
in the causal future/past $J_M^\pm(\supp(\varphi))$ of the support of $\varphi$
and $P\, G_\pm \varphi = \varphi = G_\pm\, P \varphi$. 
This means that the datum of a retarded/advanced Green's homotopy 
$\Lambda_\pm$ for $(F_{(E,P)},Q_{(E,P)})$ is equivalent to 
the datum of a retarded/advanced Green's operator 
$G_\pm: \Gamma_c(E) \to \Gamma(E)$ for $P$. 
\end{ex}

\begin{ex}
The complexes of linear differential operators from Examples 
\ref{ex:deRham} and \ref{ex:Maxwell-p-forms} are Green 
hyperbolic. We will prove this fact later in Section
\ref{sec:witness} by using the concept of {\em Green's witnesses}.
See in particular Examples \ref{ex:deRham-witness}, \ref{ex:Maxwell-p-forms-witness},
\ref{ex:Lambdapm-deRham} and \ref{ex:Lambdapm-Maxwell-p-forms}.
\end{ex}

\begin{rem}
Our approach to the definition of retarded and advanced Green's 
homotopies $\Lambda_\pm$ is purely algebraic, in the sense that 
we define $\Lambda_\pm$ as a collection of linear maps 
subject to suitable algebraic conditions. 
In doing so, we follow the approach to retarded and advanced 
Green's operators $G_\pm$ adopted in \cite{BGP,Bar}. 
Indeed, these references prove continuity (with respect to 
suitable topologies) of $G_\pm$ as a consequence 
of their algebraic definition. As a by-product, $G_\pm$ 
admit a presentation in terms of suitable distributional kernels. 
Retarded and advanced Green's homotopies are more abstract concepts, 
hence one does not expect that their algebraic definition already 
ensures continuity and presentation via distributional kernels. 
To achieve this goal one might, for instance, try to endow 
the relevant  complexes of sections with suitable additional 
structures, such as topologies. Unfortunately, topological vector spaces
are well-known to have a bad interplay with homological algebra (since they do not form
an Abelian category), 
while having a homologically well-behaved concept 
of retarded and advanced Green's homotopies is crucial to establish 
fundamental properties, such as uniqueness, 
see Proposition \ref{propo:Greenhyp-unique}. 
To circumvent this issue, we shall introduce in Section \ref{sec:witness} 
the concept of a Green's witness $W$, consisting of a collection of 
degree decreasing linear differential operators. 
Using $W$ one constructs retarded and advanced Green's homotopies 
$\Lambda_\pm := W\, G_\pm$ by composing ordinary retarded and advanced 
Green's operators with linear differential operators. As a consequence, 
$\Lambda_\pm$ turn out to be continuous and admit distributional kernels 
in the classical sense. 
Summing up: on the one hand, the algebraic definition 
of retarded and advanced Green's homotopies $\Lambda_\pm$ ensures 
their uniqueness in the sense of contractible spaces of choices; 
on the other hand, the presence of a Green's witness $W$ 
(which is very frequently available in concrete examples), 
ensures also that the spaces of choices contain 
analytically well-behaved points $\Lambda_\pm := W\, G_\pm$. 
\end{rem}

We have the following recognition principle for Green hyperbolic complexes.
\begin{propo}\label{propo:Greenhyp-acyclic}
A complex of linear differential operators $(F,Q)$ on $M$ is Green hyperbolic 
if and only if, for all compact subsets $K \subseteq M$, 
the cochain complexes $\FFF_{J_M^+(K)}$ and $\FFF_{J_M^-(K)}$ are both acyclic. 
\end{propo}
\begin{proof}
Given a retarded/advanced Green's homotopy $\Lambda_\pm$, 
for each compact subset $K_0 \subseteq M$, one has 
\begin{flalign}
\partial(\pr_{0,K_0} \Lambda_{\pm}) = \pr_{0,K_0}(\delta \Lambda_{\pm}) = \id \in [\FFF_{J_M^\pm(K_0)},\FFF_{J_M^\pm(K_0)}]^0 \quad.
\end{flalign}
This defines a contracting homotopy, hence the cochain complex $\FFF_{J_M^\pm(K_0)} \in \Ch_\bbK$ is acyclic. 
Vice versa, when the cochain complexes
$\FFF_{J_M^\pm(K)} \in \Ch_\bbK$ are acyclic for all 
compact subsets $K \subseteq M$, it follows that the natural 
transformation $\FFF_{J_M^\pm(-)} \to 0$ in $\Ch_\bbK^\cc$ 
is a weak equivalence. Since forming mapping complexes
$\map(-,\FFF_{J_M^\pm(-)}): (\Ch_\bbK^\cc)^\op \to \Ch_\bbK$ 
preserves weak equivalences, we obtain the quasi-isomorphism 
\begin{flalign}\label{eqn:map-acyclic}
0\cong  \map(0,\FFF_{J_M^\pm(-)}) \overset{\sim}{\longrightarrow} \map(\FFF_{J_M^\pm(-)},\FFF_{J_M^\pm(-)}) 
\end{flalign}
in $\Ch_\bbK$, meaning that the mapping complex $\map(\FFF_{J_M^\pm(-)},\FFF_{J_M^\pm(-)})$ is acyclic. 
It follows that the $0$-cocycles
$\id \in \map(\FFF_{J_M^\pm(-)},\FFF_{J_M^\pm(-)})^0$ 
must be exact, i.e.\ there exist $(-1)$-cochains 
$\Lambda_\pm \in \map(\FFF_{J_M^\pm(-)},\FFF_{J_M^\pm(-)})^{-1}$ 
such that $\delta \Lambda_\pm = \id$. These provide 
retarded and advanced Green's homotopies according to Definition \ref{defi:Green-hyp-cpx}.
\end{proof}

For a complex of linear differential operators $(F,Q)$ on $M$,
Definition \ref{defi:Green-hyp-cpx} actually provides a 
(possibly empty) set $\G\H_\pm$ of retarded/advanced 
Green's homotopies, given by the pullback 
\begin{flalign}
\xymatrix{
\G\H_\pm \ar@{-->}[r] \ar@{-->}[d] & \{\ast\} \ar[d]^-{\id} \\
\map \big( \FFF_{J_M^\pm(-)},\FFF_{J_M^\pm(-)} \big)^{-1} \ar[r]_-{\delta} & Z^0 \Big( \map \big( \FFF_{J_M^\pm(-)},\FFF_{J_M^\pm(-)} \big) \Big)
}
\end{flalign}
in $\Set$, where $\{\ast\} \in \Set$ 
denotes the singleton, the vertical arrow 
denotes the map that assigns the $0$-cocycle 
$\id \in Z^0(\map(\FFF_{J_M^\pm(-)},\FFF_{J_M^\pm(-)}))$ 
in the mapping complex and the horizontal arrow is given by the 
mapping complex differential $\delta$ acting on $(-1)$-cochains. 
The set $\G\H_\pm$ provides however only an insufficient picture
on the moduli problem of retarded/advanced Green's homotopies 
as it ignores homotopical phenomena. For instance, two non-identical retarded/advanced
Green's homotopies $\Lambda_{\pm}$ and $\Lambda_{\pm}^\prime$ 
that differ by an exact term $\Lambda_{\pm} - \Lambda_{\pm}^\prime = \delta\lambda_\pm$, for
some $\lambda_\pm\in \map(\FFF_{J_M^\pm(-)},\FFF_{J_M^\pm(-)})^{-2}$,
should be considered as ``being the same'', as they define the same
cohomology class $[\Lambda_{\pm}]=[\Lambda_{\pm}^\prime]$,
but they describe different elements in the set $\G\H_\pm$.
This issue can be solved by upgrading the set $\G\H_\pm$ to a space 
(Kan complex) through the following standard construction. 
Denoting the normalized chains functor 
by $N: \sSet \to \Ch_\bbK$ 
and the simplicial $n$-simplex 
by $\Delta^n \in \sSet$, we define the simplicial set 
$\und{\G\H}_\pm \in \sSet$ as the pullback 
\begin{flalign}\label{eqn:GHpm}
\xymatrix{
\und{\G\H}_\pm \ar@{-->}[r] \ar@{-->}[d] & \{\ast\} \ar[d]^-{\id} \\
\big[ N(\Delta^\bullet), \map \big( \FFF_{J_M^\pm(-)}, \FFF_{J_M^\pm(-)} \big) \big]^{-1} \ar[r]_-{\partial} & Z^0 \big( \big[ N(\Delta^\bullet), \map \big( \FFF_{J_M^\pm(-)},\FFF_{J_M^\pm(-)} \big) \big] \big)
}
\end{flalign}
in $\sSet$, where $\{\ast\} \in \sSet$ 
denotes the (constant) simplicial set with only one point, 
the vertical arrow denotes the simplicial map 
that assigns the $0$-cocycle 
$\id \in Z^0([N(\Delta^\bullet),\map(\FFF_{J_M^\pm(-)},\FFF_{J_M^\pm(-)})])$ 
that takes the constant value $\id$ in the mapping complex 
and the horizontal arrow is given by the internal hom differential 
$\partial$ acting on $(-1)$-cochains. 
The set of $0$-simplices $(\und{\G\H}_\pm)_0 = \G\H_\pm$
coincides with the set of retarded/advanced Green's homotopies
and the higher simplices encode the desired homotopical 
phenomena mentioned above. One realizes that, when non-empty, the simplicial set 
$\und{\G\H}_\pm$ is affine over the simplicial vector space 
\begin{flalign}\label{eqn:GHpm-vector}
\und{GH}_\pm := Z^{-1} \big( \big[ N(\Delta^\bullet), \map \big( \FFF_{J_M^\pm(-)}, \FFF_{J_M^\pm(-)} \big) \big] \big) \in \sVec_\bbK \quad, 
\end{flalign} 
with the affine action 
\begin{flalign}
\und{\G\H}_\pm \times \und{GH}_\pm \longrightarrow \und{\G\H}_\pm
\end{flalign}
in $\sSet$ that sends $n$-simplices 
$\rho_{\pm} \in (\und{\G\H}_\pm)_n$ and 
$\eta_\pm \in (\und{GH}_\pm)_n$ to their sum 
$\rho_\pm + \eta_\pm \in (\und{\G\H}_\pm)_n$. 
As a consequence of its affine structure, 
the simplicial set of retarded/advanced Green's homotopies 
$\und{\G\H}_\pm \in \sSet$ is a Kan complex, 
see e.g.\ \cite[Lem.~8.2.8]{Weibel}.  
These preliminaries allow us to prove a uniqueness result 
for retarded and advanced Green's homotopies. 

\begin{propo}\label{propo:Greenhyp-unique}
The Kan complex of retarded/advanced Green's homotopies 
$\und{\G\H}_\pm \in \sSet$ from \eqref{eqn:GHpm} is either empty 
or contractible. In particular, when it exists, a retarded/advanced Green's homotopy 
is unique up to the contractible space of choices $\und{\G\H}_\pm$. 
\end{propo}
\begin{proof}
Suppose that $\und{\G\H}_\pm \in \sSet$ is non-empty.
We already observed that the Kan complex of retarded/advanced Green's 
homotopies $\und{\G\H}_\pm \in \sSet$ from \eqref{eqn:GHpm} is affine 
over the simplicial vector space $\und{GH}_\pm \in \sVec_\bbK$ 
from \eqref{eqn:GHpm-vector}. By definition of the latter, it follows
that $\und{GH}_\pm \cong \Gamma(\tau^{\leq0}(\map(\FFF_{J_M^\pm(-)},\FFF_{J_M^\pm(-)})[-1])) \in \sVec_\bbK$ 
is isomorphic to the simplicial vector space assigned by 
the Dold-Kan correspondence  $\Gamma: \Ch_\bbK^{\leq0} \to \sVec_\bbK$ 
(see \cite[Sec.~4.1]{SchwedeShipley} for a concise overview)
to the good truncation $\tau^{\leq0}: \Ch_\bbK \to \Ch_\bbK^{\leq0}$ 
of the $(-1)$-shifted mapping complex 
$\map(\FFF_{J_M^\pm(-)},\FFF_{J_M^\pm(-)})[-1] \in \Ch_\bbK$. 
As $\und{\G\H}_\pm$ is by hypothesis non-empty, there exists a retarded/advanced 
Green's homotopy $\Lambda_\pm$, hence it follows from Proposition \ref{propo:Greenhyp-acyclic}
that the mapping complex $\map(\FFF_{J_M^\pm(-)},\FFF_{J_M^\pm(-)}) \in \Ch_\bbK$ is acyclic. 
Since both $\tau^{\leq0}$ and $\Gamma$ preserve weak equivalences, 
it follows that $\und{GH}_\pm \in \sVec_\bbK$ is contractible, 
hence $\und{\G\H}_\pm \in \sSet$ is contractible too. 
\end{proof}

\begin{rem}
We would like to note that all definitions and results stated so far in this section admit 
a straightforward generalization to the category $\Loc_m$ of $m$-dimensional 
oriented and time-oriented globally hyperbolic Lorentzian manifolds $M$ 
with morphisms $f: M \to M^\prime$ given by orientation 
and time-orientation preserving isometric embeddings whose image 
$f(M) \subseteq M^\prime$ is open and causally convex. 
Instead of a single complex of linear differential operators $(F,Q)$ on $M$, 
consider a family $(F(M),Q_M)_{M\in\Loc_m}^{}$ of complexes of linear differential operators
that is natural with respect to $M \in \Loc_m$. Then one can upgrade 
the functors $\FFF_{J_M^\pm} \in \Ch_\bbK^\cc$ 
to let $M \in \Loc_m$ vary. Explicitly, introduce 
the category $\LocC_m$, whose objects $(M,K)$ consist of 
an object $M \in \Loc_m$ and of a compact subset $K \subseteq M$ 
and whose morphisms $f: (M,K) \to (M^\prime,K^\prime)$ 
consist of a morphism $f: M \to M^\prime$ in $\Loc_m$ such that 
$K^\prime \subseteq f(K)$. Define the functor 
$\FFF_{J^\pm}: \LocC_m^\op \to \Ch_\bbK$ that assigns to an object 
$(M,K) \in \LocC_m^\op$ the cochain complex 
$\FFF_{J^\pm_M(K)} \in \Ch_\bbK$ of sections of the graded vector 
bundle $F(M) \to M$ with support contained in $J_M^\pm(K)$
and to a morphism $f: (M,K) \to (M^\prime,K^\prime)$ in $\LocC_m$ 
the pullback $\FFF_{J^\pm}(f): \FFF_{J_{M^\prime}^\pm(K^\prime)} \to \FFF_{J_{M}^\pm(K)}$ 
in $\Ch_\bbK$ of sections along $f: M \to M^\prime$ in $\Loc_m$. 
Replacing $(F,Q)$ with $(F(M),Q_M)_{M\in\Loc_m}^{}$, $\cc$ with $\LocC_m^\op$ 
and $\FFF_{J_M^\pm} \in \Ch_\bbK^{\cc}$ with 
$\FFF_{J^\pm} \in \Ch_\bbK^{\LocC_m^\op}$ in Definition 
\ref{defi:Green-hyp-cpx}, one obtains a notion of
{\it homotopy coherent $\Loc_m$-natural} retarded and advanced Green's homotopies.
With the same substitutions, the analogs of Propositions 
\ref{propo:Greenhyp-acyclic} and \ref{propo:Greenhyp-unique} hold true. In fact, these 
proofs rely only on the structure of the mapping complex and, 
in particular, on the fact that it preserves weak equivalences, 
but they are not sensitive to the shape of the category indexing the 
$\Ch_\bbK$-valued functors that appear. 
\end{rem}

\subsection{\label{subsec:ret-adv-map}Retarded-minus-advanced quasi-isomorphism}
Each Green hyperbolic linear differential operator $P$ has an associated 
exact sequence \eqref{eqn:PGP-exact-seq} formed by $P$ itself 
and its retarded-minus-advanced propagator $G := G_+ - G_-$, 
which involves sections both with compact and with spacelike compact supports, 
see \cite{BeniniDappiaggi, Bar}. 
The goal of this section is to upgrade
the exact sequence \eqref{eqn:PGP-exact-seq} 
to the broader context of Green hyperbolic complexes. 
More precisely, Theorem \ref{th:ret-minus-adv-qiso} below 
will show that the retarded-minus-advanced cochain map 
$\Lambda_{\hh}: \FFF_{\hh\cc}[1] \to \FFF_{\hh\sc}$ 
$\in \Ch_\bbK$, see Definition \ref{defi:ret-minus-adv}, 
is a quasi-isomorphism from the ($1$-shift of the) complex of 
sections with compact support 
to the complex of sections with spacelike compact support. 
Remark \ref{rem:exact=qiso}, which will appear later in Section \ref{sec:witness},
will clarify how the well-known exact sequence \eqref{eqn:PGP-exact-seq} 
is just a special case of Theorem \ref{th:ret-minus-adv-qiso}. 
\sk

In preparation for the next definition of retarded-minus-advanced 
cochain map, let us introduce the following notation. 
For an inclusion $C_1 \subseteq C_2$ of closed subsets of $M$, 
we denote by 
\begin{flalign}
j_{C_1}^{C_2}: \FFF_{C_1} \overset{\subseteq}{\longrightarrow} \FFF_{C_2}
\end{flalign}
in $\Ch_\bbK$ the associated inclusion of the complexes 
of sections with support in $C_i$, $i=1,2$. 
In particular, associated with 
a natural transformation $\kappa: F_1 \to F_2$ 
between functors $F_i: \cc \to \cl$, $i=1,2$, sending compact subsets 
to closed subsets of $M$, one has a natural transformation 
\begin{flalign}\label{eqn:j}
j_{F_1(-)}^{F_2(-)}: \FFF_{F_1(-)} \longrightarrow \FFF_{F_2(-)}
\end{flalign}
in $\Ch_\bbK^\cc$. 
We shall frequently encounter instances of this natural 
transformation associated with the natural inclusions 
$K \subseteq J_M^\pm(K) \subseteq J_M(K) \subseteq M$, 
for all compact subsets $K \subseteq M$, which follow from 
the definition of causal future/past. 
\sk

The next definition involves the homotopy colimit dg-functor 
$\hocolim: \Ch_\bbK^\cc \to \Ch_\bbK$ from Section \ref{subsec:FunCh}. 
\begin{defi}\label{defi:ret-minus-adv}
Let $(F,Q)$ be a Green hyperbolic complex on $M$ and consider 
a choice of retarded and advanced Green homotopies $\Lambda_\pm$. 
The associated {\it retarded-minus-advanced cochain map} 
\begin{subequations}\label{eqn:ret-minus-adv}
\begin{flalign}
\Lambda_{\hh} := \hocolim(\Lambda): \FFF_{\hh\cc}[1] \longrightarrow \FFF_{\hh\sc}
\end{flalign}
in $\Ch_\bbK$ is defined evaluating the dg-functor $\hocolim$ 
from \eqref{eqn:hocolim-dg-fun} on the $(-1)$-cocycle 
\begin{flalign}
\Lambda := j_{J^+_M(-)}^{J_M(-)} \circ \Lambda_+ \circ j_{(-)}^{J^+_M(-)} - j_{J^-_M(-)}^{J_M(-)} \circ \Lambda_- \circ j_{(-)}^{J^-_M(-)} \in Z^{-1} \big(\map \big( \FFF_{(-)},\FFF_{J_M(-)} \big)\big)
\end{flalign}
in the mapping complex. (Note that $\delta \Lambda = 0$ follows from 
$\delta \Lambda_\pm = \id$.) 
\end{subequations}
\end{defi}

\begin{rem} 
The cochain complex $\FFF_{(\hh)\sc} \in \Ch_\bbK$ is defined 
in \eqref{eqn:FFF(-)colimit} as the (homotopy) colimit 
of the functor $\FFF_{(-)} \in \Ch_\bbK^{\sc}$ 
over the directed system $\scrD = \sc$ of spacelike compact 
closed subsets of $M$. Note, however, that the functor 
$J_M: \cc \to \sc$ is (homotopy) final, 
hence $\FFF_{(\hh)\sc} \in \Ch_\bbK$ 
is (quasi-)isomorphic to the (homotopy) colimit of the functor 
$\FFF_{J_M(-)} = \FFF_{(-)} \circ J_M \in \Ch_\bbK^\cc$. 
Bearing this fact in mind, we shall always implicitly identify 
$\FFF_{(\hh)\sc} \in \Ch_\bbK$ with the (homotopy) 
colimit of $\FFF_{J_M(-)} \in \Ch_\bbK^\cc$. For instance, 
this identification is implicit in Definition \ref{defi:ret-minus-adv}. 
\end{rem}

The main result of this section is the following theorem.
\begin{theo}\label{th:ret-minus-adv-qiso}
Let $(F,Q)$ be a Green hyperbolic complex on $M$. 
Then the retarded-minus-advanced cochain map
$\Lambda_{\hh}: \FFF_{\hh\cc}[1] \to \FFF_{\hh\sc}$ in $\Ch_\bbK$ 
from Definition \ref{defi:ret-minus-adv} is a quasi-isomorphism. 
\end{theo}
\begin{proof}
The proof is rather lengthy and it will be split into four steps. The first step develops 
a geometric construction, which plays a crucial role 
in the next steps, the second step defines a candidate 
quasi-inverse $\Theta$ to $\Lambda_\hh$, the third step 
constructs a homotopy $\Xi$ witnessing 
$\Theta \circ \Lambda_\hh \sim \id$ and the fourth step 
constructs a homotopy $\Upsilon$ witnessing 
$\Lambda_\hh \circ \Theta \sim \id$. 

\paragraph{Geometric construction:}
Choose two spacelike Cauchy surfaces 
$\Sigma_\pm \subseteq M$ such that $\Sigma_+ \subseteq I_M^+(\Sigma_-)$ 
lies in the chronological future of $\Sigma_-$ 
and consider an order preserving map 
\begin{flalign}\label{eqn:geom-constr}
\Sigma: \cc \longrightarrow \cc
\end{flalign}
on the directed set $\cc$ of compact subsets in $M$ 
such that, for each compact subset $K\subseteq M$,
\begin{flalign}\label{eqn:sigma}
J_M^\pm(\Sigma_\mp) \cap J_M^\mp(K) \subseteq \Sigma(K) \quad. \tag{$\Sigma$}
\end{flalign}
Note that property~\eqref{eqn:sigma} implies in particular that $K\subseteq \Sigma(K)$,
for each compact subset $K\subseteq M$.
A (minimal) concrete example for such $\Sigma: \cc \to \cc$ 
is given by setting, for each compact subset $K\subseteq M$,
$\Sigma(K) := (J_M^+(\Sigma_-) \cap J_M^-(K)) \cup (J_M^-(\Sigma_+) \cap J_M^+(K))$.

\paragraph{Quasi-inverse $\Theta$ of $\Lambda_\hh$:}
Choosing in addition to the data from the previous paragraph
a partition of unity $\{\chi_+,\chi_-\}$ subordinate to the open cover 
$\{I_M^+(\Sigma_-),I_M^-(\Sigma_+)\}$ of $M$, 
one constructs a candidate quasi-inverse $\Theta$ to the 
retarded-minus-advanced cochain map $\Lambda_\hh$. 
Explicitly, the cochain map 
\begin{subequations}\label{eqn:Theta}
\begin{flalign}
\Theta: \FFF_{\hh\sc} \longrightarrow \FFF_{\hh\cc}[1]
\end{flalign}
in $\Ch_\bbK$ will be determined uniquely by 
\begin{flalign}
\hocolim \Big( j_{(-)}^{J_M(-)} \Big) \circ \Theta := \pm \partial \theta_\pm \in [\FFF_{\hh\sc},\FFF_{\hh\sc}]^1 \quad, 
\end{flalign}
\end{subequations}
with $\theta_\pm$ defined in \eqref{eqn:theta-pm} below. 
Uniqueness of $\Theta$, provided it exists, is a consequence
of $\hocolim\big(j_{(-)}^{J_M(-)}\big): \FFF_{\hh\cc} \to \FFF_{\hh\sc}$ 
in $\Ch_\bbK$ being degree-wise injective. Recalling the dg-adjunction $\hocolim \dashv \Delta$ 
from Proposition \ref{propo:dg-adjunction}, the $0$-cochain
\begin{subequations}\label{eqn:theta-pm}
\begin{flalign}
\theta_\pm \in \big[ \FFF_{\hh\sc},\FFF_{\hh\sc} \big]^0 \cong \map \big( \FFF_{J_M(-)},\Delta \FFF_{\hh\sc} \big)^0 
\end{flalign}
in the mapping complex is defined by assigning its components 
$\pr_{q,\und{K}} \theta_\pm \in [\FFF_{J_M(K_0)},\FFF_{\hh\sc}]^{-q}$,
for all $q \geq 0$ and $\und{K}: [q] \to \cc$, according to 
\begin{flalign}
(\pr_{q,\und{K}} \theta_\pm)^n \varphi := \iota_{q,\Sigma(\und{K})} (\chi_\pm \varphi) \in \FFF_{\hh\sc}^{n-q} \quad,
\end{flalign}
\end{subequations}
for all $n \in \bbZ$ and $\varphi \in \FFF_{J_M(K_0)}^n$. 
Note that the equation displayed above involves the order 
preserving map $\Sigma$ from \eqref{eqn:geom-constr} 
(regarded here as a functor). 
In particular, property~\eqref{eqn:sigma} entails that $K_0\subseteq \Sigma(K_0)$ and hence
that $\chi_\pm \varphi$ is supported in $J_M(K_0) \subseteq J_M(\Sigma(K_0))$. 
Recalling from Section \ref{subsec:FunCh} the differential 
$\delta = \delta_{\hh} + \delta_{\vv}$ on the mapping complex 
and the differential $\dd = - \dd_{\hh} + \dd_{\vv}$ 
on the homotopy colimit, direct inspection shows that 
\begin{flalign}
\big( \pr_{q,\und{K}} (\delta \theta_\pm) \big)^n \varphi = \iota_{q,\Sigma(\und{K})} \big( Q (\chi_\pm \varphi) - \chi_\pm Q \varphi \big) \quad, 
\end{flalign}
for all $q \geq 0$, $\und{K}: [q] \to \cc$, $n \in \bbZ$ and 
$\varphi \in \FFF_{J_M(K_0)}^n$. Because $Q$ consists of 
differential operators and $\chi_+ + \chi_- = 1$, it follows that 
$Q(\chi_+ \varphi) - \chi_+ Q \varphi = -(Q(\chi_- \varphi) - \chi_- Q \varphi))$, 
hence the latter is a section supported in 
$J_M^+(\Sigma_-) \cap J_M^-(\Sigma_+) \cap J_M(K_0) \subseteq \Sigma(K_0)$, 
where the inclusion follows from property~\eqref{eqn:sigma}. 
In particular, this proves that 
$\partial \theta_+ = - \partial \theta_- \in [\FFF_{\hh\sc},\FFF_{\hh\sc}]^1$ 
factors through the inclusion $\hocolim(j_{(-)}^{J_M(-)})$, 
ensuring that $\Theta$ as defined by \eqref{eqn:Theta} exists. 

\begin{rem}
Let us emphasize that the construction of the cochain map 
$\Theta: \FFF_{\hh\sc} \to \FFF_{\hh\cc}[1]$ in $\Ch_\bbK$ 
makes sense regardless of $(F,Q)$ being Green hyperbolic 
and relies only on the causal structure of $M$ 
(through the choices of the two spacelike Cauchy surfaces 
$\Sigma_\pm$, of the order preserving map $\Sigma$ 
and of the partition of unity $\{\chi_+,\chi_-\}$). 
\end{rem}

\paragraph{Homotopy $\Xi$ witnessing $\Theta \circ \Lambda_\hh \sim \id$:} 
We shall now construct a homotopy $\Xi$ witnessing that 
$\Theta$ is a left quasi-inverse of $\Lambda_\hh$, 
i.e.\ $\partial \Xi = \id - \Theta \circ \Lambda_\hh$. 
Explicitly, pulling the shifts out of the internal hom displayed below 
(the corresponding isomorphism contributes a sign $(-1)^n$ 
in degree $n$) and then using again the 
dg-adjunction $\hocolim \dashv \Delta$ 
from Proposition \ref{propo:dg-adjunction}, the $(-1)$-cochain 
\begin{subequations}\label{eqn:Xi}
\begin{flalign}
\Xi \in \big[ \FFF_{\hh\cc}[1],\FFF_{\hh\cc}[1] \big]^{-1} \cong \map \big( \FFF_{(-)},\Delta \FFF_{\hh\cc} \big)^{-1}
\end{flalign}
in the mapping complex is defined by 
\begin{flalign}
\Xi := \xi_- \circ \Lambda_+ \circ j_{(-)}^{J_M^+(-)} + \xi_+ \circ \Lambda_- \circ j_{(-)}^{J_M^-(-)} + \xi \quad, 
\end{flalign}
\end{subequations}
where $\Lambda_\pm$ are the retarded and advanced Green's homotopies 
chosen in Definition \ref{defi:ret-minus-adv} 
and $\xi_\mp$ and $\xi$ are defined in \eqref{eqn:ximp} 
and respectively in \eqref{eqn:xi} below. The $0$-cochains
\begin{subequations}\label{eqn:ximp}
\begin{flalign}
\xi_\mp \in \map \big( \FFF_{J_M^\pm(-)},\Delta \FFF_{\hh\cc} \big)^0
\end{flalign}
in the mapping complex are defined by assigning the components 
$\pr_{q,\und{K}} \xi_\mp \in [\FFF_{J_M^\pm(K_0)},\FFF_{\hh\cc}]^{-q}$, 
for all $q \geq 0$ and $\und{K}: [q] \to \cc$, according to
\begin{flalign}
(\pr_{q,\und{K}} \xi_\mp)^n \varphi := \iota_{q,\Sigma(\und{K})} (\chi_\mp \varphi) \in \FFF_{\hh\cc}^{n-q} \quad, 
\end{flalign}
\end{subequations}
for all $n \in \bbZ$ and $\varphi \in \FFF_{J_M^\pm(K_0)}^n$. 
Note that the order preserving map $\Sigma$ from \eqref{eqn:geom-constr} 
enters this construction. In particular, $\chi_\mp \varphi$ is supported 
in $J_M^{\mp}(\Sigma_\pm) \cap J_M^\pm(K_0) \subseteq \Sigma(K_0)$ 
because of property~\eqref{eqn:sigma}. Direct inspection shows that 
\begin{flalign}\label{eqn:delta-xi-mp}
\mp \delta \xi_\mp = \Theta \circ j_{J_M^\pm(-)}^{J_M(-)} \in \map \big( \FFF_{J_M^\pm(-)},\Delta \FFF_{\hh\cc} \big)^1 \quad, 
\end{flalign}
with $\Theta$ from \eqref{eqn:Theta} regarded as a $1$-cocycle 
$\Theta \in \map(\FFF_{J_M(-)},\Delta \FFF_{\hh\cc})^1$ 
in the mapping complex. The $(-1)$-cochain
\begin{subequations}\label{eqn:xi}
\begin{flalign}
\xi \in \map \big( \FFF_{(-)},\Delta \FFF_{\hh\cc} \big)^{-1}
\end{flalign}
in the mapping complex is defined by assigning the components 
$\pr_{q,\und{K}} \xi \in [\FFF_{K_0},\FFF_{\hh\cc}]^{-q-1}$, 
for all $q \geq 0$ and $\und{K}: [q] \to \cc$, according to
\begin{flalign}
(\pr_{q,\und{K}} \xi)^n \varphi := \sum_{k=0}^q (-1)^{q-k} \iota_{q+1,\und{K}^{\leq k} \subseteq \Sigma(\und{K}^{\geq k})} \varphi \in \FFF_{\hh\cc}^{n-q-1} \quad,
\end{flalign}
\end{subequations}
for all $n \in \bbZ$ and $\varphi \in \FFF_{K_0}^n$. 
The last formula may be understood as the sum over 
all paths of length $q+1$ in the commutative diagram 
$\und{K} \subseteq \Sigma(\und{K})$ in $\cc$ 
arising from $\und{K}: [q] \to \cc$. 
The sign of the $k$-th summand is interpreted as arising 
from pulling the morphism $K_k \subseteq \Sigma(K_k)$ in $\cc$ 
through $\und{K}^{\geq k}: [q-k] \to \cc$ 
and recalling that each morphism in $\cc$ contributes $-1$ 
to the total degree in the homotopy colimit. 
Direct inspection shows that 
\begin{flalign}\label{eqn:delta-xi}
\delta \xi = \eta - \xi_- \circ j_{(-)}^{J_M^+(-)} - \xi_+ \circ j_{(-)}^{J_M^-(-)} \in \map \big( \FFF_{(-)},\Delta \FFF_{\hh\cc} \big)^{0} \quad,
\end{flalign}
where $\eta: \id \to \Delta \circ \hocolim$ denotes the unit 
of the dg-adjunction $\hocolim \dashv \Delta$. 
Combining \eqref{eqn:delta-xi-mp} and \eqref{eqn:delta-xi}, one finds 
\begin{flalign}
\partial \Xi = \id - \Theta \circ \Lambda_\hh \in \big[ \FFF_{\hh\cc}[1],\FFF_{\hh\cc}[1] \big]^{0} \quad. 
\end{flalign} 

\paragraph{Homotopy $\Upsilon$ witnessing ${\Lambda_\hh \circ \Theta} \sim \id$:}
We shall now construct a homotopy $\Upsilon \in [\FFF_{\hh\sc},\FFF_{\hh\sc}]^{-1}$
witnessing that $\Theta$ is a right quasi-inverse of $\Lambda_\hh$, 
i.e.\ $\partial \Upsilon = \id - \Lambda_\hh \circ \Theta$. 
Explicitly, consider the $(-1)$-cochain
\begin{flalign}\label{eqn:Upsilon}
\Upsilon := \Big( j_{J_M^+(-)}^{J_M(-)} \Big)_\hh \circ (\Lambda_+)_\hh \circ \upsilon_+ + \Big( j_{J_M^-(-)}^{J_M(-)} \Big)_\hh \circ (\Lambda_-)_\hh \circ \upsilon_- + \upsilon \in [\FFF_{\hh\sc},\FFF_{\hh\sc}]^{-1}
\end{flalign}
in the internal hom, where $(-)_\hh := \hocolim$ denotes the homotopy 
colimit dg-functor and $\upsilon_\pm$ and $\upsilon$ are defined 
in \eqref{eqn:upsilonpm} and respectively \eqref{eqn:upsilon} below. 
Recalling from \cite{Bar} the strictly past compact $\spc$ 
and strictly future compact $\sfc$ directed systems, the $0$-cochain 
\begin{subequations}\label{eqn:upsilonpm}
\begin{flalign}
\upsilon_\pm \in \big[ \FFF_{\hh\sc},\FFF_{\hh\substack{\spc\\ \sfc}} \big]^0
\end{flalign}
in the internal hom is uniquely determined by 
\begin{flalign}
\Big( j_{J_M^\pm(-)}^{J_M(-)} \Big)_\hh \circ \upsilon_\pm := \theta_\pm \in \big[ \FFF_{\hh\sc},\FFF_{\hh\sc} \big]^0 \quad, 
\end{flalign}
\end{subequations}
where $\theta_\pm$ is defined in \eqref{eqn:theta-pm}. 
Indeed, $\Big( j_{J_M^\pm(-)}^{J_M(-)} \Big)_\hh$ is degree-wise injective, 
hence $\upsilon_\pm$ is unique, provided it exists. 
Furthermore, property~\eqref{eqn:sigma} of the chosen order preserving map $\Sigma: \cc \to \cc$
entails that $\chi_\pm \varphi$ is supported in 
$J_M^\pm(\Sigma_\mp) \cap J_M(K_0) \subseteq J_M^\pm(\Sigma(K_0))$ 
when $\varphi$ is supported in $J_M(K_0)$, 
hence $\theta_\pm$ factors through 
$\Big( j_{J_M^\pm(-)}^{J_M(-)} \Big)_\hh$ and
$\upsilon_\pm$ as defined by \eqref{eqn:upsilonpm} exists. 
Factoring out $\Big( j_{J_M^\pm(-)}^{J_M(-)} \Big)_\hh$, 
it follows from \eqref{eqn:Theta} that 
\begin{flalign}\label{eqn:partial-upsilon-pm}
\pm \partial \upsilon_\pm = \Big( j_{(-)}^{J_M^\pm(-)} \Big)_\hh \circ \Theta \in \big[ \FFF_{\hh\sc},\FFF_{\hh\substack{\spc\\ \sfc}} \big]^1 \quad.
\end{flalign}
The $(-1)$-cochain 
\begin{flalign}\label{eqn:upsilon}
\upsilon \in \big[ \FFF_{\hh\sc},\FFF_{\hh\sc} \big]^{-1} \cong \map \big( \FFF_{J_M(-)},\Delta \FFF_{\hh\sc} \big)^{-1}
\end{flalign}
in the mapping complex is defined by the same formula as \eqref{eqn:xi} 
(however here $\varphi \in \FFF_{J_M(K_0)}^n$). 
In particular, in a similar fashion one finds 
\begin{flalign}\label{eqn:partial-upsilon}
\partial \upsilon = \id - \Big( j_{J_M^+(-)}^{J_M(-)} \Big)_\hh \circ \upsilon_+ - \Big( j_{J_M^-(-)}^{J_M(-)} \Big)_\hh \circ \upsilon_- \in \big[ \FFF_{\hh\sc},\FFF_{\hh\sc} \big]^{0} \quad.
\end{flalign}
Combining \eqref{eqn:partial-upsilon-pm} and \eqref{eqn:partial-upsilon}, 
one finds 
\begin{flalign}
\partial \Upsilon = \id - \Lambda_\hh \circ \Theta \quad. 
\end{flalign}
This completes the proof of Theorem \ref{th:ret-minus-adv-qiso}.
\end{proof}

\subsection{\label{subsec:Poisson}Poisson structures from differential pairings}
Working over $\bbK = \bbR$, this section defines a suitable concept of {\it differential pairing} 
for complexes of linear differential operators. 
This concept generalizes the usual fiber metrics on vector bundles 
and, at the same time, encodes the structure 
that makes Stokes' theorem available upon integration. 
In analogy with the classical situation of a formally self-adjoint 
Green hyperbolic operator, see e.g.\ \cite{BGP}, it shall be shown that 
endowing a complex of linear differential operators with such a 
differential pairing determines two types of Poisson structures. 
The first type relies crucially on $(F,Q)$ being a 
Green hyperbolic complex and consists of three Poisson structures 
$\tau_M^{\pm}, \tau_M: \FFF_{\hh\cc}[1]^{\wedge 2} \to \bbR$ 
in $\Ch_\bbR$, defined on the $1$-shift of the complex of sections 
with compact support and related to each other by suitable homotopies, 
reflecting the multiple equivalent ways 
\eqref{eqn:standardtaudifferent} of presenting 
the classical Poisson structure \eqref{eqn:standardtau}. 
The second type of Poisson structure
$\sigma_\Sigma: \FFF_{\hh\sc}^{\wedge 2} \to \bbR$ in $\Ch_\bbR$ 
relies on the choice of a spacelike Cauchy surface $\Sigma \subseteq M$ 
and is defined on the complex of sections with spacelike compact support. 
\begin{defi}\label{defi:diff-pairing}
Let $(F,Q)$ be a complex of linear differential operators on $M$. 
A {\it differential pairing} $(-,-)$ on $(F,Q)$ is a 
graded anti-symmetric linear bi-differential operator 
$(-,-): \FFF^{\otimes 2} \to \Omega^\bullet(M)[m-1]$ 
that is a cochain map with respect to the differentials
$Q_\otimes$ and $\dd_{\dR\,[m-1]}$, i.e.\
\begin{flalign}
(Q\varphi_1,\varphi_2) + (-1)^{\vert \varphi_1\vert}\,(\varphi_1,Q\varphi_2) 
= (-1)^{m-1}\,\dd_{\dR}(\varphi_1,\varphi_2)\quad,
\end{flalign}
for all homogeneous sections $\varphi_1,\varphi_2\in\FFF$.
\end{defi}

\begin{ex}\label{ex:CS-diff-pairing}
For $m=3$, recall the complex of linear differential operators 
$(F_{\CS},Q_{\CS})$ associated with linear Chern-Simons theory 
from Example \ref{ex:deRham}. 
On $(F_{\CS},Q_{\CS})$ one can consider the differential pairing 
$(-,-)_{\CS}$, whose only non-vanishing components are given by 
\begin{flalign}\label{eq:CS-diff-pairing}
(\varepsilon_1,\varepsilon_2)_{\CS} &:= \varepsilon_1 \wedge \varepsilon_2 \quad, \nn \\ 
(A,\varepsilon)_{\CS} &:= - A \wedge \varepsilon =: - (\varepsilon,A)_{\CS} \quad, \nn \\ 
(A^\ddagger,\varepsilon)_{\CS} &:= A^\ddagger \wedge \varepsilon =: (\varepsilon,A^\ddagger)_{\CS} \quad,\nn \\ 
(A_1,A_2)_{\CS} &:= - A_1 \wedge A_2 \quad, \nn\\
(\varepsilon^\ddagger,\varepsilon)_{\CS} &:= - \varepsilon^\ddagger \wedge \varepsilon =: - (\varepsilon,\varepsilon^\ddagger)_{\CS} \quad,\nn \\
(A^\ddagger,A)_{\CS} &:= A^\ddagger \wedge A =: - (A,A^\ddagger)_{\CS} \quad,
\end{flalign}
for all $\varepsilon, \varepsilon_1, \varepsilon_2 \in \FFF_{\CS}^{-1} = \Omega^0(M)$, 
$A, A_1, A_2 \in \FFF_{\CS}^{0} = \Omega^1(M)$, 
$A^\ddagger \in \FFF_{\CS}^{1} = \Omega^2(M)$ 
and $\varepsilon^\ddagger \in \FFF_{\CS}^{2} = \Omega^3(M)$. 
Note that $(-,-)_{\CS}$ is just an appropriately shifted version 
of the graded commutative multiplication 
$\wedge: \Omega^\bullet(M) \otimes \Omega^\bullet(M) \to \Omega^\bullet(M)$. 
This shift is the source of the signs displayed in 
\eqref{eq:CS-diff-pairing}. 
\end{ex}

\begin{ex}\label{ex:single-diff-op-diff-pairing}
Recall the complex of linear differential operators 
$(F_{(E,P)},Q_{(E,P)})$ from Example \ref{ex:single-diff-op}. 
Additionally, assume the vector bundle $E$ comes endowed 
with a fiber metric $\langle-,-\rangle$ and suppose that 
$P$ is of the form $P = \Box^\nabla + B$, 
for $\nabla$ a connection on $E$ and $B$ an endomorphism of $E$. 
(Note that all normally hyperbolic linear differential operators 
are of this form, see \cite[Lem.~1.5.5]{BGP}.) 
Furthermore, suppose that the connection $\nabla$ is metric 
and that the endomorphism $B$ is symmetric. With these preparations 
one endows $(F_{(E,P)},Q_{(E,P)})$ with a differential pairing 
$(-,-)_{(E,P)}$, whose only non-vanishing components are given by 
\begin{flalign}
(\varphi_1,\varphi_2)_{(E,P)} &:= (-1)^m \langle \varphi_1 \wedge \ast \nabla \varphi_2 - \varphi_2 \wedge \ast \nabla \varphi_1 \rangle \quad, \nn \\
(\varphi^\ddagger,\varphi)_{(E,P)} &:= \langle \varphi^\ddagger \wedge \ast \varphi \rangle =: - (\varphi,\varphi^\ddagger)_{(E,P)} \quad, 
\end{flalign}
for all $\varphi, \varphi_1, \varphi_2 \in \FFF_{(E,P)}^0 = \Gamma(E)$ 
and $\varphi^\ddagger \in \FFF_{(E,P)}^1 = \Gamma(E)$. 
Compatibility with the differentials follows from $\nabla$ being 
a metric connection and $B$ being a symmetric endomorphism. 
The construction above applies manifestly to the Klein-Gordon field. 
\end{ex}

\begin{ex}\label{ex:MW-diff-pairing}
Also the complex of linear differentials operators $(F_{\MW},Q_{\MW})$ 
associated with Maxwell $p$-forms from Example \ref{ex:Maxwell-p-forms} 
can be endowed with a differential pairing $(-,-)_{\MW}$. 
For instance, when $p=1$ (corresponding to the more 
familiar electromagnetic vector potential), 
the only non-vanishing components of $(-,-)_{\MW}$ are given by 
\begin{flalign}
(A,\varepsilon)_{\MW} &:= - \varepsilon \wedge \ast \dd_\dR A =: - (\varepsilon,A)_{\MW} \quad, \nn \\ 
(A^\ddagger,\varepsilon)_{\MW} &:= (-1)^m \varepsilon \wedge \ast A^\ddagger =: (\varepsilon,A^\ddagger)_{\MW} \quad, \nn \\ 
(\varepsilon^\ddagger,\varepsilon)_{\MW} &:= \varepsilon^\ddagger \wedge \ast \varepsilon =: - (\varepsilon,\varepsilon^\ddagger)_{\MW} \quad, \nn\\
 (A_1,A_2)_{\MW} &:= (-1)^{m-1} \big(A_1 \wedge \ast \dd_\dR A_2 - A_2 \wedge \ast \dd_\dR A_1\big) \quad,\nn\\
 (A^\ddagger,A)_{\MW} &:= A^\ddagger \wedge \ast A =: - (A,A^\ddagger)_{\MW} \quad,
\end{flalign}
for all $\varepsilon \in \FFF_{\MW}^{-1} = \Omega^0(M)$, 
$A, A_1, A_2 \in \FFF_{\MW}^{0} = \Omega^1(M)$, 
$A^\ddagger \in \FFF_{\MW}^{1} = \Omega^1(M)$ 
and $\varepsilon^\ddagger \in \FFF_{\MW}^{2} = \Omega^0(M)$. 
\end{ex}

Let $(-,-)$ be a differential pairing 
on a complex of linear differential operators $(F,Q)$ on $M$.  
The key ingredient to construct the first type of Poisson structures  is the evaluation pairing defined by the composition 
\begin{flalign}\label{eqn:ev-M}
\xymatrix@C=5em{
\ev_M : \FFF_{\cc}[1] \otimes \FFF \ar[r]^-{\id_{\bbR[1]} \otimes (-,-)} & \Omega_{\cc}^\bullet(M)[m] \ar[r]^-{\int_M} & \bbR
}
\end{flalign}
in $\Ch_\bbR$. The construction above uses that, being a 
bi-differential operator, $(-,-)$ preserves supports 
and that, by Stokes' theorem, integration over $M$ defines a cochain map 
$\int_M: \Omega_{\cc}^\bullet(M)[m] \to \bbR$ in $\Ch_\bbR$. 
(Recall that $m = \dim(M)$ is the dimension of $M$.) 
\sk

Let us construct the Poisson structure $\tau_M^\pm$, 
generalizing the last equivalent presentation in 
\eqref{eqn:standardtaudifferent} 
of the classical Poisson structure \eqref{eqn:standardtau}. 
Given a retarded/advanced Green's homotopy 
$\Lambda_\pm$ for $(F,Q)$ and recalling also \eqref{eqn:j}, 
define the $0$-cochain 
\begin{subequations}
\begin{flalign}
\Lambda_{\pm\,\hh} \in [\FFF_{\hh\cc}[1],\FFF]^{0} \cong [\FFF_{\hh\cc},\FFF]^{-1}
\end{flalign} 
as the adjunct of the $(-1)$-cochain 
\begin{flalign}
\hocolim \Big( j_{J^\pm_M(-)}^{M} \circ \Lambda_\pm \circ j_{(-)}^{J^\pm_M(-)} \Big) \in \map(\FFF_{(-)},\Delta \FFF)^{-1}
\end{flalign}
\end{subequations}
with respect to the dg-adjunction $\hocolim \dashv \Delta$ from 
Proposition \ref{propo:dg-adjunction}. 
With these preparations we define the cochain map 
\begin{flalign}\label{eqn:tau-pm}
\tau_M^\pm := \pm \ev_M \circ (\id \otimes \Lambda_{\pm\,\hh}) \mp \ev_M \circ \gamma \circ (\Lambda_{\pm\,\hh} \otimes \id): \FFF_{\hh\cc}[1]^{\otimes 2} \longrightarrow \bbR 
\end{flalign}
in $\Ch_\bbR$, where $\gamma$ denotes the symmetric braiding on 
$\Ch_\bbR$ and the $1$-shift of the $\FFF_{(-)}$-component 
of the natural transformation $\hocolim \to \colim$ 
from \eqref{eqn:hocolim-colim} is suppressed from our notation. 
To check that $\tau_M^\pm$ as defined above is a cochain map, 
for all homogeneous 
$\varphi_1, \varphi_2 \in \FFF_{\hh\cc}[1]$, compute 
\begin{flalign}
(\partial \tau_M^\pm) (\varphi_1 \otimes \varphi_2) &= \big( \pm \ev_M (\id \otimes \partial \Lambda_{\pm\,\hh}) \mp \ev_M \circ \gamma \circ (\partial \Lambda_{\pm\,\hh} \otimes \id) \big) (\varphi_1 \otimes \varphi_2) \nn \\
&= \pm (-1)^{|\varphi_1|} \int_M (\varphi_1,\varphi_2) \mp (-1)^{(|\varphi_1| + 1) |\varphi_2|} \int_M(\varphi_2,\varphi_1) \nn \\
&= 0 \quad, 
\end{flalign}
where in the first step we used that both $\ev_M$ and $\gamma$ 
are cochain maps, in the second step we used that the $1$-cochain 
$\partial \Lambda_{\pm\,\hh} \in [\FFF_{\hh\cc}[1],\FFF]^{1} \cong [\FFF_{\hh\cc},\FFF]^{0}$ 
is just the adjunct of the inclusion 
$\FFF_{(-)} \to \Delta \FFF$ in $\Ch_\bbR^{\cc}$ 
with respect to the dg-adjunction $\hocolim \dashv \Delta$ 
from Proposition \ref{propo:dg-adjunction} 
and in the last step we used the graded anti-symmetry of $(-,-)$. 
Since $\tau_M^\pm$ is manifestly graded anti-symmetric, 
it descends to a Poisson structure 
$\tau_M^\pm: \FFF_{\hh\cc}[1]^{\wedge 2} \to \bbR$. 
\sk

Let us also construct the Poisson structure $\tau_M$. 
Given a choice of retarded and advanced Green's homotopies 
$\Lambda_\pm$ for $(F,Q)$, recall the associated retarded-minus-advanced 
cochain map $\Lambda_\hh$ from Definition \ref{defi:ret-minus-adv} 
and consider the composition
\begin{flalign}\label{eqn:tautildemap}
\widetilde{\tau}_M &:= \ev_M \circ (\id \otimes \Lambda_\hh): \FFF_{\hh\cc}[1]^{\otimes 2} \longrightarrow \bbR 
\end{flalign}
in $\Ch_\bbR$, where we suppressed from our notation 
both the $1$-shift of the $\FFF_{(-)}$-component of the natural 
transformation $\hocolim \to \colim$ from \eqref{eqn:hocolim-colim} 
and the adjunct of the inclusion 
$\FFF_{J_M(-)} \to \Delta \FFF$ in $\Ch_\bbR^{\cc}$ 
with respect to the dg-adjunction $\hocolim \dashv \Delta$ 
from Proposition \ref{propo:dg-adjunction}. 
Graded anti-symmetrization defines the cochain map 
\begin{flalign}\label{eqn:taumap}
\tau_M := \asym(\widetilde{\tau}_M): \FFF_{\hh\cc}[1]^{\otimes 2} \longrightarrow \bbR
\end{flalign}
in $\Ch_\bbR$, which by construction descends to a Poisson structure 
$\tau_M : \FFF_{\hh\cc}[1]^{\wedge 2}\to\bbR$ on $\FFF_{\hh\cc}[1]$. 
\begin{rem}\label{rem:taumap}
Graded anti-symmetrization is required because we did not impose any compatibility conditions between 
the Green's homotopies and the differential pairing. Hence, the cochain map $\widetilde{\tau}_M$ 
in \eqref{eqn:tautildemap} will in general fail to be graded anti-symmetric.
For a large class of examples, including all the 
examples presented in this paper, there exist particular choices of retarded/advanced Green's homotopies
that are compatible with the differential pairing and thereby
turn $\widetilde{\tau}_M$ into a graded anti-symmetric cochain map.
Such choices make the graded anti-symmetrization construction in \eqref{eqn:taumap} superfluous. 
The details will be discussed in Section \ref{subsec:self-adj}.
\end{rem}

The Poisson structures $\tau_M^\pm, \tau_M$ are related by 
homotopies that are constructed by the graded anti-symmetrization 
\begin{subequations}\label{eqn:lambda-M}
\begin{flalign}
\lambda_M := \asym(\widetilde{\lambda}_M) \in [\FFF_{\hh\cc}[1]^{\otimes 2},\bbR]^{-1} 
\end{flalign}
of the $(-1)$-cochain 
\begin{flalign}
\widetilde{\lambda}_M \in [\FFF_{\hh\cc}[1]^{\otimes 2},\bbR]^{-1}
\end{flalign} 
defined below. By abuse of notation, we shall denote 
by $\Lambda_{\pm\,\hh} := \hocolim(\Lambda_\pm)$ 
also the homotopy colimit of the chosen retarded/advanced Green's 
homotopy $\Lambda_\pm$, interpreting 
$\Lambda_{\pm\,\hh} \in [\FFF_{\hh\substack{\spc\\ \sfc}}[1],\FFF_{\hh\substack{\spc\\ \sfc}}]^0$ 
as a $0$-cochain in the internal hom involving cochain complexes 
$\FFF_{\hh\substack{\spc\\ \sfc}}$ of sections with strictly past/future compact support. 
Furthermore, we shall suppress from our notation
all components of the natural transformation $\hocolim \to \colim$ 
from \eqref{eqn:hocolim-colim}, 
as well as the adjuncts of the inclusions $\FFF_{(-)} \to \Delta \FFF$, 
$\FFF_{J_M(-)} \to \Delta \FFF$ and 
$\FFF_{J^\pm_M(-)} \to \Delta \FFF$ in $\Ch_\bbR^{\cc}$ 
with respect to the dg-adjunction $\hocolim \dashv \Delta$ 
from Proposition \ref{propo:dg-adjunction}. 
Recalling that the intersection of a strictly past compact subset 
of $M$ with a strictly future compact one is compact 
and that the differential pairing $(-,-)$ preserves supports, 
we define 
$\widetilde{\lambda}_M \in [\FFF_{\hh\cc}[1]^{\otimes 2},\bbR]^{-1}$ by 
\begin{flalign}
\widetilde{\lambda}_M (\varphi_1 \otimes \varphi_2) := - \int_M (\Lambda_{+\,\hh} \varphi_1, \Lambda_{-\,\hh} \varphi_2) \quad, 
\end{flalign}
\end{subequations}
for all homogeneous $\varphi_1, \varphi_2 \in \FFF_{\hh\cc}[1]$. 
Direct inspection shows that 
\begin{flalign}
\partial \widetilde{\lambda}_M = \tau_M^+ - \widetilde{\tau}_M = (\tau_M^- - \widetilde{\tau}_M) \circ \gamma \quad. 
\end{flalign}
Since their verifications are very similar, 
let us focus on the first equality only. 
For all homogeneous $\varphi_1, \varphi_2 \in \FFF_{\hh\cc}[1]$, 
one computes 
\begin{flalign}
(\partial \widetilde{\lambda}_M) (\varphi_1 \otimes \varphi_2) &= - \int_M (\Lambda_{+\,\hh} \dd_{[1]} \varphi_1, \Lambda_{-\,\hh} \varphi_2) - (-1)^{|\varphi_1|} \int_M (\Lambda_{+\,\hh} \varphi_1, \Lambda_{-\,\hh} \dd_{[1]} \varphi_2) \nn \\ 
&= - \int_M (\dd \Lambda_{+\,\hh} \varphi_1, \Lambda_{-\,\hh} \varphi_2) + \int_M (\varphi_1, \Lambda_{-\,\hh} \varphi_2) \nn \\ 
&\qquad\phantom{=} - (-1)^{|\varphi_1|} \int_M (\Lambda_{+\,\hh} \varphi_1, \dd \Lambda_{-\,\hh} \varphi_2) + (-1)^{|\varphi_1|} \int_M (\Lambda_{+\,\hh} \varphi_1, \varphi_2) \nn \\ 
&= - \int_M \dd_{\dR\, [m-1]} (\Lambda_{+\,\hh} \varphi_1, \Lambda_{-\,\hh} \varphi_2) \nn \\ 
&\qquad\phantom{=} + \int_M (\varphi_1, \Lambda_{-\,\hh} \varphi_2) - (-1)^{|\varphi_1| |\varphi_2|} \int_M (\varphi_2, \Lambda_{+\,\hh} \varphi_1) \nn \\ 
&= (\tau_M^+ - \widetilde{\tau}_M) (\varphi_1 \otimes \varphi_2) \quad. 
\end{flalign}
For the first step we used 
$\partial(\widetilde{\lambda}_M) = \widetilde{\lambda}_M \circ \dd_{[1] \otimes}$, 
for the second step we used $\dd \circ \Lambda_{\pm\,\hh} - \Lambda_{\pm\,\hh} \circ \dd_{[1]} = \partial (\Lambda_{\pm\,\hh}) = \id$, 
for the third step we used 
$\Lambda_{\hh} = \Lambda_{+\,\hh} - \Lambda_{-\,\hh}$, 
the graded anti-symmetry of $(-,-)$ and 
its compatibility with the differentials 
$(-,-) \circ Q_{\otimes} = \dd_{\dR\,[m-1]} \circ (-,-)$ 
(the passage from $\dd$ to $Q$ is suppressed here because 
the natural transformation $\hocolim \to \colim$ 
from \eqref{eqn:hocolim-colim} has been suppressed from our notation). 
For the last step we used Stokes' theorem and the definitions 
of $\tau_M^\pm$ in \eqref{eqn:tau-pm} and of 
$\widetilde{\tau}_M$ in \eqref{eqn:tautildemap}. 
Since $\lambda_M$ is by definition anti-symmetrized, 
it descends to the homotopy 
$\lambda_M \in [\FFF_{\hh\cc}[1]^{\wedge 2},\bbR]^{-1}$ 
such that 
\begin{flalign}\label{eqn:partial-lambda-M}
\partial \lambda_M = \asym(\partial \widetilde{\lambda}_M) = \tau_M^+ - \tau_M = \tau_M - \tau_M^- \quad. 
\end{flalign}

We collect our findings so far in the proposition stated below. 
\begin{propo}\label{propo:covariant-Poisson}
Let $(F,Q)$ be a Green hyperbolic complex on $M$ endowed with a differential pairing $(-,-)$.
For a choice of retarded and advanced Green's homotopies $\Lambda_{\pm}$, 
\eqref{eqn:tau-pm} and \eqref{eqn:taumap} define Poisson structures 
$\tau_M^\pm, \tau_M: \FFF_{\hh\cc}[1]^{\wedge 2} \to \bbR$ on $\FFF_{\hh\cc}[1]$. 
These Poisson structures coincide 
$\tau_M^\pm = \tau_M \pm \partial \lambda_M$ up to the homotopy 
$\pm \lambda_M \in [\FFF_{\hh\cc}[1]^{\wedge 2},\bbR]^{-1}$ 
from \eqref{eqn:lambda-M}. 
\end{propo}

\begin{rem}\label{rem:tau}
The cochain map $\tau_M^\pm$ \eqref{eqn:tau-pm} generalizes 
the classical Poisson structure \eqref{eqn:standardtau} 
as presented by the last step of \eqref{eqn:standardtaudifferent}, 
while $\tau_M$ generalizes the classical Poisson structure 
as presented by the second and third steps 
of \eqref{eqn:standardtaudifferent}. 
In contrast to the classical situation, in general 
the Poisson structures $\tau_M^\pm,\tau_M$ 
do not coincide, but are related by the homotopy $\pm \lambda_M$, 
as stated by Proposition \ref{propo:covariant-Poisson}. 
For a large class of examples, including all the 
examples presented in this paper, there exist particular choices of retarded/advanced Green's homotopies
that are compatible with the differential pairing. 
In this case it turns out that $\tau_M^+ = \tau_M^- = \tau_M$ 
coincide. We will come back to this in more detail 
in Section \ref{subsec:self-adj}. 
\end{rem}

To construct the second type of Poisson structure on a 
complex of linear differential operators $(F,Q)$ endowed 
with a differential pairing $(-,-)$, the key ingredient is a second 
type of evaluation pairing, which depends on the choice of a 
spacelike Cauchy surface $\Sigma \subseteq M$ 
and is defined by the composition 
\begin{flalign}\label{eqn:ev-Sigma}
\xymatrix@C=4em{
\ev_\Sigma : \FFF_{\sc} \otimes \FFF \ar[r]^-{(-,-)} & \Omega_{\sc}^\bullet(M)[m-1] \ar[r]^-{\iota^\ast} & \Omega_{\cc}^\bullet(\Sigma)[m-1] \ar[r]^-{\int_\Sigma} & \bbR
}
\end{flalign}
in $\Ch_\bbR$. The construction above uses that, being a 
bi-differential operator, $(-,-)$ preserves supports, 
that intersecting  a spacelike compact subset of $M$ with 
a spacelike Cauchy surface $\Sigma$ returns a compact subset of $\Sigma$ 
and that, by Stokes' theorem, integration over $\Sigma$ defines a cochain map 
$\int_\Sigma: \Omega_{\cc}^\bullet(\Sigma)[m-1] \to \bbR$ in $\Ch_\bbR$. 
(Recall that $\dim(\Sigma) = \dim(M) -1 = m-1$.) 
The construction of the second Poisson structure is summarized in
the proposition below, which is a straightforward
consequence of the graded anti-symmetry of the differential pairing $(-,-)$. 
\begin{propo}\label{propo:time-zero-Poisson}
Let $(F,Q)$ be a complex of linear differential operators on $M$
endowed with a differential pairing $(-,-)$. Denote
by $\Sigma\subseteq M$ a spacelike Cauchy surface of $M$.
Then the composition
\begin{flalign}\label{eqn:sigmamap}
\xymatrix@C=5em{
\sigma_\Sigma: \FFF_{\hh\sc}^{\otimes 2} \ar[r] & \FFF_{\sc} \otimes \FFF \ar[r]^-{(-1)^{m-1} \ev_\Sigma^{}} & \bbR 
}
\end{flalign}
in $\Ch_\bbR$ is graded anti-symmetric and hence it descends
to a Poisson structure $\sigma_\Sigma : \FFF_{\hh\sc}^{\wedge 2}\to\bbR$ on $\FFF_{\hh\sc}$. 
The unlabeled morphism consists in its first factor 
of the $\FFF_{J_M(-)}$-component of the natural transformation 
$\hocolim \to \colim$ from \eqref{eqn:hocolim-colim} 
and in its second factor of the adjunct of the inclusion 
$\FFF_{J_M(-)} \to \Delta \FFF$ in $\Ch_\bbR^{\cc}$ 
with respect to the dg-adjunction $\hocolim \dashv \Delta$ 
from Proposition \ref{propo:dg-adjunction}. 
\end{propo}

Classically, the retarded-minus-advanced propagator 
is compatible with the usual analogs of the two types 
of Poisson structures considered above, 
hence it is an isomorphism of Poisson vector spaces. 
Their cochain complex analogs are Poisson complexes, i.e.\ 
pairs $(V,\tau)$ consisting of a cochain 
complex $V \in \Ch_\bbR$ and a Poisson structure 
$\tau: V^{\wedge 2} \to \bbR$ in $\Ch_\bbR$. 
Morphisms of Poisson complexes $f: (V_1,\tau_1) \to (V_2,\tau_2)$ 
are cochain maps $f: V_1 \to V_2$ such that 
$\tau_2 \circ f^{\wedge 2} = \tau_1$. 
It will be shown that, in the context of Green hyperbolic complexes, 
the retarded-minus-advanced quasi-isomorphism $\Lambda_{\hh}$ is 
{\it not} a morphism from the Poisson complexes 
$(\FFF_{\hh\cc}[1],\tau_M^\pm)$ or $(\FFF_{\hh\cc}[1],\tau_M)$ of Proposition 
\ref{propo:covariant-Poisson} to the Poisson complex 
$(\FFF_{\hh\sc},\sigma_\Sigma)$ of Proposition 
\ref{propo:time-zero-Poisson}, however this failure 
is controlled by prescribed homotopies. 
Such weaker morphisms can be interpreted
as morphisms in a suitable simplicial category 
of Poisson complexes, see \cite[Sec.~3.1]{GwilliamHaugseng}.
\sk

For concreteness, let us concentrate on constructing 
the homotopy $\lambda_{\hh}$ such that 
$\sigma_\Sigma \circ \Lambda_{\hh}^{\wedge 2} = \tau_M + \partial \lambda_{\hh}$. 
(One can easily relate $\sigma_\Sigma$ also to the Poisson structure 
$\tau_M^\pm$ from \eqref{eqn:tau-pm} by combining $\lambda_\hh$ 
with the homotopy $\lambda_M$ from \eqref{eqn:lambda-M} 
according to 
$\sigma_\Sigma \circ \Lambda_{\hh}^{\wedge 2} = \tau_M^{\pm} + \partial (\lambda_{\hh} \mp \lambda_M)$.) 
The relevant 
\begin{subequations}\label{eqn:lambdah}
\begin{flalign}
\lambda_{\hh} := \asym(\widetilde{\lambda}_{\hh}) \in \big[ \FFF_{\hh\cc}[1]^{\wedge 2}, \bbR \big]^{-1}
\end{flalign}
is the graded anti-symmetrization of the $(-1)$-cochain 
\begin{flalign}
\widetilde{\lambda}_{\hh} \in \big[ \FFF_{\hh\cc}[1]^{\otimes 2}, \bbR \big]^{-1}
\end{flalign}
defined, for all homogeneous 
$\varphi_1, \varphi_2 \in \FFF_{\hh\cc}[1]$, by 
\begin{flalign}
\widetilde{\lambda}_{\hh} (\varphi_1\otimes\varphi_2) := \int_{\Sigma^+} \big( (\Lambda_-)_{\hh} \varphi_1, \Lambda_{\hh} \varphi_2 \big) + \int_{\Sigma^-} \big( (\Lambda_+)_{\hh} \varphi_1, \Lambda_{\hh} \varphi_2 \big) \quad, 
\end{flalign}
\end{subequations}
where on the right-hand side $\Sigma^\pm := J_M^\pm(\Sigma)$, 
$(-)_{\hh} := \hocolim$ denotes the homotopy colimit 
dg-functor \eqref{eqn:hocolim-dg-fun} 
and the relevant components of the natural transformation 
$\hocolim \to \colim$ from \eqref{eqn:hocolim-colim} 
are suppressed from our notation. 
Indeed, direct inspection shows that, for all homogeneous sections 
$\varphi_1, \varphi_2 \in \FFF_{\hh\cc}[1]$, one has 
\begin{flalign}\label{eqn:lambdah-computation}
(\partial \widetilde{\lambda}_{\hh})(\varphi_1\otimes\varphi_2) &= \int_{\Sigma^+} \big( (\Lambda_-)_{\hh} \dd_{[1]} \varphi_1, \Lambda_{\hh} \varphi_2 \big) + (-1)^{|\varphi_1|} \int_{\Sigma^+} \big( (\Lambda_-)_{\hh} \varphi_1, \Lambda_{\hh} \dd_{[1]} \varphi_2 \big) \nn \\ 
&\quad \hphantom{=} + \int_{\Sigma^-} \big( (\Lambda_+)_{\hh} \dd_{[1]} \varphi_1, \Lambda_{\hh} \varphi_2 \big) + (-1)^{|\varphi_1|} \int_{\Sigma^-} \big( (\Lambda_+)_{\hh} \varphi_1, \Lambda_{\hh} \dd_{[1]} \varphi_2 \big) \nn \\
&= \int_{\Sigma^+} \big( \dd (\Lambda_-)_{\hh} \varphi_1, \Lambda_{\hh} \varphi_2 \big) - \int_{\Sigma^+} \big( \varphi_1, \Lambda_{\hh} \varphi_2 \big) + (-1)^{|\varphi_1|} \int_{\Sigma^+} \big( (\Lambda_-)_{\hh} \varphi_1, \dd \Lambda_{\hh} \varphi_2 \big) \nn \\ 
&\quad \hphantom{=} + \int_{\Sigma^-} \big( \dd (\Lambda_+)_{\hh} \varphi_1, \Lambda_{\hh} \varphi_2 \big) - \int_{\Sigma^-} \big( \varphi_1, \Lambda_{\hh} \varphi_2 \big) + (-1)^{|\varphi_1|} \int_{\Sigma^-} \big( (\Lambda_+)_{\hh} \varphi_1, \dd \Lambda_{\hh} \varphi_2 \big) \nn \\ 
&= \int_{\Sigma^+} \dd_{\dR\, [m-1]} \big( (\Lambda_-)_{\hh} \varphi_1, \Lambda_{\hh} \varphi_2 \big) + \int_{\Sigma^-} \dd_{\dR\, [m-1]} \big( (\Lambda_+)_{\hh} \varphi_1, \Lambda_{\hh} \varphi_2 \big) - \int_M \big( \varphi_1, \Lambda_\hh \varphi_2 \big) \nn \\ 
&= (-1)^{m-1} \int_\Sigma \iota^\ast \big( \Lambda_\hh \varphi_1, \Lambda_\hh \varphi_2 \big) - \int_M \big( \varphi_1, \Lambda_\hh \varphi_2 \big)\nn\\
&= \sigma_{\Sigma}(\Lambda_{\hh}\varphi_1\otimes \Lambda_\hh\varphi_2) - \widetilde{\tau}_M(\varphi_1\otimes \varphi_2) \quad, 
\end{flalign}
where $\Sigma^\pm := J_M^\pm(\Sigma)$. For the first step we used 
$\partial(\widetilde{\lambda}_{\hh}) = \widetilde{\lambda}_{\hh} \circ \dd_{[1] \otimes}$, 
for the second step we used $\dd \circ (\Lambda_\pm)_\hh - (\Lambda_\pm)_\hh \circ \dd_{[1]} = \partial (\Lambda_\pm)_\hh = \id$ 
and $\dd \circ \Lambda_{\hh} = \Lambda_{\hh} \circ \dd_{[1]}$, 
for the third step we used 
$(-,-) \circ Q_{\otimes} = \dd_{\dR\,[m-1]} \circ (-,-)$ 
(the passage from $\dd$ to $Q$ is suppressed here because 
the natural transformation $\hocolim \to \colim$ 
from \eqref{eqn:hocolim-colim} has been suppressed from our notation) 
and for the fourth step we used Stokes' theorem. 
The last step is the definition of $\widetilde{\tau}_M$ in \eqref{eqn:tautildemap}
and of $\sigma_\Sigma$ in \eqref{eqn:sigmamap}. Recalling also the definition of $\tau_M$
in \eqref{eqn:taumap}, the previous computation shows that 
\begin{flalign}
\partial \lambda_{\hh} = \asym(\partial \widetilde{\lambda}_{\hh}) = 
\sigma_\Sigma \circ \Lambda_{\hh}^{\wedge 2} - \tau_M \quad. 
\end{flalign}
Let us summarize our findings.
\begin{theo}\label{th:Linfty-qiso}
Let $(F,Q)$ be a Green hyperbolic complex on $M$ 
endowed with a differential pairing $(-,-)$. 
Denote by $\Sigma\subseteq M$ a spacelike Cauchy surface of $M$. 
Then the retarded-minus-advanced quasi-isomorphism $\Lambda_{\hh} :
\FFF_{\hh\cc}[1]\to \FFF_{\hh\sc}$ from 
Theorem \ref{th:ret-minus-adv-qiso} is compatible with the Poisson 
structures $\tau_M$ and $\sigma_\Sigma$ from Propositions 
\ref{propo:covariant-Poisson} and \ref{propo:time-zero-Poisson} 
up to the homotopy $\lambda_{\hh}$ defined in \eqref{eqn:lambdah}, i.e.\  
$\sigma_\Sigma \circ \Lambda_{\hh}^{\wedge 2} - \tau_M = \partial \lambda_{\hh}$.
\end{theo}

\begin{rem}
Classically, the Poisson structure $\sigma_\Sigma$ is independent of
the choice of spacelike Cauchy surface $\Sigma\subseteq M$. 
The analogous conclusion in the present context may be formalized 
as follows. Let $\Sigma, \Sigma^\prime \subseteq M$ be two spacelike Cauchy 
surfaces of $M$. Then there exists a homotopy $\lambda_{\Sigma\Sigma^\prime}$ 
comparing the Poisson structures $\sigma_{\Sigma}$ and 
$\sigma_{\Sigma^\prime}$, namely such that 
$\sigma_\Sigma - \sigma_{\Sigma^\prime} = \partial \lambda_{\Sigma\Sigma^\prime}$. 
This statement is an immediate corollary of Theorem  
\ref{th:Linfty-qiso}. 
In fact, associated with $\Sigma$ and $\Sigma^\prime$ one has 
corresponding homotopies $\lambda_{\hh}$ and $\lambda_{\hh}^\prime$, 
both given by \eqref{eqn:lambdah}, such that 
$\sigma_\Sigma \circ \Lambda_{\hh}^{\wedge 2} - \tau_M = \partial \lambda_{\hh}$
and 
$\sigma_{\Sigma^\prime} \circ \Lambda_{\hh}^{\wedge 2} - \tau_M = \partial \lambda_{\hh}^\prime$. 
Comparing these equations, one gets 
$(\sigma_\Sigma - \sigma_{\Sigma^\prime}) \circ \Lambda_{\hh}^{\wedge 2} = \partial (\lambda_{\hh} - \lambda_{\hh}^\prime)$, 
therefore any choice of a quasi-inverse for $\Lambda_\hh$ 
allows one to construct a homotopy $\lambda_{\Sigma\Sigma^\prime}$ 
comparing $\sigma_{\Sigma}$ and $\sigma_{\Sigma^\prime}$. 
For instance, one choice of $\lambda_{\Sigma\Sigma^\prime}$ 
is obtained combining $\lambda_{\hh}$ and $\lambda_{\hh}^\prime$ 
with the quasi-inverse $\Theta$ from \eqref{eqn:Theta} 
and the homotopy $\Upsilon$ from \eqref{eqn:Upsilon} witnessing 
$\Lambda_\hh \circ \Theta \sim \id$. 
\end{rem}


\section{\label{sec:witness}Green's witnesses}
\subsection{\label{subsec:witness}Definition and results}
This section introduces the concept of a Green's witness, 
which consists of a collection of degree decreasing linear 
differential operators defined on a complex of linear differential 
operators satisfying suitable conditions. 
We shall show that finding a Green's witness is particularly useful 
because it entails that the underlying complex of linear differential operators 
is Green hyperbolic. In particular, we shall obtain several examples 
of Green hyperbolic complexes by constructing Green's witnesses. 
Furthermore, 
Green's witnesses provide explicit retarded and advanced Green's 
homotopies that are strictly natural, in contrast to the general case. 
This provides several simplifications compared to Section 
\ref{sec:Green-hyp-cpx}. Most notably, the retarded-minus-advanced 
quasi-isomorphism simplifies considerably, which is particularly 
interesting in view of applications to concrete examples 
of Green hyperbolic complexes admitting a Green's witness. 
Those simplifications shall be explored in detail in this section. 

\begin{defi}\label{defi:Green-witness}
A {\it Green's witness} $W = (W^n)_{n \in \bbZ}$ 
for a complex of linear differential operators $(F,Q)$ on $M$ consists of 
a collection of degree decreasing linear differential operators 
$W^n: \FFF^n \to \FFF^{n-1}$ such that, 
for all $n \in \bbZ$, the linear differential operators 
\begin{flalign}\label{eqn:P}
P^n := Q^{n-1}\, W^n + W^{n+1}\, Q^n: \FFF^n \longrightarrow \FFF^n
\end{flalign}
are Green hyperbolic (in the ordinary sense), 
see Section \ref{subsec:Green}. 
\end{defi}

\begin{rem}\label{rem:P}
Recall that, for a complex of linear differential operators $(F,Q)$, 
the cochain complex of its sections is denoted by 
$\FFF \in \Ch_\bbK$, whose degree $n \in \bbZ$ component
is the vector space $\FFF^n$ from \eqref{eqn:FFF} 
and whose differential is $Q$. A Green's witness $W$ for $(F,Q)$ 
is in particular a $(-1)$-cochain in the internal hom 
$[\FFF,\FFF] \in \Ch_\bbK$ whose differential $\partial W = P$ 
returns a $0$-coboundary whose components $P^n$ 
are ordinary Green hyperbolic linear differential operators. 
Since $Q$, $W$ and $P$ consist of linear differential operators, 
which in particular preserve supports, they naturally restrict 
to the cochain complexes of sections with restricted supports. 
\end{rem}

\begin{rem}\label{rem:CG}
Our concept of Green's witness has its roots in Hodge theory 
and, more generally, in the theory of elliptic complexes. 
A recent and closely related concept, which inspired us, 
is that of gauge fixing operator, appearing in the works of Costello and Gwilliam \cite{CG,CG2}.
We however would like to point out the following differences.
1.)~The analogs in \cite{CG,CG2} of the differential operators $P^n$
are assumed to be elliptic instead of Green hyperbolic. This is due to the fact that
Costello and Gwilliam consider Riemannian instead of Lorentzian manifolds $M$.
2.)~\cite{CG,CG2} include a self-adjointness condition and a square-zero condition. 
In Definition \ref{defi:self-adj-witness},
we shall also introduce a relaxation of the square-zero 
condition, see item~(i), and a certain self-adjointness 
condition, see item~(ii). Let us mention that both conditions 
are met by all examples we shall present.
\end{rem}

\begin{ex}\label{ex:deRham-witness}
One observes that the de~Rham codifferential $\delta_\dR$, 
defined out of the metric and orientation of $M$, is a Green's witness 
for the de~Rham complex $(\Lambda^\bullet M, \dd_\dR)$ 
from Example \ref{ex:deRham}. Indeed, the d'Alembert operator 
$\Box := \delta_{\dR}\, \dd_{\dR} + \dd_{\dR}\, \delta_{\dR}$ 
is normally hyperbolic, hence Green hyperbolic, in all degrees. 
For $m=3$, the $1$-shift of $\delta_\dR$ defines a Green's witness 
$W_{\CS}$ for the complex of linear differential operators 
$(F_{\CS},Q_{\CS})$ associated with linear Chern-Simons theory. 
\end{ex}

\begin{ex}\label{ex:single-diff-op-witness}
Recalling Example \ref{ex:single-diff-op}, 
when $P$ is a Green hyperbolic linear differential operator 
(e.g.\ the Klein-Gordon operator $P = \Box + m^2: C^\infty(M) \to C^\infty(M)$), 
one observes that $W_{(E,P)}$, 
consisting of a single non-vanishing component 
$W_{(E,P)}^1:=\id: \FFF_{(E,P)}^1 \to \FFF_{(E,P)}^0$, 
is a Green's witness for $(F_{(E,P)},Q_{(E,P)})$. 
\end{ex}

\begin{ex}\label{ex:Maxwell-p-forms-witness}
Recalling Example \ref{ex:Maxwell-p-forms}, 
for the complex of linear differential operators $(F_{\MW},Q_{\MW})$ 
associated with Maxwell $p$-forms on $M$
one constructs a Green's witness $W_{\MW}$ 
by defining its only non-vanishing components according to 
\begin{flalign}
W_{\MW}^n :=
\begin{cases}
\delta_{\dR} \quad, & n=-p+1,\ldots,0 \quad, \\
\id \quad, & n=1 \quad, \\ 
\dd_{\dR} \quad, & n=2,\ldots,p+1 \quad.
\end{cases}
\end{flalign}
Direct inspection shows that the only components of the linear differential operator $P_{\MW} = Q_{\MW} W_{\MW}+ W_{\MW}Q_{\MW}$ 
that do not necessarily vanish are d'Alembert operators 
$P_{\MW}^n = \delta_{\dR}\, \dd_{\dR} + \dd_{\dR}\, \delta_{\dR} = \Box$, for $n=-p,\ldots,p+1$. 
Since those are normally hyperbolic and hence Green hyperbolic, 
it follows that $W_{\MW}$ is a Green's witness 
for the complex of linear differential operators $(F_{\MW},Q_{\MW})$. 
\end{ex}

\begin{theo}\label{th:witness-Lambdapm}
Let $(F,Q)$ be a complex of linear differential operators on $M$
endowed with a Green's witness $W$. Then $(F,Q)$ 
is a Green hyperbolic complex. Furthermore, denote by 
$G_\pm^n: \FFF_{\pc/\fc}^n \to \FFF_{\pc/\fc}^n$ 
the (extended) retarded/advanced Green's operator associated with 
the Green hyperbolic operator $P^n$ from Definition 
\ref{defi:Green-witness}. Then the $(-1)$-cochain 
\begin{subequations}
\begin{flalign}
\Lambda_{\pm} \in \hom \big( \FFF_{J_M^{\pm}(-)},\FFF_{J_M^{\pm}(-)} \big)^{-1} \subseteq \map \big( \FFF_{J_M^{\pm}(-)},\FFF_{J_M^{\pm}(-)} \big)^{-1} \quad, 
\end{flalign}
defined, for all compact subsets $K \subseteq M$ and all $n \in \bbZ$, by 
\begin{flalign}
(\pr_K \Lambda_\pm)^n := W^n\, G_\pm^n \quad,
\end{flalign}
\end{subequations}
is a choice of retarded/advanced Green's homotopy for $(F,Q)$. 
\end{theo}
\begin{proof}
Recalling Proposition \ref{propo:Greenhyp-acyclic}, 
for the first part of the statement 
it suffices to show that, for all compact subsets $K \subseteq M$, 
the complex $\FFF_{J_M^\pm(K)} \in \Ch_\bbK$ 
of sections supported in the causal 
future/past of $K$ is acyclic. To achieve this goal, 
we define a $(-1)$-cochain 
\begin{subequations}\label{eqn:Lambdapm-ext}
\begin{flalign}
\Lambda_{\pm,K} \in [\FFF_{J_M^\pm(K)},\FFF_{J_M^\pm(K)}]^{-1}
\end{flalign}
degree-wise for all $n \in \bbZ$ by 
\begin{flalign}
\Lambda_{\pm,K}^n := W^n\, G_\pm^n \quad.
\end{flalign}
\end{subequations}
Direct inspection shows that $\partial \Lambda_{\pm,K} = \id$, 
hence $\FFF_{J_M^\pm(K)} \in \Ch_\bbK$ is acyclic for all compact subsets $K \subseteq M$, as claimed. 
More explicitly, one has 
\begin{flalign}
(\partial \Lambda_{\pm,K})^n = Q^{n-1}\, W^n\, G_\pm^n + W^{n+1}\, G_\pm^{n+1}\, Q^n = P^n\, G_\pm^n = \id \quad,
\end{flalign}
for all $n \in \bbZ$, 
where we used in the first step the definition of $\partial$ from 
\eqref{eqn:internalhom}, in the second step the identity 
\begin{flalign}\label{eqn:QG=GQ}
G_\pm^{n+1}\, Q^n  = Q^n\, G_\pm^n\quad, 
\end{flalign}
which follows from $Q^n\, P^n = P^{n+1}\, Q^n$, see Remark \ref{rem:P}, 
and in the last step the properties of a retarded/advanced Green's operator, 
see Section \ref{subsec:Green}.
\sk

For the second part of the statement there remains to show 
that the $(-1)$-cochains $\Lambda_{\pm,K} \in [\FFF_{J_M^\pm(K)},\FFF_{J_M^\pm(K)}]^{-1}$, 
for all compact subsets $K \subseteq M$, 
assemble to form a $(-1)$-cochain $\Lambda_{\pm}$ 
in the enriched hom according to $\pr_K \Lambda_\pm = \Lambda_{\pm,K}$. 
Since strict naturality follows from the support properties 
of the retarded/advanced Green's operator, 
see Section \ref{subsec:Green}, the claim follows. 
\end{proof}

\begin{rem}
Note that Green's witnesses provide particularly simple retarded 
and advanced Green's homotopies. Indeed, those are strictly natural 
as they lie in the enriched hom. This should be compared 
to the case of generic retarded and advanced 
Green's homotopies, which instead lie in the mapping complex 
and hence are not necessarily strictly natural. 
The latter consist of a much more intricate hierarchy 
of data expressing naturality in a homotopy coherent fashion. 
While this is not an issue from an abstract point of view, 
for concrete applications the retarded and advanced Green's homotopies 
from Theorem \ref{th:witness-Lambdapm} are certainly convenient. 
For instance, this fact will be exploited in \cite{BeniniMusanteSchenkel}
where we construct examples of strict algebraic quantum field theories and of 
strict time-orderable prefactorization algebras from $\Loc_m$-natural Green hyperbolic complexes that
are endowed with a natural Green's witness.
\end{rem}

\begin{rem}\label{rem:tildeLambda}
Recall that the spaces of retarded and advanced Green's homotopies 
are either empty or contractible, see Proposition \ref{propo:Greenhyp-unique}. 
Informally, this means that, when they exist, 
retarded and advanced Green's homotopies are unique 
up to higher homotopies, which are themselves unique up 
to even higher homotopies, and so on. 
Let us illustrate this phenomenon concretely in the case of a complex 
of linear differential operators $(F,Q)$ endowed with a Green's witness $W$. 
In this setting, along with the retarded and advanced Green's homotopies 
$\Lambda_\pm$ from Theorem \ref{th:witness-Lambdapm}, 
there exist other choices of retarded and advanced 
Green's homotopies $\widetilde{\Lambda}_\pm$. Consider, for instance, 
\begin{subequations}\label{eqn:tildeLambdapm}
\begin{flalign}
\widetilde{\Lambda}_\pm \in \hom \big( \FFF_{J_M^\pm(-)},\FFF_{J_M^\pm(-)} \big)^{-1} \subseteq \map \big( \FFF_{J_M^\pm(-)},\FFF_{J_M^\pm(-)} \big)^{-1} \quad, 
\end{flalign}
which is defined, for all compact subsets $K \subseteq M$ and $n \in \bbZ$, by 
\begin{flalign}
(\pr_K \widetilde{\Lambda}_\pm)^n := G_\pm^{n-1}\,  W^n\quad.
\end{flalign}
\end{subequations}
(Notice the order reversal compared 
to Theorem \ref{th:witness-Lambdapm}.) 
Strict naturality of $\widetilde{\Lambda}_\pm$ follows by the same 
argument as in the proof of Theorem \ref{th:witness-Lambdapm}. 
To confirm that $\widetilde{\Lambda}_\pm$ is indeed a retarded/advanced 
Green's homotopy, recall \eqref{eqn:QG=GQ} and \eqref{eqn:P} and compute 
\begin{flalign}
(\partial \widetilde{\Lambda}_\pm)^n = Q^{n-1}\, G_\pm^{n-1}\, W^n + G_\pm^n\, W^{n+1}\, Q^n = G_\pm^n\, P^n = \id \quad.
\end{flalign}
Let us exhibit an explicit (higher) homotopy $\lambda_\pm$ that 
relates $\Lambda_\pm$ 
and $\widetilde{\Lambda}_\pm$. Consider 
\begin{subequations}
\begin{flalign}
\lambda_\pm \in \hom \big( \FFF_{J_M^\pm(-)},\FFF_{J_M^\pm(-)} \big)^{-2} \subseteq \map \big( \FFF_{J_M^\pm(-)},\FFF_{J_M^\pm(-)} \big)^{-2} \quad, 
\end{flalign}
which is defined, for all compact subsets $K \subseteq M$ and $n \in \bbZ$, by 
\begin{flalign}
(\pr_K \lambda_\pm)^n := W^{n-1}\, G_\pm^{n-1}\, G_\pm^{n-1}\, W^n \quad.
\end{flalign}
\end{subequations}
Strict naturality of $\lambda_\pm$ follows as usual from the support 
properties of the retarded/advanced Green's operator, 
see Section \ref{subsec:Green}. 
A straightforward computation shows that 
\begin{flalign}
(\partial \lambda_\pm)^n & = Q^{n-2}\, W^{n-1}\, G_\pm^{n-1}\, G_\pm^{n-1}\, W^n -W^{n}\, G_\pm^{n}\, G_\pm^{n}\, W^{n+1}\, Q^n \nn \\ 
& = (P^{n-1} - W^{n}\, Q^{n-1})\, G_\pm^{n-1}\, G_\pm^{n-1}\, W^n - W^{n}\, G_\pm^{n}\, G_\pm^{n}\, (P^n - Q^{n-1}\, W^{n}) \nn \\
& = G_\pm^{n-1}\, W^n - W^{n}\, Q^{n-1}\, G_\pm^{n-1}\, G_\pm^{n-1}\, W^n - W^{n}\, G_\pm^{n} + W^{n}\, G_\pm^{n}\, G_\pm^{n}\, Q^{n-1}\, W^{n} \nn \\
& = \widetilde{\Lambda}_\pm^n - \Lambda_\pm^n \quad, 
\end{flalign}
where we used \eqref{eqn:P} in the second step, 
the properties of retarded/advanced Green's operators, 
see Section \ref{subsec:Green}, in the third step 
and \eqref{eqn:QG=GQ} in the last step. 
It is not difficult to come up with other choices 
$\widetilde{\lambda}_\pm \in \map(\FFF_{J_M^\pm(-)},\FFF_{J_M^\pm(-)})^{-2}$ 
relating $\Lambda_\pm$ and $\widetilde{\Lambda}_\pm$. 
Proposition \ref{propo:Greenhyp-unique} ensures that any two choices 
coincide up to a (even higher) homotopy, and so on. 
\end{rem}

\begin{ex}\label{ex:Lambdapm-deRham}
Let $(\Lambda^\bullet M,\dd_{\dR})$ be the de~Rham complex 
from Example \ref{ex:deRham} and recall 
the Green's witness $\delta_{\dR}$ from Example \ref{ex:deRham-witness}. 
Then Theorem \ref{th:witness-Lambdapm} states that 
$(\Lambda^\bullet M,\dd_{\dR})$ is a Green hyperbolic complex 
and provides a specific choice of retarded/advanced Green's 
homotopy $\Lambda_{\dR\,\pm}$, whose only non-vanishing components are 
\begin{flalign}
(\pr_K\, \Lambda_{\dR\,\pm})^n := \delta_{\dR}\, G_{\Box\,\pm} \quad,
\end{flalign}
for all compact subsets $K \subseteq M$ and $n = 1, \ldots, m$, 
where $G_{\Box\,\pm}$ denotes the retarded/advanced Green's operator 
for the d'Alembert operator $\Box$ 
acting on $k$-forms, $k=0,\ldots,m$. 
Since $\dd_{\dR}\, G_{\Box\,\pm} = G_{\Box\,\pm}\, \dd_{\dR}$ 
and $\delta_{\dR}\, G_{\Box\,\pm} = G_{\Box\,\pm}\, \delta_{\dR}$,
we observe that the retarded/advanced Green's homotopy 
$\widetilde{\Lambda}_{\dR\,\pm}$ from Remark \ref{rem:tildeLambda} 
coincides with $\Lambda_{\dR\,\pm}$. 
In the same fashion, for $m=3$, one concludes that the complex of 
linear differential operators $(F_{\CS},Q_{\CS})$ associated with 
linear Chern-Simons theory is Green hyperbolic and 
a specific choice of retarded/advanced Green's homotopy 
$\Lambda_{\CS\,\pm}$ is just the $1$-shift of $\Lambda_{\dR\,\pm}$. 
\end{ex}

\begin{ex}\label{ex:Lambdapm-single-diff-op}
Let $(F_{(E,P)},Q_{(E,P)})$ be the complex of linear differential operators 
from Example \ref{ex:single-diff-op} and recall the Green's witness 
$W_{(E,P)}$ from Example \ref{ex:single-diff-op-witness}. 
Then Theorem \ref{th:witness-Lambdapm} states that 
$(F_{(E,P)},Q_{(E,P)})$ is a Green hyperbolic complex 
and provides a specific choice of retarded/advanced Green's 
homotopy $\Lambda_{(E,P)\,\pm}$, which in this case coincides with 
the (unique) retarded/advanced Green's operator $G_\pm$ for $P$, 
see also Example \ref{ex:Lambdapm=Gpm}. 
We observe that also in this case the a priori different 
(yet equivalent) retarded/advanced Green's homotopy 
$\widetilde{\Lambda}_{(E,P)\,\pm}$ from Remark \ref{rem:tildeLambda} 
actually coincides with $\Lambda_{(E,P)\,\pm}$. 
\end{ex}

\begin{ex}\label{ex:Lambdapm-Maxwell-p-forms}
Let $(F_{\MW},Q_{\MW})$ be the complex of linear differential operators 
from Example \ref{ex:Maxwell-p-forms} associated with Maxwell $p$-forms 
and recall the Green's witness $W_{\MW}$ from 
Example \ref{ex:Maxwell-p-forms-witness}. 
Then Theorem \ref{th:witness-Lambdapm} states that 
$(F_{\MW},Q_{\MW})$ is a Green hyperbolic complex 
and provides a specific choice of retarded/advanced Green's 
homotopy $\Lambda_{\MW\,\pm}$, whose only non-vanishing components are 
\begin{flalign}
(\pr_K \Lambda_{\MW\,\pm})^n :=
\begin{cases}
\delta_{\dR}\, G_{\Box\,\pm} \quad, & n=-p+1,\ldots,0\quad, \\
G_{\Box\,\pm} \quad, & n=1\quad, \\ 
\dd_{\dR}\, G_{\Box\,\pm} \quad, & n=2,\ldots,p+1\quad,
\end{cases}
\end{flalign}
for all compact subsets $K \subseteq M$, 
where $G_{\Box\,\pm}$ denotes the retarded/advanced Green's operator 
for the d'Alembert operator $\Box$. 
Once again we observe that also in this example the 
retarded/advanced Green's homotopy $\widetilde{\Lambda}_{\MW\,\pm}$ 
from Remark \ref{rem:tildeLambda} coincides with $\Lambda_{\MW\,\pm}$. 
This behavior is not accidental, but a consequence 
of the fact that these examples admit a so-called 
{\it formally self-adjoint} Green's witness 
in the sense of Definition \ref{defi:self-adj-witness}, 
see Remark \ref{rem:self-adj-witness}. 
\end{ex}

\begin{rem}\label{rem:Lambda}
A Green's witness $W$ simplifies considerably 
the construction of the retarded-minus-advanced cochain map. 
In fact, Theorem \ref{th:witness-Lambdapm} provides specific choices 
of retarded and advanced Green's homotopies $\Lambda_\pm$ 
that are strictly natural, in contrast to the general case.
This entails that the construction of the retarded-minus-advanced 
cochain map $\Lambda_{\hh}$ from Definition \ref{defi:ret-minus-adv} 
``descends to ordinary colimits''. 
More explicitly, with this specific choice of retarded and advanced 
Green's homotopies, $\Lambda$ from \eqref{eqn:ret-minus-adv} actually 
lies in the enriched hom subcomplex of the mapping complex 
and hence it defines a natural transformation 
$\Lambda: \FFF_{(-)}[1] \to \FFF_{J_M(-)}$ in $\Ch_\bbK^\cc$. 
Taking the ordinary colimit, one obtains a cochain map
\begin{flalign}\label{eqn:ret-minus-adv-qiso-witness}
\Lambda := \colim(\Lambda) : \FFF_{\cc}[1] \longrightarrow \FFF_{\sc}
\end{flalign}
in $\Ch_\bbK$, which we denote with abuse of notation by the same symbol $\Lambda$.
Recalling also that $\cc$ is a filtered category, 
it follows that the natural quasi-isomorphism \eqref{eqn:hocolim-colim} 
comparing homotopy and ordinary colimits forms the commutative diagram
\begin{flalign}\label{eqn:ret-minus-adv-witness}
\xymatrix@C=5em{
\FFF_{\hh\cc}[1] \ar[r]^-{\Lambda_{\hh}} \ar[d]_-{\sim} & \FFF_{\hh\sc} \ar[d]^-{\sim} \\ 
\FFF_{\cc}[1] \ar[r]_-{\Lambda} & \FFF_{\sc}
}
\end{flalign}
in $\Ch_\bbK$. The latter makes precise the previous claim that 
the retarded-minus-advanced cochain map descends to ordinary colimits.
\end{rem} 

\begin{rem}\label{rem:Lambda-qiso}
In the presence of a Green's witness $W$ for $(F,Q)$,
the statement and proof of Theorem \ref{th:ret-minus-adv-qiso} can be simplified 
by passing from homotopy to ordinary colimits. Indeed, using
the commutative diagram \eqref{eqn:ret-minus-adv-witness}, one finds that 
$\Lambda_{\hh}$ is a quasi-isomorphism if and only if $\Lambda$ is one. 
(Recall that, if two out of the three cochain maps $f,g,gf$ 
are quasi-isomorphisms, then all three are such. 
This is readily seen by passing to cohomology and observing that 
if two out of the three graded linear maps $H(f),H(g),H(gf)=H(g)H(f)$ 
are isomorphisms, then all three are such.) 
For $\Lambda$ one finds simpler (compared to Section 
\ref{subsec:ret-adv-map}) explicit expressions of a quasi-inverse 
\begin{flalign}
\Theta: \FFF_{\sc} \longrightarrow \FFF_{\cc}[1] \quad, \qquad \varphi \longmapsto \pm \big( Q (\chi_\pm \varphi) - \chi_\pm Q \varphi \big) \quad,
\end{flalign}
in $\Ch_\bbK$ for $\Lambda$, of a homotopy 
\begin{flalign}\label{eqn:Xi-witness}
\Xi \in \big[ \FFF_{\cc}[1], \FFF_{\cc}[1] \big]^{-1} \quad, \qquad \varphi \longmapsto - \chi_- \Lambda_+ \varphi - \chi_+ \Lambda_- \varphi \quad,
\end{flalign}
witnessing $\Theta \circ \Lambda \sim \id$ and of a homotopy 
\begin{flalign}
\Upsilon \in \big[ \FFF_{\sc}, \FFF_{\sc} \big]^{-1} \quad, \qquad \varphi \longmapsto \Lambda_+ \chi_+ \varphi + \Lambda_- \chi_- \varphi \quad,
\end{flalign} 
witnessing $\Lambda \circ \Theta \sim \id$. 
(The seeming sign discrepancy between \eqref{eqn:Xi-witness} and 
\eqref{eqn:Xi} is explained by the fact that \eqref{eqn:Xi} involves 
pulling the shifts out of the internal hom, which contributes a sign 
$(-1)^n$ in degree $n$ and hence $-1$ in degree $-1$.)
Let us stress that this simplified version of Theorem \ref{th:ret-minus-adv-qiso} 
is available for Examples \ref{ex:Lambdapm-deRham}, 
\ref{ex:Lambdapm-single-diff-op} and 
\ref{ex:Lambdapm-Maxwell-p-forms}. 
\end{rem}

\begin{rem}\label{rem:exact=qiso}
Theorem \ref{th:ret-minus-adv-qiso} specializes to the usual exact sequence 
\eqref{eqn:PGP-exact-seq} when applied to the Green hyperbolic complex 
$(F_{(E,P)},Q_{(E,P)})$ associated with a 
Green hyperbolic linear differential operator $P$ acting on sections of 
a vector bundle $E \to M$, see Examples \ref{ex:single-diff-op} and \ref{ex:Lambdapm=Gpm}. 
Indeed, recalling also Example \ref{ex:Lambdapm-single-diff-op}, 
we observe that in this case 
the only non-vanishing component of the retarded-minus-advanced cochain map
$\Lambda_{(E,P)} :  \FFF_{(E,P)\, \cc}[1] \to \FFF_{(E,P)\,\sc}$ from \eqref{eqn:ret-minus-adv-qiso-witness} 
coincides with the retarded-minus-advanced propagator 
$G = G_+ - G_-$ associated with $P$. 
It follows that the cone complex 
\begin{flalign}
\cone\big(\Lambda_{(E,P)}\big) = \Big( \xymatrix@C=1.4em{\cdots \ar[r] & \overset{(-2)}{0} \ar[r] & \overset{(-1)}{\Gamma_{\cc}(E)} \ar[r]^-{-P} & \overset{(0)}{\Gamma_{\cc}(E)} \ar[r]^-{G} & \overset{(1)}{\Gamma_{\sc}(E)} \ar[r]^-{-P} & \overset{(2)}{\Gamma_{\sc}(E)} \ar[r] & \overset{(3)}{0} \ar[r] & \cdots} \Big)
\end{flalign}
is manifestly isomorphic to \eqref{eqn:PGP-exact-seq} 
regarded as a cochain complex concentrated between degrees $-1$ and $2$. 
Recalling that a cochain map is a quasi-isomorphism if and only if 
its cone complex is acyclic \cite[Cor.~1.5.4]{Weibel}, we conclude that 
\eqref{eqn:PGP-exact-seq} being exact is equivalent to 
$\Lambda_{(E,P)}$ being a quasi-isomorphism. 
This explains how Theorem \ref{th:ret-minus-adv-qiso} 
generalizes the well-known exact sequence \eqref{eqn:PGP-exact-seq} 
associated with a Green hyperbolic linear differential operator. 
We would like to mention that also 
\cite[Th.~5.2.3]{Lupo} proves exactness of the sequence 
\eqref{eqn:PGP-exact-seq} by exhibiting a witnessing contracting 
homotopy. The latter is closely related to our quasi-inverse and witnessing homotopies 
from Remark \ref{rem:Lambda-qiso} specialized to 
the complex of linear differential operators 
$(F_{(E,P)},Q_{(E,P)})$ endowed with the Green's witness $W_{(E,P)}$
from Example \ref{ex:single-diff-op-witness}. 
\end{rem}

\subsection{\label{subsec:self-adj}Formally self-adjoint Green's witnesses}
This section introduces the concept of 
a formally self-adjoint Green's witness, 
i.e.\ a Green's witness that is compatible with a given 
differential pairing in a way that is partially reminiscent 
of formal self-adjointness in the ordinary sense. 
Formally self-adjoint Green's witnesses are such that 
the specific choice of retarded and advanced Green's homotopies 
from Theorem \ref{th:witness-Lambdapm} simplifies 
the construction and the comparison of the 
Poisson structures $\tau_M$ and $\sigma_\Sigma$ 
from Propositions \ref{propo:covariant-Poisson} 
and \ref{propo:time-zero-Poisson}. 
It will be shown at the end 
of this section that all examples considered here admit 
a formally self-adjoint Green's witness 
and hence the associated Poisson structures can be constructed 
and compared in the simpler way illustrated below. 

\begin{defi}\label{defi:self-adj-witness}
Let $(F,Q)$ be complex of linear differential operators 
on $M$ endowed with a differential pairing $(-,-)$.
A Green's witness $W$ for $(F,Q)$ is called 
{\it formally self-adjoint} when the compatibility conditions 
listed below are met: 
\begin{enumerate}[label=(\roman*)]
\item $Q\, W\, W = W\, W\, Q$, 
\item $\int_M (W \varphi_1, \varphi_2) = (-1)^{|\varphi_1|} \int_M (\varphi_1, W \varphi_2)$, 
\end{enumerate}
for all homogeneous sections $\varphi_1, \varphi_2 \in \FFF$ 
with compact overlapping support. 
\end{defi}

\begin{rem}\label{rem:self-adj-witness}
Recalling also Definition \ref{defi:Green-witness}, 
let us emphasize some immediate consequences of 
Definition \ref{defi:self-adj-witness}. 
\begin{enumerate}[label=(\roman*)]
\item It follows that $P\, W = W\, P$, 
therefore one also has $G_\pm\, W = W\, G_\pm$, 
where $G_\pm$ denotes the retarded/advanced Green's operator 
associated with the Green hyperbolic linear differential operator 
$P := Q\, W + W\, Q$. 
In particular, the specific choice $\Lambda_\pm$ 
of retarded and advanced Green's homotopies from Theorem 
\ref{th:witness-Lambdapm} coincides with the one 
$\widetilde{\Lambda}_\pm$ from Remark \ref{rem:tildeLambda}. 
As a consequence, the retarded-minus-advanced quasi-isomorphism 
$\Lambda : \FFF_{\cc}[1] \to \FFF_{\sc}$ in $\Ch_\bbR$ 
from \eqref{eqn:ret-minus-adv-qiso-witness} can be equivalently expressed 
as $G\, W = \Lambda = W\, G$, where $G := G_+ - G_-$ denotes 
the retarded-minus-advanced propagator associated with $P$. 
 
\item One finds that $P$ is formally self-adjoint, namely 
$\int_M (P \varphi_1, \varphi_2) = \int_M (\varphi_1, P \varphi_2)$ 
for all homogeneous sections $\varphi_1, \varphi_2 \in \FFF$ 
with compact overlapping support. From this fact it follows that, 
for all homogeneous sections $\varphi_1, \varphi_2 \in \FFF_\cc$ 
with compact support, 
$\int_M (G_\pm \varphi_1, \varphi_2) = \int_M (\varphi_1, G_\mp \varphi_2)$, 
and hence also 
$\int_M (G \varphi_1, \varphi_2) = - \int_M (\varphi_1, G \varphi_2)$. 
\end{enumerate}
These observations shall be systematically used in the proof 
of Proposition \ref{propo:covariant-Poisson-self-adj-witness}. 
\end{rem}

Recall that a Green's witness $W$ leads to a simplified quasi-isomorphism 
$\Lambda: \FFF_{\cc}[1] \to \FFF_{\sc}$  
between ordinary (as opposed to homotopy) colimits, 
see \eqref{eqn:ret-minus-adv-qiso-witness}  and \eqref{eqn:ret-minus-adv-witness}. 
When $W$ is formally self-adjoint, 
one realizes that also Propositions \ref{propo:covariant-Poisson}, 
\ref{propo:time-zero-Poisson} and Theorem \ref{th:Linfty-qiso} 
simplify considerably. 
\begin{propo}\label{propo:covariant-Poisson-self-adj-witness}
Let $(F,Q)$ be a complex of differential operators on $M$ 
endowed with a differential pairing $(-,-)$ 
and a formally self-adjoint Green's witness $W$. 
Consider the retarded-minus-advanced cochain map 
$\Lambda : \FFF_{\cc}[1] \to \FFF_{\sc}$ in $\Ch_\bbR$ 
from Remarks \ref{rem:Lambda} and \ref{rem:Lambda-qiso}. Then the composition 
\begin{flalign}\label{eqn:tau-witness}
\xymatrix@C=4em{
\tau_M: \FFF_{\cc}[1]^{\otimes 2} \ar[r]^-{\id \otimes \Lambda} & \FFF_{\cc}[1] \otimes \FFF_{\sc} \subseteq \FFF_{\cc}[1] \otimes \FFF \ar[r]^-{\ev_M} & \bbR
}
\end{flalign}
in $\Ch_\bbR$ is graded anti-symmetric 
and hence it descends to a Poisson structure $\tau_M : \FFF_{\cc}[1]^{\wedge 2}\to \bbR$
on $\FFF_{\cc}[1]$. (Recall the definition of $\ev_M$ from \eqref{eqn:ev-M}.) 
\end{propo}
\begin{proof}
Denoting the symmetric braiding on $\Ch_\bbR$ 
by $\gamma$, one has to check that 
$\tau_M \circ \gamma = - \tau_M$. Since $\bbR \in \Ch_\bbR$ 
is concentrated in degree $0$, it suffices to check this equality
upon evaluation on all sections $\varphi_1 \in \FFF_\cc[1]^q$, 
$\varphi_2 \in \FFF_\cc[1]^{-q}$ and for all $q \in \bbZ$. 
Recalling Remark \ref{rem:self-adj-witness}, one has 
\begin{flalign}
\tau_M \gamma (\varphi_1 \otimes \varphi_2) &= (-1)^{q} \int_M (\varphi_2, W G \varphi_1) = (-1)^{q+1} \int_M (W G \varphi_1, \varphi_2) = -\int_M (\varphi_1, G W \varphi_2) \nn \\
&= -\tau_M (\varphi_1 \otimes \varphi_2) \quad.
\end{flalign}
The first step follows using the braiding $\gamma$ and 
$\Lambda = W\, G$, see item~(i) of Remark \ref{rem:self-adj-witness}, 
the second step follows from graded anti-symmetry of $(-,-)$, 
the third step follows combining item~(ii) 
of Definition \ref{defi:self-adj-witness} 
and item~(ii) of Remark \ref{rem:self-adj-witness} and 
the last step follows from $\Lambda = G\, W$, 
see item~(i) of Remark \ref{rem:self-adj-witness}.
\end{proof}
\begin{rem}\label{rem:covariant-Poisson-self-adj-witness}
With the choice of retarded and advanced Green's homotopies 
$\Lambda_\pm$ from Theorem \ref{th:witness-Lambdapm} 
and up to the quasi-isomorphism $\FFF_{\hh\cc}[1] \simeq \FFF_\cc[1]$, 
the cochain map $\widetilde{\tau}_M$ from \eqref{eqn:tautildemap} 
coincides with the Poisson structure $\tau_M$ 
from \eqref{eqn:tau-witness}. 
Indeed, the same calculation as in the proof of Proposition 
\ref{propo:covariant-Poisson-self-adj-witness} 
shows that $\widetilde{\tau}_M$ is already graded anti-symmetric, 
making the graded anti-symmetrization in \eqref{eqn:taumap} 
superfluous, as anticipated by Remark \ref{rem:taumap}. 
Furthermore,  calculations similar to the one in the proof 
of Proposition \ref{propo:covariant-Poisson-self-adj-witness} show that,
in the present setting,
all Poisson structures from Proposition \ref{propo:covariant-Poisson} 
coincide $\tau_M^+ = \tau_M^- = \tau_M$, 
as anticipated by Remark \ref{rem:tau}. 
\end{rem}

The next simplified version of Proposition \ref{propo:time-zero-Poisson} 
is a straightforward consequence of the graded anti-symmetry 
of the differential pairing $(-,-)$. 
\begin{propo}\label{propo:time-zero-Poisson-ordinary-colim}
Let $(F,Q)$ be a complex of linear differential operators on $M$ 
endowed with a differential pairing $(-,-)$. 
Denote by $\Sigma\subseteq M$ a spacelike Cauchy surface of $M$. 
Then the composition 
\begin{flalign}
\xymatrix@C=5em{
\sigma_\Sigma: \FFF_{\sc}^{\otimes 2} \ar[r]^-{\id \otimes \subseteq} & \FFF_{\sc} \otimes \FFF \ar[r]^-{(-1)^{m-1} \ev_\Sigma} & \bbR 
}
\end{flalign}
in $\Ch_\bbR$ is graded anti-symmetric and hence  it descends to a Poisson structure
$\sigma_\Sigma: \FFF_{\sc}^{\wedge 2} \to \bbR $ on $\FFF_{\sc}$.
(Recall the definition of $\ev_\Sigma$ from \eqref{eqn:ev-Sigma}.)
\end{propo}

Even though $\Lambda : \FFF_{\cc}[1] \to \FFF_{\sc}$ 
is {\it not} compatible with the simpler Poisson structures 
from Propositions \ref{propo:covariant-Poisson-self-adj-witness} 
and \ref{propo:time-zero-Poisson-ordinary-colim}, 
there is a homotopy $\lambda$ controlling this failure, 
which is a simpler counterpart of the homotopy $\lambda_{\hh}$ 
from \eqref{eqn:lambdah}. 
This determines a simplification of Theorem \ref{th:Linfty-qiso}. 
More explicitly, the relevant homotopy 
\begin{subequations}\label{eqn:lambda}
\begin{flalign}
\lambda := \asym(\widetilde{\lambda}) \in \big[ \FFF_{\cc}[1]^{\wedge 2}, \bbR \big]^{-1}
\end{flalign}
is the graded anti-symmetrization of the $(-1)$-cochain 
\begin{flalign}
\widetilde{\lambda} \in \big[ \FFF_{\cc}[1]^{\otimes 2}, \bbR \big]^{-1}
\end{flalign}
defined, for all homogeneous sections 
$\varphi_1, \varphi_2 \in \FFF_{\cc}[1]$, by 
\begin{flalign}
\widetilde{\lambda} (\varphi_1\otimes \varphi_2) := \int_{\Sigma^+} ( \Lambda_- \varphi_1, \Lambda \varphi_2 ) + \int_{\Sigma^-} ( \Lambda_+ \varphi_1, \Lambda \varphi_2 ) \quad, 
\end{flalign}
\end{subequations}
where $\Sigma^\pm := J_M^\pm(\Sigma)$. 
A similar, yet slightly simpler, computation 
as \eqref{eqn:lambdah-computation} leads to the result below. 
\begin{theo}\label{th:Linfty-qiso-self-adj-witness}
Let $(F,Q)$ be a complex of linear differential operators on $M$ 
endowed with a differential pairing $(-,-)$ 
and a formally self-adjoint Green's witness $W$. 
Denote by $\Sigma\subseteq M$ a spacelike Cauchy surface of $M$. 
Then the retarded-minus-advanced quasi-isomorphism $\Lambda$ 
from Remarks \ref{rem:Lambda} and \ref{rem:Lambda-qiso} is compatible with 
the Poisson structures $\tau_M$ and $\sigma_\Sigma$ from 
Propositions \ref{propo:covariant-Poisson-self-adj-witness} and 
\ref{propo:time-zero-Poisson-ordinary-colim} up to the homotopy 
$\lambda$ defined in \eqref{eqn:lambda}, i.e.\ 
$\sigma_\Sigma \circ \Lambda^{\wedge 2} - \tau_M = \partial \lambda$. 
\end{theo}

\begin{ex}\label{ex:CS-self-adj-witness}
For $m = 3$, recall from Examples \ref{ex:deRham}, 
\ref{ex:CS-diff-pairing} and \ref{ex:deRham-witness}
the complex of linear differential operators $(F_{\CS},Q_{\CS})$ 
associated with linear Chern-Simons theory, 
the differential pairing $(-,-)_{\CS}$ and the Green's witness $W_{\CS}$. 
It follows from the standard properties 
of the de Rham codifferential that $W_{\CS} = \delta_{\dR\, [1]}$ 
is a formally self-adjoint Green's witness. 
\end{ex}

\begin{ex}\label{ex:single-diff-op-self-adj-witness}
Recall from Examples \ref{ex:single-diff-op}, 
\ref{ex:single-diff-op-diff-pairing} and \ref{ex:single-diff-op-witness}
the complex of linear differential operators $(F_{(E,P)},Q_{(E,P)})$, 
the differential pairing $(-,-)_{(E,P)}$ and the Green's witness $W_{(E,P)}$. 
In this example $W_{(E,P)}$ is (trivially) a formally self-adjoint 
Green's witness and the Poisson structure from Proposition 
\ref{propo:covariant-Poisson-self-adj-witness} agrees with 
the standard Poisson structure given by the retarded-minus-advanced propagator, see \eqref{eqn:standardtau}. 
\end{ex}

\begin{ex}\label{ex:MW-self-adj-witness}
Recall from Examples \ref{ex:Maxwell-p-forms}, 
\ref{ex:MW-diff-pairing} and \ref{ex:Maxwell-p-forms-witness}
the complex of linear differential operators $(F_{\MW},Q_{\MW})$ 
associated with Maxwell $p$-forms, 
the differential pairing $(-,-)_{\MW}$
and the Green's witness $W_{\MW}$. It follows 
from the standard properties of the de Rham (co)differential 
that $W_{\MW}$ is a formally self-adjoint Green's witness. 
For $p = 1$, the Poisson complex from Proposition 
\ref{propo:covariant-Poisson-self-adj-witness} agrees 
with the one for linear Yang-Mills theory from \cite{linYM} 
and for $p$ arbitrary with the one from \cite{hrep}. 
\end{ex}


\section*{Acknowledgments}
We would like to thank the reviewers for their suggestions 
and for pointing out valuable references to the literature. 
The work of M.B.\ and G.M.\ is fostered by 
the National Group of Mathematical Physics (GNFM-INdAM (IT)). 
G.M.\ is supported by a PhD scholarship of the University of Genova (IT).
A.S.\ gratefully acknowledges the support of 
the Royal Society (UK) through a Royal Society University 
Research Fellowship (URF\textbackslash R\textbackslash 211015)
and Enhancement Awards (RGF\textbackslash EA\textbackslash 180270, 
RGF\textbackslash EA\textbackslash 201051 and RF\textbackslash ERE\textbackslash 210053).

\section*{Data availability statement}
All data generated or analyzed during this study are contained in this document.

\section*{Conflict of interest statement}
The authors have no conflict of interest to declare that are relevant to the content of this article.



\end{document}